%% file: paper.tex
\documentclass[10pt, conference]{IEEEtran}
\input{config}
\begin{document}
\bstctlcite{MyBSTcontrol}
\title{Intrusion Tolerance for Networked Systems \\through Two-Level Feedback Control}
\author{\IEEEauthorblockN{Kim Hammar\IEEEauthorrefmark{2} and Rolf Stadler\IEEEauthorrefmark{2}}
  \IEEEauthorblockA{\IEEEauthorrefmark{2}
    KTH Royal Institute of Technology, Sweden\\
    Email: \{kimham, stadler\}@kth.se
  }
  \today
}
\maketitle
\begin{abstract}
We formulate intrusion tolerance for a system with service replicas as a two-level optimal control problem. On the local level node controllers perform intrusion recovery, and on the global level a system controller manages the replication factor. The local and global control problems can be formulated as classical problems in operations research, namely, the machine replacement problem and the inventory replenishment problem. Based on this formulation, we design \tolerancee, a novel control architecture for intrusion-tolerant systems. We prove that the optimal control strategies on both levels have threshold structure and design efficient algorithms for computing them. We implement and evaluate \tolerancee in an emulation environment where we run $10$ types of network intrusions. The results show that \tolerancee can improve service availability and reduce operational cost compared with state-of-the-art intrusion-tolerant systems.
\end{abstract}

\section{Introduction}
As our reliance on online services is growing, there is an increasing demand for reliable systems that provide correct service without disruption. Traditionally, the main cause of disruption in networked systems has been hardware failures and power outages \cite{fault_tolerance_orig,cristian1991understanding}. While tolerance against these types of failures is important, a growing source of disruptions is network intrusion \cite{cs_trend}.

Intrusions are fundamentally different from hardware failures as an attacker can behave arbitrarily, i.e., Byzantine, which leads to unanticipated failure behavior. Due to the high costs of such failures and the fact that it is impractical to prevent all intrusions, the ability to \textit{tolerate} intrusions becomes necessary \cite{resilience_nature}. This ability is particularly important for safety-critical applications, e.g., real-time control systems \cite{scada_intrusion_tolerance_practice,wolff_space_byzantine,614114,Lala1986,1455383,259424}, control-planes for software defined networks \cite{sdn_byzantine}, \scada systems \cite{8416480,6576306}, and e-commerce applications \cite{9713988}.

We call a system \textit{intrusion-tolerant} if it provides \textit{correct service} while intrusions occur \cite{deswarte1991intrusion,10.1007/3-540-45177-3_1,1290202}. The common approach to build intrusion-tolerant systems is to replicate the system over a set of \textit{nodes}, each of which can respond to service requests. Through such redundancy, \textit{compromised} and \textit{crashed} nodes can be substituted by \textit{healthy} nodes.

Current intrusion-tolerant systems are based on three main building blocks: (\textit{i}) a protocol for service replication that tolerates a subset of compromised and crashed nodes; (\textit{ii}) a replication strategy that adjusts the replication factor; and (\textit{iii}) a recovery strategy that determines when to recover potentially compromised nodes \cite{deswarte1991intrusion,10.1007/3-540-45177-3_1,resilience_nature}.

Replication protocols that satisfy the condition in (\textit{i}) are called \textit{Byzantine fault-tolerant} (\bft) and have been studied extensively (see survey \cite{bft_systems_survey}). In comparison, few prior works have studied (\textit{ii}) and (\textit{iii}). As a consequence, current intrusion-tolerant systems typically use a fixed replication factor \cite{virt_replica, 4365686, spire} and rely on inefficient recovery strategies, such as periodic recoveries \cite{pbft,virt_replica}, heuristic rule-based recoveries \cite{10.1007/3-540-45177-3_1}, or manual recoveries by system administrators \cite{int_prevention,rampart}.

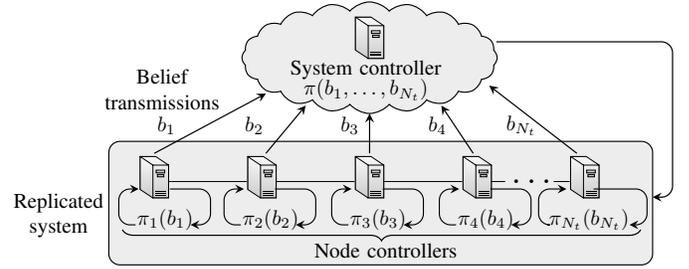
\begin{figure}
  \centering
  \scalebox{1.06}{
    \input{tikz/tolerance_13.tex}
  }
  \caption{Two-level feedback control for intrusion tolerance; node controllers with strategies $\pi_{1},\hdots, \pi_{N_t}$ compute belief states $b_1,\hdots,b_{N_t}$ and make local recovery decisions; a global system controller with strategy $\pi$ receives belief states and manages the replication factor $N_t$.}
  \label{fig:tolerance_8}
\end{figure}
In this paper, we address the above limitations and present \tolerancee, which stands for \underline{T}w\underline{o}-\underline{l}\underline{e}vel \underline{r}ecovery \underline{an}d replication \underline{c}ontrol with f\underline{e}edback. \tolerancee is a control architecture for intrusion-tolerant systems with two levels of control (see Fig. \ref{fig:tolerance_8}). On the local level node controllers perform intrusion recovery, and on the global level a system controller manages the replication factor. The associated control problems can be formulated as two classical problems in operations research, namely, the \textit{machine replacement problem} \cite{or_machine_replace_2} and the \textit{inventory replenishment problem} \cite{inventory_replenishment_donaldson}. Based on this formulation, we prove structural properties of the optimal control strategies and design efficient algorithms for computing them.

Key benefits of the \tolerancee architecture are: (\textit{i}) through feedback in terms of alerts and log files, the system can promptly adapt to network intrusions; (\textit{ii}) by using two control levels rather than one, the system can tolerate partial failures and network partitions; and (\textit{iii}) through the connection with classical problems from operations research, the system can leverage well-established control techniques with theoretical guarantees.

To assess the performance of \tolerancee, we implement it in an emulation environment where we run $10$ types of network intrusions. The results show that \tolerancee can achieve higher service availability and lower operational cost than state-of-the-art intrusion-tolerant systems.

Our contributions can be summarized as follows:
\begin{enumerate}
\item We present and evaluate \tolerancee, a novel control architecture for intrusion-tolerant systems that uses two levels of control to decide when to perform recovery and when to increase the replication factor.
\item We prove properties of the optimal control strategies and design efficient algorithms for computing them.
\item We implement \tolerancee in an emulation environment and evaluate its performance against $10$ types of network intrusions.
\end{enumerate}

\textbf{Software and data availability.} Source code, container images, and a dataset of $6400$ intrusion traces are available in the supplementary material \cite{supplementary,csle_docs}.
\section{The Intrusion Tolerance Use Case}
We consider a set of nodes that collectively offer a service to a client population. Each node is segmented into two domains: an application domain, which runs a service replica, and a privileged domain, which runs security and control functions. The replicas are coordinated through a replication protocol that relies on digital signatures and guarantees correct service if no more than $f$ nodes are compromised or crashed simultaneously.

Clients access the service through gateways, which also are accessible to an attacker. The attacker's goal is to intrude on the system and compromise replicas while avoiding detection. We assume that the attacker a) does not have physical access to nodes; b) can not forge digital signatures; and c) can only access the service replicas, not the privileged domains (i.e., we consider the \textit{hybrid failure model} \cite{wormit}). Apart from these restrictions, the attacker can control a compromised replica in arbitrary, i.e. Byzantine, ways. It can shut it down, delay its service responses, communicate with other replicas, etc.

To prevent the number of compromised and crashed nodes to exceed $f$, we consider three types of response actions: (\textit{i}) recover compromised nodes; (\textit{ii}) evict crashed nodes from the system; and (\textit{iii}) add new nodes. Each of these actions incurs a cost that must be weighed against the security benefit.

Note that, while we focus on the response actions: eviction, addition, and recovery of nodes in this paper, the use case can be extended to include additional response actions, such as rate limiting \cite{8692706} and access control \cite{hammar_stadler_tnsm}. This extension is however beyond the scope of this paper.
\section{Background on Intrusion-Tolerant Systems}
\subsection{Fault-Tolerant Systems}
Research on fault-tolerant systems has almost a century-long history with the seminal work being made by von Neumann \cite{vN56} and Shannon \cite{MOORE1956191} in 1956. The early work focused on tolerance against hardware failures. Since then the field has broadened to include tolerance against software bugs, operator mistakes, and malicious attacks \cite{old_fault_tolerance,avivzienis1967design,fault_tolerance_orig,cristian1991understanding,byzantine_generals_lamport}.

The common approach to build a fault-tolerant service is \textit{redundancy}, whereby the service is provided by a set of replicas. Through such redundancy, compromised and crashed replicas can be substituted by healthy replicas as long as the healthy replicas can coordinate their service responses. This coordination problem is known as the \textit{consensus} problem.

\begin{figure*}
  \centering
  \scalebox{1.05}{
    \input{tikz/tolerance_6.tex}
  }
  \caption{The \tolerancee architecture; $N_t$ nodes provide a replicated service to a client population; service responses are coordinated through an intrusion-tolerant consensus protocol; local node controllers decide when to perform recovery and a global system controller manages the replication factor $N_t$.}
  \label{fig:tolerance_6}
\end{figure*}
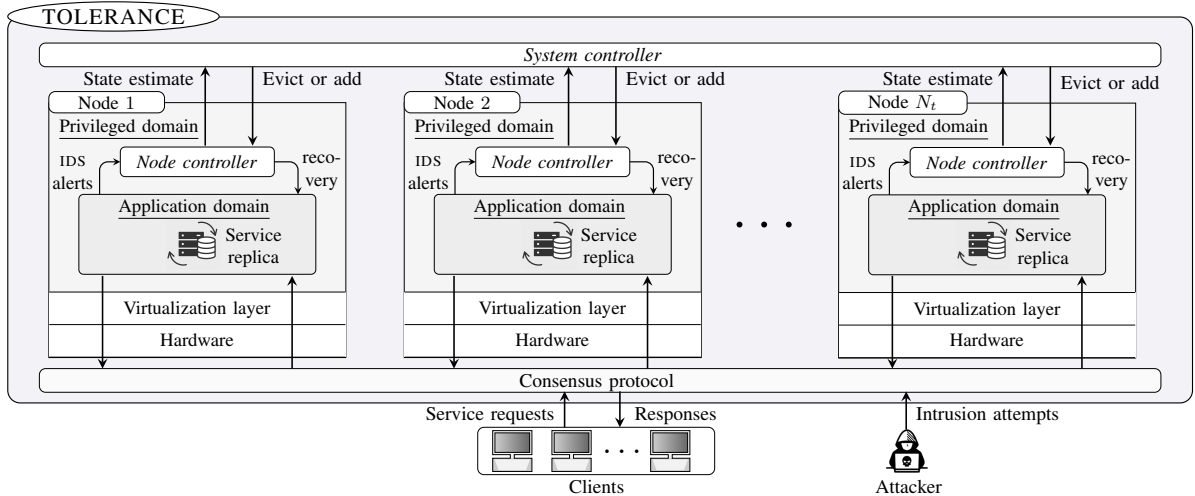
\subsection{Consensus}\label{sec:dist_consensus}
Consensus is the problem of reaching agreement among distributed nodes subject to failures \cite{cachin}. The solvability of consensus depends on synchrony and failure assumptions.

The main synchrony assumptions are: (\textit{i}) the \textit{synchronous model}, where there is an upper bound on the communication delay between nodes; (\textit{ii}) the \textit{partially synchronous model}, where an upper bound exists but the system may have periods of instability where the bound is violated; and (\textit{iii}) the \textit{asynchronous model}, where no bound exists \cite{cachin, partial_synchrony,dolev_partial_sync}.

The main failure assumptions are: (\textit{i}) the \textit{crash-stop failure model}, where nodes fail by crashing; (\textit{ii}) the \textit{Byzantine failure model}, where nodes fail arbitrarily; and (\textit{iii}) the \textit{hybrid failure model}, where nodes fail arbitrarily but are equipped with trusted components that fail by crashing \cite{dist_computing_book_formal,wormit}.

Deterministic consensus is \textit{not} solvable in the asynchronous model \cite[Thm. 1]{flp}. In the partially synchronous model, however, consensus \textit{is} solvable with $N$ nodes and $f=\frac{N-1}{2}$ crash-stop failures, $f=\frac{N-1}{3}$ Byzantine failures, and $f=\frac{N-1}{2}$ hybrid failures \cite[Thm. 1]{async_consensus_lb}\cite[Thms. 5.8,5.11]{dist_computing_book_formal}\cite[Thm. 2]{giuliana_thesis}\cite[Thm. 1]{ittai_hybrid_failures}. Similarly, consensus \textit{is} solvable in the synchronous model with $f=N-1$ crash-stop failures, $f=\frac{N-1}{2}$ Byzantine failures, and $f=\frac{N-1}{2}$ hybrid failures \cite[Thm. 5.2]{dist_computing_book_formal}\cite[Cor. 14]{synchronous_byzantine_bound} \cite[Thm. 1]{synchronous_byzantine_bound_2}.
\subsection{Intrusion-Tolerant Systems}\label{sec:avg_recov}
Intrusion-tolerant systems extend fault-tolerant systems with intrusion detection, recovery, and response \cite{ift_5,deswarte1991intrusion,10.1007/3-540-45177-3_1,932173}. We call a system \textit{intrusion-tolerant} if it remains secure and operational while intrusions occur \cite{deswarte1991intrusion,resilience_nature}. We use the following metrics to quantify intrusion tolerance:
\begin{enumerate}
\item \textit{Average time-to-recovery $T^{(\mathrm{R})}$}: the average time from node compromise until recovery starts.
\item \textit{Average availability $T^{(\mathrm{A})}$}: the fraction of time where the number of compromised and crashed nodes is at most $f$. (We assume the \textit{primary-partition} model \cite{birman_primary_partition} to circumvent the \capp theorem \cite[Thm. 2]{cap_theorem_1}.)
\item \textit{Frequency of recoveries $F^{(\mathrm{R})}$}: the fraction of time where recovery occurs.
\end{enumerate}
\section{Intrusion Tolerance through \\Two-Level Feedback Control}\label{sec:system_arch}
In this section, we describe \tolerancee: a two-level control architecture for intrusion-tolerant systems (see Fig. \ref{fig:tolerance_6}). It is a distributed system with $N_t \geq 2f+1+k$ \textit{nodes} connected through an authenticated network \cite{reconfigurable_consensus}. Each node runs a \textit{service replica}. The replicas are coordinated through a \textit{reconfigurable} \cite{reconfigurable_consensus} consensus protocol that guarantees \textit{correct service} if no more than $f$ nodes are compromised or crashed simultaneously (e.g., reconfigurable \minbft \cite[\S 4.2]{giuliana_thesis}).

\tolerancee uses two levels of control: local and global. On the local level, each node runs a \textit{node controller} that monitors the service replica through alerts from an Intrusion Detection System (\ids). Based on these alerts, the controller estimates the replica's state, i.e., whether it is compromised or not, and decides when it should be recovered. As each recovery incurs a cost, the challenge for the controller is to balance the recovery costs against the security benefits. To guarantee correct service, at most $k$ nodes are allowed to recover simultaneously.

The global level includes a \textit{system controller} that collects state estimates from the nodes and adjusts the replication factor $N_t$. When deciding if $N_t$ should be increased, the controller faces a classical dilemma in reliability theory \cite{reliability_theory_barlow}. On the one hand, it aims at high redundancy to maximize service availability. On the other hand, it does not want an excessively large and costly system.

Since the only task of the system controller is to execute control actions and communicate with the node controllers, it can be deployed on a standard crash-tolerant system, e.g., a \raft-based system \cite{raft}. For this reason, we consider the probability that the system controller crashes to be negligible.

Similar to the \vmfit and the \wormit architectures \cite{virt_replica, 4365686,wormit}, each node in \tolerancee is segmented into two domains: a \textit{privileged domain}, which can only fail by crashing, and an \textit{application domain}, which may be compromised by an attacker. The controllers and the \ids{}s execute in the privileged domain, whereas the service replicas execute in the application domain. The separation between the two domains can be realized in several ways. One option is to use a secure coprocessor to execute the privileged domain (e.g., \textsc{ibm 4758}) \cite{pbft,ittai_hybrid_failures}. Another option, which does not require special hardware, is to use a security kernel to run the privileged domain, as in the \wormit architecture \cite{wormit}. A third option, used in the \vmfit architecture \cite{virt_replica}, is to separate the application domain from the privileged domain using a secure virtualization layer that can be formally verified \cite{Dam675835}. \tolerancee implements the last option for the following reasons: (\textit{i}) virtualization enables efficient recovery of a compromised replica by replacing its virtual container \cite{4365686,virt_replica}; and (\textit{ii}) virtualization simplifies implementation of software diversification, which reduces the correlation between compromise events across nodes \cite{5958251}.
\subsection{Correctness}
We say that a system provides \textit{correct service} if the healthy replicas satisfy the following properties:
\begin{align*}
  \tag{Liveness} &\text{Each request is eventually executed}.\label{eq:liveness}\\
  \tag{Validity} &\text{Each executed request was sent by a client.}\label{eq:validity}\\
  \tag{Safety} &\text{Each replica executes the same request sequence.}\label{eq:safety}
\end{align*}
\begin{proposition}\label{prop:correctness}
\tolerancee provides correct service if the following holds:
\begin{enumerate}[(a)]
\item An attacker can not forge digital signatures.
\item Network links are authenticated and reliable \cite[p. 42]{cachin}.
\item At most $k$ nodes recover simultaneously and at most $f$ nodes are compromised or crashed simultaneously.
\item $N_t \geq 2f + 1 + k$ at all times $t$.
\item The system is partially synchronous \cite{partial_synchrony}.
\end{enumerate}
\end{proposition}
\begin{proof}[Proof (Sketch)]
(a)-(b) and the fact that the controllers can only fail by crashing imply the hybrid failure model (\S \ref{sec:dist_consensus}). (c)-(d) state that at least $f + 1 + k$ nodes are healthy at all times $t$. These properties together with the tolerance threshold $f = \frac{N_t-1-k}{2}$ of the consensus protocol (e.g., \minbft \cite[\S 4.2]{giuliana_thesis}) imply (\ref{eq:safety}) (\cite[Thms. 1--2]{giuliana_thesis}). Next, it follows from (e) that the healthy nodes will eventually agree on the response to any service request, which allows to circumvent \flp \cite[Thm. 1]{flp} and achieve (\ref{eq:liveness}). Finally, (\ref{eq:validity}) is ensured by the consensus protocol (e.g., \minbft \cite[\S 4.2]{giuliana_thesis}).
\end{proof}
\noindent \textbf{Remark.} By appropriate use of cryptographic methods and firewalls, \tolerancee can be extended to provide confidentiality in addition to (\ref{eq:safety}), (\ref{eq:liveness}), and (\ref{eq:validity}), see e.g., \cite{270414,6038579}. We leave this extension for future work.
\section{Formalizing Intrusion Tolerance As a \\Constrained Two-Level Control Problem}\label{sec:system_model}
To guarantee correct service, the controllers in \tolerancee must ensure that: a) the number of compromised and crashed nodes is at most $f$, which is achieved by recovery; and b) the number of nodes satisfies $N_t \geq 2f+1+k$, which is achieved by replacing crashed nodes (Prop. \ref{prop:correctness}). In the following, we formulate the problem of meeting these constraints while minimizing operational cost as a control problem with a local and a global level. On the local level, node controllers minimize cost while meeting a), and on the global level, the system controller minimizes cost while meeting b).

\vspace{1mm}

\noindent\textbf{Notation.} Random variables are denoted by upper-case letters (e.g., $X$) and their values by lower-case (e.g., $x$). $\mathbb{P}$ is a probability measure. The expectation of an expression $\phi$ with respect to $X$ is written as $\mathbb{E}_X[\phi]$. When $\phi$ includes many random variables that depend on $\pi$, we simply write $\mathbb{E}_{\pi}[\phi]$. $x \sim \phi$ means that $x$ is sampled from $\phi$. We use $\mathbb{P}[x]$ as a shorthand for $\mathbb{P}[X=x]$. Calligraphy letters (e.g., $\mathcal{V}$) represent sets. $\llbracket \cdot \rrbracket$ is the Iverson bracket. $\lfloor \cdot \rfloor$ is the floor function. Further notation is listed in Table \ref{tab:notation}.

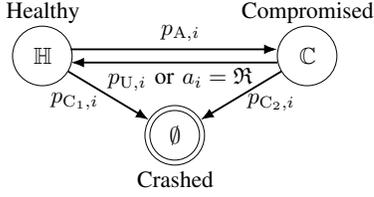
\begin{figure}
  \centering
  \scalebox{1.1}{
    \input{tikz/state_transition_node.tex}
  }
  \caption{State transition diagram of node $i$ (\ref{eq:recovery_dynamics}): disks represent states; arrows represent state transitions; labels indicate probabilities and conditions for state transition; self-transitions are not shown.}
  \label{fig:state_transition_node}
\end{figure}

\begin{table}
  \centering
  \scalebox{0.9}{
  \begin{tabular}{ll} \toprule
    {\textit{Notation(s)}} & {\textit{Description}} \\ \midrule
    $\mathcal{N}_t, N_t, f$ & Set of nodes, number of nodes, tolerance threshold (Prop. \ref{prop:correctness})\\
    $T^{(\mathrm{R})}$ & Average time-to-recovery when an intrusion occurs (\S \ref{sec:avg_recov})\\
    $F^{(\mathrm{R})}, T^{(\mathrm{A})}$ & Frequency of recoveries, average service availability (\ref{eq:objective_recovery})\\
    $T^{(f)}, k$ & Time to system failure (Fig. \ref{fig:reliability_curves_3}), \# parallel recoveries (Prop. \ref{prop:correctness})\\
    $R(t),R_i(t)$ & Reliability function for the system and for a node (Fig. \ref{fig:reliability_curves_3})\\
    $J_i, J$ & Objectives of node $i$ (\ref{eq:objective_recovery}) and the system controller (\ref{eq:objective_response})\\
    $\pi_{i,t}, \Pi_{\mathrm{N}}$ & Strategy and strategy space of node $i$, $\pi_{i,t} \in \Pi_{\mathrm{N}}$ (\ref{eq:recovery_problem})\\
    $\pi_{i,t}^{\star}, \alpha_i^{\star}$& Optimal strategy and threshold for node $i$ (\ref{eq:recovery_problem})\\
    $s_{i,t},o_{i,t}$ & State (\ref{eq:state_node_evolution}) and observation (\ref{eq:obs_function}) of node $i$ at time $t$\\
    $b_{i,t}$ & Belief state of node $i$ at time $t$ (\ref{eq:belief_def})\\
    $S_{i,t},O_{i,t}, B_{i,t}$ & Random variables with realizations $s_{i,t}$ (\ref{eq:state_node_evolution}), $o_{i,t}$ (\ref{eq:obs_function}), $b_{i,t}$ (\ref{eq:belief_def})\\
    $a_{i,t},c_{i,t}$ & Action and cost of node controller $i$ at time $t$ (\ref{eq:recovery_problem})\\
    $A_{i,t},C_{i,t}$ & Random variables with realizations $a_{i,t}$ and $c_{i,t}$ (\ref{eq:recovery_problem})\\
    $\mathcal{S}_{\mathrm{N}},\mathcal{A}_{\mathrm{N}},\mathcal{O}_{\mathrm{N}}$ & State, action and observation spaces of nodes\\
    $\mathfrak{W},\mathfrak{R} = 0,1$ & The ($\mathfrak{W}$)ait and ($\mathfrak{R}$)ecovery actions (Fig. \ref{fig:state_transition_node}) \\
    $\mathbb{H},\mathbb{C}=0,1$ & The ($\mathbb{H}$)ealthy and ($\mathbb{C}$)compromised node states (Fig. \ref{fig:state_transition_node})\\
    $\emptyset$ & The crashed node state (Fig. \ref{fig:state_transition_node}) \\
    $f_{\mathrm{N},i}, Z_i$ & Transition (\ref{eq:recovery_dynamics}) and observation (\ref{eq:obs_function}) functions for node $i$\\
    $c_{\mathrm{N}}(s_{i,t}, a_{i,t})$ & Cost function for a node (\ref{eq:objective_recovery})\\
    $p_{\mathrm{A},i}$ & Probability that node $i$ is compromised (\ref{eq:recovery_dynamics})\\
    $p_{\mathrm{C}_1,i}$ & Probability that node $i$ crashes in the healthy state (\ref{eq:recovery_dynamics})\\
    $p_{\mathrm{C}_2,i}$ & Probability that node $i$ crashes in the compromised state (\ref{eq:recovery_dynamics})\\
    $p_{\mathrm{U},i}$ & Probability that the service replica of node $i$ is updated (\ref{eq:recovery_dynamics})\\
    $\Delta_{\mathrm{R}}$ & Maximum allowed time between node recoveries (\ref{eq:recovery_problem})\\
    $\mathcal{W},\mathcal{R}$ & Wait and recovery sets for node controllers (Appendix \ref{appendix:theorem_1})\\
    $S_{t}, s_{t}$ & State of the system controller at time $t$ (\ref{eq:transition_system_controller}), $s_t$ realizes $S_t$\\
    $a_{t},c_t$ & Action and cost of system controller at time $t$ (\ref{eq:transition_system_controller})\\
    $A_{t},C_{t}$ & Random variables with realizations $a_t,c_t$ (\ref{eq:response_problem})\\
    $f_{\mathrm{S}}$ & Transition function of the system controller (\ref{eq:transition_system_controller})\\
    $\mathcal{S}_{\mathrm{S}},\mathcal{A}_{\mathrm{S}}$ & State and action spaces of the system controller\\
    $s_{\mathrm{max}}$ & Maximum number of nodes (\S \ref{sec:formal_global_level})\\
    $\pi, \Pi_{\mathrm{S}}$ & Strategy and strategy space of the system controller (\ref{eq:response_problem})\\
    $\epsilon_{\mathrm{A}}$ & Lower bound on the average service availability (\ref{eq:response_problem})\\
    \bottomrule\\
  \end{tabular}
}
  \caption{Notation.}\label{tab:notation}
\end{table}
\subsection{The Local Level: Controlling Intrusion Recovery}\label{sec:local_level}
The local level consists of $N_t$ node controllers, each of which decides when to perform intrusion recovery. We model this control problem as an instance of the \textit{machine replacement problem} from operations research \cite{or_machine_replace_2}.

Let $\mathcal{N}_t \triangleq \{1, 2, \hdots, N_t\}$ be the set of nodes and $\pi_{1,t}, \pi_{2,t},\hdots, \pi_{N,t}$ the corresponding control strategies at time $t$. Node $i$ has state $s_{i,t} \in \mathcal{S}_{\mathrm{N}}$ with three values: $\emptyset$ if it is crashed, $\mathbb{C}$ if it is compromised, and $\mathbb{H}$ if it is healthy (see Fig. \ref{fig:state_transition_node}). Controller $i$ takes action $a_{i,t} \in \mathcal{A}_{\mathrm{N}}$ with two values: $\mathfrak{R}$ is the recovery action and $\mathfrak{W}$ is the wait action. We assume that a recovery action at time $t$ is completed by $t+1$.

The evolution of $s_{i,t}$ can be written as
\begin{align}
  s_{i,t+1} &\sim f_{\mathrm{N},i}(\cdot \mid S_{i,t}=s_{i,t}, A_{i,t}=a_{i,t}),\label{eq:state_node_evolution}
\end{align}
where $f_{\mathrm{N},i}$ is a Markovian transition function defined as
\begin{subequations}\label{eq:recovery_dynamics}
  \begin{align}
    &f_{\mathrm{N},i}(\emptyset \mid \emptyset, \cdot) \triangleq 1 \label{eq:node_transition_1}\\
    &f_{\mathrm{N},i}(\emptyset \mid \mathbb{H}, \cdot) \triangleq p_{\mathrm{C}_1,i}\label{eq:node_transition_2}\\
    &f_{\mathrm{N},i}(\emptyset \mid \mathbb{C}, \cdot) \triangleq p_{\mathrm{C}_2,i}\label{eq:node_transition_8}\\
    &f_{\mathrm{N},i}(\mathbb{H} \mid \mathbb{H}, \mathfrak{R})\triangleq (1-p_{\mathrm{A},i})(1-p_{\mathrm{C}_1,i})\label{eq:node_transition_3}\\
    &f_{\mathrm{N},i}(\mathbb{H} \mid \mathbb{H}, \mathfrak{W}) \triangleq (1-p_{\mathrm{A},i})(1-p_{\mathrm{C}_1,i})\label{eq:node_transition_3_2}\\
    &f_{\mathrm{N},i}(\mathbb{H} \mid \mathbb{C}, \mathfrak{R})\triangleq(1-p_{\mathrm{A},i})(1-p_{\mathrm{C}_2,i})\label{eq:node_transition_9}\\
    &f_{\mathrm{N},i}(\mathbb{H} \mid \mathbb{C}, \mathfrak{W}) \triangleq (1-p_{\mathrm{C}_2,i})p_{\mathrm{U},i}\label{eq:node_transition_4}\\
    &f_{\mathrm{N},i}(\mathbb{C} \mid \mathbb{H}, \mathfrak{W}) \triangleq f_{\mathrm{N},i}(\mathbb{C} \mid \mathbb{H}, \mathfrak{R})=(1-p_{\mathrm{C_1},i})p_{\mathrm{A},i}\label{eq:node_transition_6}  \\
    &f_{\mathrm{N},i}(\mathbb{C} \mid \mathbb{C}, \mathfrak{R})\triangleq (1-p_{\mathrm{C_2},i})p_{\mathrm{A},i}\label{eq:node_transition_10}\\
    &f_{\mathrm{N},i}(\mathbb{C} \mid \mathbb{C}, \mathfrak{W}) \triangleq (1-p_{\mathrm{C}_2,i})(1-p_{\mathrm{U},i}),\label{eq:node_transition_7}
  \end{align}
\end{subequations}
where $p_{\mathrm{A},i}$ is the probability that the attacker compromises the node during the time interval $[t, t+1]$, $p_{\mathrm{C}_1,i}$ is the probability that the node crashes in the healthy state, $p_{\mathrm{C}_2,i}$ is the probability that the node crashes in the compromised state, and $p_{\mathrm{U},i}$ is the probability that the node's software is updated. These parameters can be set based on domain knowledge or be obtained through system measurements. In fact, companies such as Google, Meta, and IBM, have documented procedures for estimating such parameters, see e.g., \cite{google_failure_model, facebook_faulure_model}.

Equations (\ref{eq:node_transition_1})--(\ref{eq:node_transition_8}) capture the transitions to $\emptyset$, which is an absorbing state \cite[Def. 4.4]{douc2018markov}. (A crashed node that is restarted is considered as a new node in the model.) Next, (\ref{eq:node_transition_3})--(\ref{eq:node_transition_4}) define the transitions to $\mathbb{H}$, which is reached when the controller takes action $\mathfrak{R}$ (\ref{eq:node_transition_3})--(\ref{eq:node_transition_9}) or when its software is updated (\ref{eq:node_transition_4}). Lastly, (\ref{eq:node_transition_6})--(\ref{eq:node_transition_7}) capture the transitions to $\mathbb{C}$, which happen when an intrusion occurs.

It follows from (\ref{eq:recovery_dynamics}) that the number of time-steps until a node fails (crash or compromise) is geometrically distributed (see Fig. \ref{fig:reliability_curve_1}). Note that (\ref{eq:recovery_dynamics}) applies independently to each node, which means that compromise and crash events are statistically independent across nodes. This independence reflects our assumption that the system is geographically distributed and uses software diversification (\S \ref{sec:system_arch}).

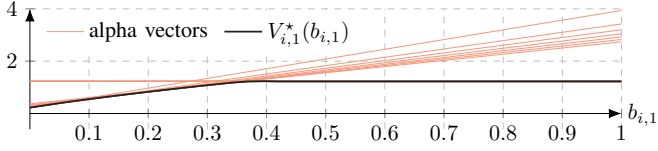
\begin{figure}
  \centering
  \scalebox{0.835}{
    \input{tikz/value_fun.tex}
  }
  \caption{The optimal value function $V^{\star}_{i,t}(b_{i,t})$ for Prob. \ref{prob_1}; the dashed red lines indicate the alpha-vectors \cite{smallwood_1,krishnamurthy_2016} and the solid black lines indicate $V^{\star}_{i,t}(b_{i,t})$; the parameters for computing $V^{\star}_{i,t}(b_{i,t})$ are listed in Appendix \ref{appendix:hyperparameters}.}
  \label{fig:value_funs}
\end{figure}

A node state is hidden from its controller. At time $t$, controller $i$ observes $o_{i,t} \in \mathcal{O} \subset \mathbb{N}_0$, which is based on the number of \ids alerts received during the time interval $[t-1,t]$, weighted by priority. (While we focus on the \ids alert metric in this paper, alternative sources of metrics can be used; a comparison between different metrics is available in Appendix \ref{appendix:infrastructure_metrics}.) $o_{i,t}$ is drawn from the random variable $O_{i,t}$ with distribution
\begin{align}
Z_i(o_{i,t} \mid s_{i,t}) &\triangleq \mathbb{P}\left[O_{i}=o_{i,t} \mid S_{i,t}=s_{i,t}\right].\label{eq:obs_function}
\end{align}
Based on its observations, controller $i$ computes the belief state
\begin{align}
  b_{i,t} \triangleq \mathbb{P}[S_{i,t} = \mathbb{C} \mid o_{i,1}, a_{i,1},\hdots, a_{i,t-1}, o_{i,t}, b_{i,1}],\label{eq:belief_def}
\end{align}
which is a sufficient statistic for $s_{i,t}$ (see Appendix \ref{appendix_belief_comp}). $\pi_{i,t}$ can thus be defined as $\pi_{i,t}: [0,1] \rightarrow \{\mathfrak{W}, \mathfrak{R}\}$.

When selecting the strategy $\pi_{i,t}$, the controller balances two conflicting goals: minimize the average time-to-recovery $T_{i}^{(\mathrm{R})}$ and minimize the frequency of recoveries $F^{(R)}_{i}$. The weight $\eta \geq 1$ controls the trade-off between these two objectives, which results in the bi-objective
\begin{align}
  \minimize J_i &\triangleq \lim_{T \rightarrow \infty}\left[\eta T_{i,T}^{(\mathrm{R})} + F^{(\mathrm{R})}_{i,T}\right]\label{eq:objective_recovery}\\
                &=\lim_{T \rightarrow \infty}\Bigg[\frac{1}{T}\sum_{t=1}^T\underbrace{\eta s_{i,t} - a_{i,t}\eta s_{i,t} + a_{i,t}}_{\triangleq c_{\mathrm{N}}(s_{i,t}, a_{i,t})}\Bigg],\nonumber
\end{align}
where $T_{i,T}^{(\mathrm{R})}$ and $F^{(\mathrm{R})}_{i,T}$ denote the average values at time $T$, $\mathbb{H},\mathbb{C}=0,1$, $\mathfrak{W},\mathfrak{R}=0,1$, and $c_{\mathrm{N}}$ is the cost function.

We define the intrusion recovery problem as that of minimizing $J_i$ (\ref{eq:objective_recovery}) subject to a bounded-time-to-recovery (\btr) constraint \cite{rebound_btr}. Formally,
\begin{figure}
  \centering
  \scalebox{0.87}{
    \input{tikz/reliability_curve_1.tex}
  }
  \caption{Probability that a node is compromised ($\mathbb{C}$) or crashed $(\emptyset)$ by time-step $t$ if no recoveries occur; the curves relate to $p_{\mathrm{A},i}$ (\ref{eq:recovery_dynamics}); hyperparameters are listed in Appendix \ref{appendix:hyperparameters}.}
  \label{fig:reliability_curve_1}
\end{figure}
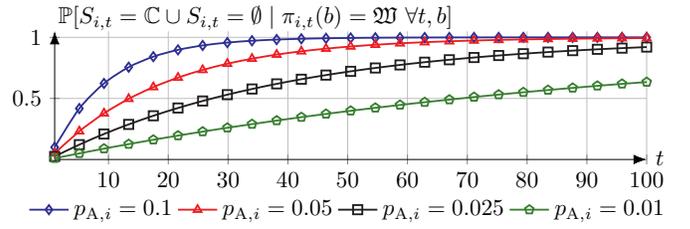
\begin{problem}[Optimal Intrusion Recovery]\label{prob_1}
  \begin{subequations}\label{eq:recovery_problem}
    \begin{align}
      \minimize_{\pi_{i,t} \in \Pi_{\mathrm{N}}} \quad &\mathbb{E}_{\pi_{i,t}}\left[J_i \mid b_{i,1}=p_{\mathrm{A},i}\right] && \forall i \in \mathcal{N}_t\label{eq:node_maximization}\\
      \text{\normalfont subject to}\quad & a_{i,k\Delta_{\mathrm{R}}} = \mathfrak{R} &&\forall i,k \label{eq:recovery_constraint}\\
                                          &s_{i,t+1} \sim f_{\mathrm{N},i}(\cdot \mid s_{i,t}, a_{i,t}) &&\forall i,t \label{eq:dynamics_recov_constraint}\\
      \quad & o_{i,t+1} \sim Z_i(\cdot \mid s_{i,t}) &&\forall i,t \label{eq:obs_recov_constraint}\\
      \quad & a_{i, t+1} = \pi_{i,t}(b_{i,t}) &&\forall i,t, \label{eq:action_recov_constraint}
    \end{align}
  \end{subequations}
\end{problem}
where $t,k=1,2,\hdots$; $\Pi_{\mathrm{N}}$ is the strategy space; $b_{i,1}$ is the initial state distribution of node $i$; $\mathbb{E}_{\pi_{i,t}}$ denotes the expectation over the random variables $(S_{i,t}, O_{i,t}, A_{i,t}, B_{i,t})_{t \in \{1,2,\hdots\}}$ when following strategy $\pi_{i,t}$; (\ref{eq:recovery_constraint}) is the \btr constraint; (\ref{eq:dynamics_recov_constraint}) is the dynamics constraint; (\ref{eq:obs_recov_constraint}) captures the observations; and (\ref{eq:action_recov_constraint}) captures the actions. (Remark: Prob. \ref{prob_1} does not include a constraint on the maximum number of parallel recoveries (i.e., $k$ in Prop. \ref{prop:correctness}) as we assume this constraint is enforced by the implementation.)

We say that a \textit{strategy $\pi_{i,t}^{\star}$ is optimal} if it solves (\ref{eq:recovery_problem}). Figure \ref{fig:value_funs} shows the expected cost of $\pi_{i,t}^{\star}$. We note that $\pi_{i,t}^{\star}$ has threshold structure, as stated below.
\begin{theorem}\label{thm:stopping_policy}
Assuming
  \begin{align}
    &\tag{A}p_{\mathrm{A},i}, p_{\mathrm{U},i}, p_{\mathrm{C}_1,i}, p_{\mathrm{C}_2,i} \in (0,1) \label{eq:assumption_A}\\
    & \tag{B}p_{\mathrm{A},i} + p_{\mathrm{U},i} \leq 1 \label{eq:assumption_B}\\
    &\tag{C}\frac{p_{\mathrm{C}_1,i}(p_{\mathrm{U},i}-1)}{p_{\mathrm{A},i}(p_{\mathrm{C}_1,i}-1) + p_{\mathrm{C}_1,i}(p_{\mathrm{U},i}-1)} \leq p_{\mathrm{C}_2,i} \label{eq:assumption_C}\\
    &\tag{D}Z_i(o_{i,t} \mid s_{i,t}) > 0 \quad\quad\forall o_{i,t} \in \mathcal{O}, s_{i,t} \in \mathcal{S}_{\mathrm{N}}\label{eq:assumption_D}\\
    &\tag{E}Z_i\text{ }\text{is }\text{\tpp }\text{ }\label{eq:assumption_E}\text{\cite[Def. 10.2.1]{krishnamurthy_2016}},
  \end{align}
  then there exists an optimal strategy $\pi_{i,t}^{\star}$ that solves Prob. \ref{prob_1} for each node $i$ and satisfies
  \begin{align}
    \pi_{i,t}^{\star}(b_{i,t}) = \mathfrak{R} \iff b_{i,t} &\geq \alpha_{i,t}^{\star} && \forall t,\label{eq:threhsold_structure}
  \end{align}
where $\alpha_{i,t}^{\star} \in [0,1]$ is a threshold.
\end{theorem}
\begin{proof}
See Appendix \ref{appendix:theorem_1}.
\end{proof}
\begin{corollary}\label{corr:increasing_thresholds}
The thresholds satisfy $\alpha_{i,t+1}^{\star} \geq \alpha_{i,t}^{\star}$ for $t \in [k\Delta_{R}, (k+1)\Delta_{R}]$ and $i \in \mathcal{N}$. As $\Delta_{\mathrm{R}} \rightarrow \infty$, all thresholds converge to $\alpha_{i}^{\star}$, which is time-independent.
\end{corollary}
\begin{proof}
See Appendix \ref{appendix:corollary_1}.
\end{proof}
Theorem \ref{thm:stopping_policy} states that under assumptions generally met in practice, there exists an optimal strategy for each node that performs recovery when the belief (\ref{eq:belief_def}) exceeds a threshold (\ref{eq:threhsold_structure}). Further, Cor. \ref{corr:increasing_thresholds} says two things: (\textit{i}) the threshold increases as the time until the next periodic recovery decreases; and (\textit{ii}) when there are no periodic recoveries (i.e., when $\Delta_{R}=\infty$), the threshold is independent of time.

The above statements rely on assumptions \ref{eq:assumption_A}--\ref{eq:assumption_E}. \ref{eq:assumption_A}--\ref{eq:assumption_C} are mild assumptions stating that the attack, crash, and upgrade probabilities are small but non-zero and that the difference $p_{\mathrm{C}_2,i}-p_{\mathrm{C}_1,i} > 0$ is sufficiently large. \ref{eq:assumption_D} is a technical assumption saying that the probability of observing $k$ \ids alerts is non-zero for any $k \in \mathcal{O}$, which generally is true in practice (see \S \ref{sec:solution_approach}). \ref{eq:assumption_E} states that the probability of high-priority \ids alerts increases when an intrusion occurs. While \ref{eq:assumption_E} may not always hold, empirical studies suggest that it holds for many types of intrusions \cite{hammar_stadler_tnsm,hammar_stadler_tnsm_23}. If \ref{eq:assumption_E} does not hold, the threshold strategy in (\ref{eq:threhsold_structure}) can still achieve near-optimal performance but it is not guaranteed to be optimal.

Theorem \ref{thm:stopping_policy} and Cor. \ref{corr:increasing_thresholds} lead to two important practical benefits. First, the complexity of computing optimal strategies can be reduced by only considering threshold strategies (see \S \ref{sec:solution_approach}). Second, the optimal strategies can be efficiently implemented.

\subsection{The Global Level: Controlling the Replication Factor}\label{sec:formal_global_level}
The global level includes a \textit{system controller} that adjusts the replication factor $N_t$. At each time $t$, it receives the belief states $b_{1,t},\hdots,b_{N_t,t}$ from the nodes and decides whether $N_t$ should be increased (see Fig. \ref{fig:tolerance_8}). A node that fails to send $b_{i,t}$ is considered to have crashed and is evicted from the system, which decrements $N_t$. (Remark: in practice, an evicted node can rejoin the system, but it is then considered as a new node in the model.) We assume that the system controller does not crash, i.e., we assume it is crash-tolerant (\S \ref{sec:system_arch}).

A large replication factor $N_t$ improves the service availability but increases cost (see Fig. \ref{fig:reliability_curves_3}). Our goal thus is to find the optimal cost-redundancy trade-off. We model this control problem as an instance of the \textit{inventory replenishment problem} from operations research \cite{inventory_replenishment_donaldson}.

We define the state $s_{t}$ to represent the expected number of healthy nodes at time $t$. The state space is $\mathcal{S}_{\mathrm{S}} \triangleq \{0,1,\hdots,s_{\mathrm{max}}\}$ and the initial state is $s_{1}=N_1$.

At time $t$, the controller adds $a_t \in \{0,1\} \triangleq \mathcal{A}_{\mathrm{S}}$ nodes:
\begin{align}
s_{t+1} &\sim f_{\mathrm{S}}(\cdot \mid S_t=s_t, A_t=a_t),\label{eq:transition_system_controller}
\end{align}
where $f_{\mathrm{S}}$ is defined as
\begin{align*}
f_{\mathrm{S}}(s_{t+1} \mid s_t,a_t) &\triangleq \mathbb{P}\left[\left\lfloor \sum_{i \in \mathcal{N}_t} 1-B_{i,t}\right\rfloor = s_{t+1}-a_t\right].
\end{align*}

When selecting the strategy $\pi$, the controller balances two conflicting goals: maximize the average service availability $T^{(\mathrm{A})}$ and minimize the number of nodes $s_t$. We model this bi-objective as follows.
\begin{align}
  \text{minimize }&J\triangleq \lim_{T \rightarrow \infty}\left[\frac{1}{T}\sum_{t=1}^Ts_t\right]\label{eq:objective_response}\\
  \text{subject to }&T^{(\mathrm{A})} \geq \epsilon_{\mathrm{A}},\text{ which can be expressed as (Prop. \ref{prop:correctness})}\nonumber\\
  &\lim_{T \rightarrow \infty}\left[\frac{1}{T}\sum_{t=1}^{T}\llbracket s_t \geq f+1\rrbracket \right] \geq \epsilon_{\mathrm{A}},\nonumber
\end{align}
where $\epsilon_{\mathrm{A}}$ is the lower bound on service availability.

Given (\ref{eq:objective_response}) and the Markov property of $s_t$ (\ref{eq:transition_system_controller}), we define $\pi$ as a function $\pi: \mathcal{S}_{\mathrm{S}} \rightarrow \Delta(\mathcal{A}_{\mathrm{S}})$, where $\Delta(\mathcal{A}_{\mathrm{S}})$ is the set of probability distributions over $\mathcal{A}_{\mathrm{S}}$. Based on this definition, we formulate the problem of controlling the replication factor as
\begin{problem}[Optimal Replication Factor]\label{prob_2}
  \begin{subequations}\label{eq:response_problem}
    \begin{align}
      \minimize_{\pi \in \Pi_{\mathrm{S}}} \quad &\mathbb{E}_{\pi}\left[J \mid s_{1}=N_1\right]\label{eq:system_maximization}\\
      \mathrm{subject}\text{ }\mathrm{to} \quad & \mathbb{E}_{\pi}\left[T^{(\mathrm{A})}\right] \geq \epsilon_{\mathrm{A}}  \label{eq:system_avail_constraint}\\
      \quad & s_{t+1} \sim f_{\mathrm{S}}(\cdot \mid s_t,a_t)&&\forall t \label{eq:system_dynamics_constraint}\\
      \quad & a_{t+1} \sim \pi(\cdot \mid s_{t}) &&\forall t, \label{eq:action_response_constraint}
    \end{align}
  \end{subequations}
\end{problem}
where $\Pi_{\mathrm{S}}$ is the strategy space; $s_1$ is the initial state; $\mathbb{E}_{\pi}$ denotes the expectation of the random variables $(S_t, A_t)_{t=1,2,\hdots}$ under strategy $\pi$; (\ref{eq:system_avail_constraint}) is the availability constraint; (\ref{eq:system_dynamics_constraint}) is the dynamics constraint; and (\ref{eq:action_response_constraint}) captures the actions.

\begin{figure}
  \centering
  \begin{subfigure}[t]{1\columnwidth}
    \centering
    \scalebox{0.85}{
      \input{tikz/mttf.tex}
    }
    \caption{Mean time to failure (\mttf) in function of the number of initial nodes $N_1$; the curves relate to $p_{\mathrm{A},i}$ (\ref{eq:recovery_dynamics}).}
    \label{fig:mttf}
  \end{subfigure}
  \hfill
  \begin{subfigure}[t]{1\columnwidth}
    \centering
    \scalebox{0.85}{
      \input{tikz/reliability_curve_2.tex}
    }
    \caption{Reliability curves for varying number of initial nodes $N_1$; The reliability function is defined as $R(t) \triangleq \mathbb{P}[T^{(f)} > t]$.}\label{fig:reliability_curve_2}
  \end{subfigure}
  \caption{Illustration of Prob. \ref{prob_2}; $T^{(f)}$ is a random variable representing the time when $N_t < 2f +k +1$ with $f=3$ and $k=1$ (Prop. \ref{prop:correctness}); hyperparameters and formulas for computing the curves are listed in Appendix \ref{appendix:hyperparameters}--\ref{appendix_mttf}.}\label{fig:reliability_curves_3}
\end{figure}
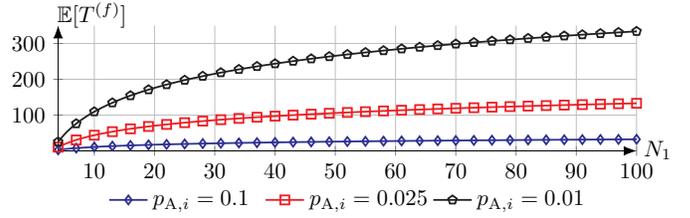
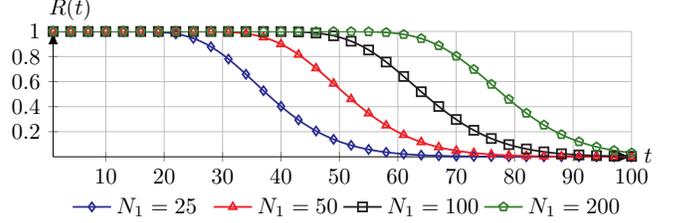

We say that a \textit{strategy $\pi^{\star}$ is optimal} if it solves (\ref{eq:response_problem}). We note that (\ref{eq:response_problem}) states a Constrained Markov Decision Problem (\cmdp) \cite{altman-constrainedMDP}, which leads to the following result.
\begin{theorem}\label{thm:structure_respone}
  Assuming
  \begin{align}
    &\tag{A} \exists \pi \in \Pi_{\mathrm{S}} \text{ such that }\mathbb{E}_{\pi}\left[T^{(\mathrm{A})}\right] \geq \epsilon_{\mathrm{A}}\label{eq:feasibility_assumption}\\
    &\tag{B} f_{\mathrm{S}}(s^{\prime} \mid s, a) > 0 \quad\quad\quad\quad\quad\quad\quad\quad\quad\quad\quad\text{ }\text{ }\text{ }\text{ }\text{ }  \label{eq:unichain_assumption}\\
    &\tag{C} \sum_{s^{\prime}=s}^{s_{\mathrm{max}}} f_{\mathrm{S}}(s^{\prime}\mid\hat{s}+1, a) \geq \sum_{s^{\prime}=s}^{s_{\mathrm{max}}} f_{\mathrm{S}}(s^{\prime} \mid \hat{s}, a) \label{eq:dominance_transitions}\\
    &\tag{D} \sum_{s^{\prime}=s}^{s_{\mathrm{max}}} f_{\mathrm{S}}(s^{\prime}|\hat{s}, 1)-f_{\mathrm{S}}(s^{\prime}|\hat{s}, 0) \text{ is increasing in $s$} \label{eq:tail_sum}
  \end{align}
  for all $(s,\hat{s},a) \in \mathcal{S}_{\mathrm{S}}\times \mathcal{S}_{\mathrm{S}} \times \mathcal{A}_{\mathrm{S}}$.

  Then there exist two strategies $\pi_{\lambda_1}$ and $\pi_{\lambda_2}$ that satisfy
  \begin{align}
    \pi_{\lambda_1}(s_{t}) &= 1 \iff s_{t} \leq \beta_1 && \forall t, s_t \in \mathcal{S}_{\mathrm{S}}\\
    \pi_{\lambda_2}(s_{t}) &= 1 \iff s_{t} \leq \beta_2 && \forall t, s_t \in \mathcal{S}_{\mathrm{S}}
  \end{align}
  and an optimal strategy $\pi^{\star}$ that satisfies
  \begin{align}
    \pi^{\star}(s_{t}) &= \kappa \pi_{\lambda_1}(s_{t}) + (1-\kappa)\pi_{\lambda_2}(s_{t}) && \forall t,s_t \in \mathcal{S}_{\mathrm{S}} \label{eq:randomized_threshold}
  \end{align}
  for some probability $\kappa \in [0,1]$, where $\lambda_1,\lambda_2$ are Lagrange multipliers and $\beta_1,\beta_2$ are thresholds.
\end{theorem}
\begin{proof}
See Appendix \ref{appendix:theorem_2}.
\end{proof}
Theorem \ref{thm:structure_respone} states that under assumptions \ref{eq:feasibility_assumption}-\ref{eq:tail_sum}, there exists an optimal strategy that can be written as a mixture of two threshold strategies. \ref{eq:feasibility_assumption} ensures that Prob. \ref{prob_2} is feasible. This assumption can be met by tuning $\epsilon_{\mathrm{A}}$ (\ref{eq:system_avail_constraint}). \ref{eq:unichain_assumption}--\ref{eq:tail_sum} are mild assumptions which generally hold in practice. \ref{eq:unichain_assumption} states that the probability of several simultaneous intrusions or recoveries is non-zero; \ref{eq:dominance_transitions} conveys that a large number of healthy nodes at current time increases the probability of having a large number of healthy nodes in the future; and \ref{eq:tail_sum} states that the transition probabilities are tail-sum supermodular \cite[Eq. 9.6]{krishnamurthy_2016}.
\section{Computing Optimal Control Strategies}\label{sec:finding_strategies}
\begin{algorithm}
\scriptsize
  \SetNoFillComment
  \SetKwProg{myInput}{Input:}{}{}
  \SetKwProg{myOutput}{Output:}{}{}
  \SetKwFunction{pomdpsolver}{\textsc{pomdpsolver}}
  \SetKwProg{myalg}{Algorithm}{}{}
  \SetKwProg{myproc}{Procedure}{}{}
  \SetKw{KwTo}{inp}
  \SetKwFor{Forp}{in parallel for}{\string do}{}%
  \SetKwProg{mylamb}{\textbf{$\lambda$}}{}{}
  \SetKwFor{Loop}{Loop}{}{EndLoop}
  \DontPrintSemicolon
  \SetKwBlock{DoParallel}{do in parallel}{end}
  \myInput{
    \upshape Problem \ref{prob_1}, a node $i\in \mathcal{N}$, and a parametric optimizer $\textsc{po}$.
  }{}
  \myOutput{
    \upshape A near-optimal recovery strategy $\hat{\pi}_{i,\bm{\theta},t}$ (\ref{eq:recovery_problem}).
  }{}
  \caption{Parametric optimization for optimal recovery strategy.}\label{alg:solver_prob1}
  \myalg{}{
    if $\Delta_{\mathrm{R}} < \infty$, $\mathrm{d} \leftarrow \Delta_{\mathrm{R}}-1$,  else $\mathrm{d} \leftarrow 1$\;
    $\Theta \leftarrow [0,1]^{\mathrm{d}}$\;
   $ \begin{aligned}
      \pi_{i,\bm{\theta},t}(b_t) \triangleq
      \begin{dcases}
        \mathfrak{R} & \text{if }b_t \geq \bm{\theta}_k \text{ where }k=\max[t,\mathrm{d}]\\
        \mathfrak{W} & \text{otherwise}
      \end{dcases} \quad \forall \bm{\theta} \in \Theta, t \geq 1
    \end{aligned}$\;
    $J_{i,\bm{\theta}} \leftarrow \mathbb{E}_{\pi_{i,\bm{\theta},t,}}[J_i]$ where $J_i$ is defined in (\ref{eq:objective_recovery})\;
    $\hat{\pi}_{i,\bm{\theta},t} \leftarrow \textsc{po}(\Theta, J_{i,\bm{\theta}})$\;
    \Return $\hat{\pi}_{i,\bm{\theta},t}$
  }
\normalsize
\end{algorithm}

\begin{algorithm}
\scriptsize
  \SetNoFillComment
  \SetKwProg{myInput}{Input:}{}{}
  \SetKwProg{myOutput}{Output:}{}{}
  \SetKwFunction{pomdpsolver}{\textsc{pomdpsolver}}
  \SetKwProg{myalg}{Algorithm}{}{}
  \SetKwProg{myproc}{Procedure}{}{}
  \SetKw{KwTo}{inp}
  \SetKwFor{Forp}{in parallel for}{\string do}{}%
  \SetKwProg{mylamb}{\textbf{$\lambda$}}{}{}
  \SetKwFor{Loop}{Loop}{}{EndLoop}
  \DontPrintSemicolon
  \SetKwBlock{DoParallel}{do in parallel}{end}
  \myInput{
    \upshape Problem \ref{prob_2} and a linear programming solver $\textsc{lp}$.
  }{}
  \myOutput{
    \upshape An optimal replication strategy $\pi^{\star}$ (\ref{eq:response_problem}).
  }{}
  \caption{Linear program for optimal replication strategy.}\label{alg:solver_prob2}
  \myalg{}{
    Solve (\ref{eq:cmdp_lp}) using $\textsc{lp}$; let $\rho^{\star}$ denote the solution of (\ref{eq:cmdp_lp}) and define
    \begin{align*}
      \pi^{\star}(a|s) &\triangleq \frac{\rho^{\star}(s,a)}{\sum_{s \in \mathcal{S}_{\mathrm{S}}}\rho^{\star}(s,a)} && \forall s \in \mathcal{S}_{\mathrm{S}}, a \in \mathcal{A}_{\mathrm{S}}
    \end{align*}
    \begin{subequations}\label{eq:cmdp_lp}
    \begin{align}
      &\minimize_{\rho}\text{ } \sum_{s \in \mathcal{S}_{\mathrm{S}}}\sum_{a \in \mathcal{A}_{\mathrm{S}}}s\rho(s,a)\label{eq:cmdp_lp_objective}\\
      &\mathrm{subject}\text{ }\mathrm{to}\text{ }\rho(s,a) \geq 0 \quad \forall s \in \mathcal{S}_{\mathrm{S}}, a \in \mathcal{A}_{\mathrm{S}}\label{eq:cmdp_lp_prob_1}\\
                               &\sum_{s \in \mathcal{S}_{\mathrm{S}}}\sum_{a \in \mathcal{A}_{\mathrm{S}}} \rho(s,a) = 1\label{eq:cmdp_lp_prob_2}\\
      &\sum_{a \in \mathcal{A}_{\mathrm{S}}}\rho(s,a) = \sum_{s^{\prime} \in \mathcal{S}_{\mathrm{S}}}\sum_{a \in \mathcal{A}_{\mathrm{S}}}\rho(s^{\prime},a)f_{\mathrm{S}}(s^{\prime}|s,a) \text{ }\forall s \in \mathcal{S}_{\mathrm{S}}\label{eq:cmdp_lp_transition}\\
      &\sum_{s \in \mathcal{S}_{\mathrm{S}}}\sum_{a \in \mathcal{A}_{\mathrm{S}}} \rho(s,a)\llbracket s_t \geq f+1\rrbracket \geq \epsilon_{\mathrm{A}}\label{eq:cmdp_lp_constraint}
    \end{align}
    \end{subequations}
    \Return $\pi^{\star}$
  }
\normalsize
\end{algorithm}
The time complexity of Prob. \ref{prob_1} (optimal intrusion recovery) is in the complexity class \pspace-hard \cite[Thm. 6]{pspace_complexity}. (Recall that $\text{\textsc{p}} \subseteq \text{\textsc{np}} \subseteq \text{\textsc{pspace}}$.) The time complexity of Prob. \ref{prob_2} (optimal replication factor), on the other hand, is polynomial \cite[Thm. 1]{KHACHIYAN198053}\cite[Thm. 4.3]{altman-constrainedMDP}.

To manage the high time complexity of Prob. \ref{prob_1}, we exploit Thm. \ref{thm:stopping_policy} and parameterize $\pi_{i,t}^{\star}$ with a finite number of thresholds. Given this parametrization, we formulate Prob. \ref{prob_1} as a parametric optimization problem, which can be solved with standard optimization algorithms, e.g., stochastic approximation \cite{spsa_impl}. Algorithm \ref{alg:solver_prob1} contains the pseudocode of our solution.

We solve Prob. \ref{prob_2} with Alg. \ref{alg:solver_prob2}, which leverages the linear programming formulation of \cmdp{}s in \cite[Thm. 4.3]{altman-constrainedMDP}. It takes a linear programming solver as input, formulates Prob. \ref{prob_2} as a linear program (\ref{eq:cmdp_lp}), and solves it using the provided solver (line $5$) \cite[Thm. 6.6.1]{krishnamurthy_2016}. (Remark: The correctness of Alg. \ref{alg:solver_prob2} is independent of Thm. \ref{thm:structure_respone} and therefore solves Prob. \ref{prob_2} even if the assumptions of Thm. \ref{thm:structure_respone} do not hold.)

\begin{table*}
  \centering
  \begin{tabular}{l|ll|ll|ll|ll}
    \toprule
    \multirow{2}{*}{Method} &
      \multicolumn{2}{c|}{$\Delta_{\mathrm{R}}=5$} &
     \multicolumn{2}{c}{$\Delta_{\mathrm{R}}=15$} &
     \multicolumn{2}{c}{$\Delta_{\mathrm{R}}=25$} &
     \multicolumn{2}{c}{$\Delta_{\mathrm{R}}=\infty$}\\
      & {Time (min)} & {$J_i$ (\ref{eq:objective_recovery})} & {Time (min)} & {$J_i$ (\ref{eq:objective_recovery})} & {Time (min)} & {$J_i$ (\ref{eq:objective_recovery})} & {Time (min)} & {$J_i$ (\ref{eq:objective_recovery})} \\
      \midrule
    \cem \cite[Alg. 1]{moss2020crossentropy} & $\bm{1.04}$ & $\bm{0.12 \pm 0.01}$ & $\bm{8.84}$ & $\bm{0.17 \pm 0.06}$ & $\bm{14.48}$ & $0.19 \pm 0.08$ & $11.81$ & $0.16 \pm 0.01$\\
    \de \cite[Fig. 3]{differential_evolution} & $2.35$ & $\bm{0.12 \pm 0.03}$ & $8.98$ & $\bm{0.17 \pm 0.01}$ & $15.45$ & $\bm{0.18 \pm 0.02}$ & $22.68$ & $0.16 \pm 0.01$\\
    \bo \cite[Alg. 1]{bayesian_opt} & $29.18$ & $\bm{0.12 \pm 0.02}$ & $62.57$ & $\bm{0.17 \pm 0.05}$ & $90.26$ & $\bm{0.18 \pm 0.12}$ & $9.07$ & $\bm{0.15 \pm 0.06}$\\
    \spsa \cite[Fig. 1]{spsa_impl} & $10.78$ & $0.18 \pm 0.01$ & $88.35$ & $0.58 \pm 0.40$ & $123.85$ & $0.77 \pm 0.48$ & $\bm{4.20}$ & $0.20 \pm 0.02$\\
    \midrule
    \ppo \cite[Alg. 1]{ppo} & $28.20$ & $0.18 \pm 0.01$ & $30.01$ & $0.19 \pm 0.02$ & $30.33$ & $0.21 \pm 0.07$ & $28.95$ & $0.21 + \pm 0.09$\\
    \ipp \cite[Fig. 4]{incremental_pruning_pomdp} & $11.11$ & $\bm{0.12}$ & $237.06$ & $\bm{0.17}$ & $743.73$ & $\bm{0.18}$ & $>10000$ & not converged\\
    \bottomrule
  \end{tabular}
  \caption{Solving Prob. \ref{prob_1} (optimal intrusion recovery) using Alg. \ref{alg:solver_prob1} (upper rows) and baselines (lower rows); columns represent $\Delta_{\mathrm{R}}$; subcolumns indicate the computational time (left) and the average cost (right); numbers indicate the mean and the $95\%$ confidence interval based on $20$ random seeds.}\label{tab:results_prob_1}
\end{table*}
\begin{figure*}
\centering
\scalebox{0.52}{
      \includegraphics{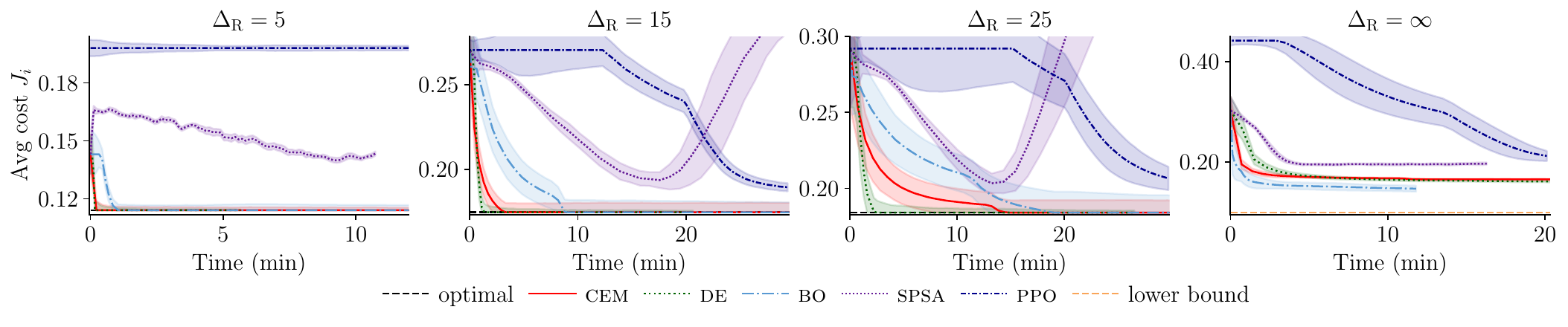}
}
\caption{Convergence curves of Alg. \ref{alg:solver_prob1} for Prob. \ref{prob_1} (optimal intrusion recovery); the curves show the mean value from evaluations with $20$ random seeds and the shaded areas indicate the $95\%$ confidence interval divided by $10$.}
    \label{fig:recovery_curves}
  \end{figure*}
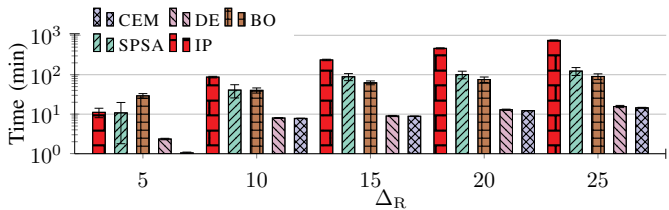
\begin{figure}
  \centering
  \scalebox{0.835}{
    \input{tikz/pomdp_times.tex}
  }
  \caption{Mean compute time to solve Prob. \ref{prob_1} for different algorithms and values of $\Delta_{\mathrm{R}}$; the error bars indicate the $95\%$ confidence interval based on $20$ measurements.}
  \label{fig:pomdp_times}
\end{figure}
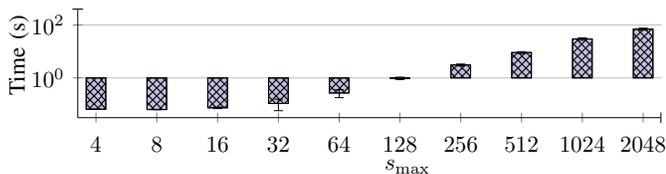
\begin{figure}
  \centering
  \scalebox{0.87}{
    \input{tikz/cmdp_times.tex}
  }
  \caption{Mean compute time to solve Prob. \ref{prob_2} (optimal replication factor) with Alg. \ref{alg:solver_prob2}; error bars indicate the $95\%$ confidence interval based on $20$ measurements.}
  \label{fig:cmdp_times}
\end{figure}
\subsection{Numerical Evaluation}
We evaluate Algs. \ref{alg:solver_prob1}--\ref{alg:solver_prob2} on instantiations of Probs. \ref{prob_1}--\ref{prob_2} with different values of $\Delta_{\mathrm{R}}$ (\ref{eq:recovery_constraint}). Hyperparameters are listed in Appendix \ref{appendix:hyperparameters}. The computing environment for the evaluation is a server with a $24$-core \textsc{intel} \textsc{xeon} \textsc{gold} \small $2.10$ GHz \normalsize \textsc{cpu} and $768$ \textsc{gb} \textsc{ram}.

For each instantiation of Prob. \ref{prob_1} (optimal intrusion recovery), we run Alg. \ref{alg:solver_prob1} with four optimization algorithms: Simultaneous Perturbation Stochastic Approximation (\spsa) \cite[Fig. 1]{spsa_impl}, Bayesian Optimization (\bo) \cite[Alg. 1]{bayesian_opt}, Cross Entropy Method (\cem) \cite[Alg. 1]{moss2020crossentropy}, and Differential Evolution (\de) \cite[Fig. 3]{differential_evolution}. We compare the results with that of two baselines: Incremental Pruning (\ipp) \cite[Fig. 4]{incremental_pruning_pomdp}, which is a dynamic programming algorithm, and Proximal Policy Optimization (\ppo) \cite[Alg. 1]{ppo}, a reinforcement learning algorithm. The results are shown in Table \ref{tab:results_prob_1} and Figs. \ref{fig:recovery_curves}--\ref{fig:cmdp_times}.

We observe in the first three rows of Table \ref{tab:results_prob_1} that most of the algorithms that utilize Thm. \ref{thm:stopping_policy} find near-optimal recovery strategies for all $\Delta_{\mathrm{R}}$. By contrast, \ipp becomes computationally intractable as $\Delta_{\mathrm{R}} \rightarrow \infty$ (bottom row of Table \ref{tab:results_prob_1}).

The convergence times are shown in Figs. \ref{fig:recovery_curves}--\ref{fig:pomdp_times}. We observe that \cem, \bo, \de, and \ppo find near-optimal strategies within an hour of computation, whereas \spsa does not converge. This is probably due to a poor selection of hyperparameters.

Lastly, Fig. \ref{fig:cmdp_times} shows the performance of Alg. \ref{alg:solver_prob2}. We see that Alg. \ref{alg:solver_prob2} solves Prob. \ref{prob_2} (optimal replication factor) within $2$ minutes for systems with up to $2048$ nodes.

\section{Testbed Implementation of \tolerancee}\label{sec:implementation}
We implement \tolerancee as a proof-of-concept on a testbed. The implementation includes three layers.
\begin{table}
  \centering
\resizebox{1\columnwidth}{!}{%
\begin{tabular}{lll} \toprule
  {\textit{Server}} & {\textit{Processors}} & {\textit{\textsc{ram} (\textsc{gb})}} \\ \midrule
  $1$, \textsc{r}\footnotesize 715 2\footnotesize\textsc{u} & two $12$-core \textsc{amd} \textsc{opteron} & $64$\\
  $2$, \textsc{r}\footnotesize 715 2\footnotesize\textsc{u}  & two $12$-core \textsc{amd} \textsc{opteron} & $64$\\
  $3$, \textsc{r}\footnotesize 715 2\footnotesize\textsc{u}  & two $12$-core \textsc{amd} \textsc{opteron} & $64$\\
  $4$, \textsc{r}\footnotesize 715 2\footnotesize\textsc{u}  & two $12$-core \textsc{amd} \textsc{opteron} & $64$\\
  $5$, \textsc{r}\footnotesize 715 2\footnotesize\textsc{u}  & two $12$-core \textsc{amd} \textsc{opteron} & $64$\\
  $6$, \textsc{r}\footnotesize 715 2\footnotesize\textsc{u}  & two $12$-core \textsc{amd} \textsc{opteron} & $64$\\
  $7$, \textsc{r}\footnotesize 715 2\footnotesize\textsc{u}  & two $12$-core \textsc{amd} \textsc{opteron} & $64$\\
  $8$, \textsc{r}\footnotesize 715 2\footnotesize\textsc{u}  & two $12$-core \textsc{amd} \textsc{opteron} & $64$\\
  $9$, \textsc{r}\footnotesize 715 2\footnotesize\textsc{u}  & two $12$-core \textsc{amd} \textsc{opteron} & $64$\\
  $10$, \textsc{r}\footnotesize 630 2\footnotesize\textsc{u}  & two $12$-core \textsc{intel xeon e}\footnotesize $5$\normalsize\textsc{-}\footnotesize $2680$\normalsize & $256$\\
  $11$, \textsc{r}\footnotesize 740 2\footnotesize\textsc{u}  & $1$ $20$-core \textsc{intel xeon gold}\footnotesize $5218$\normalsize\textsc{r} & $32$\\
  $12$, \textsc{supermicro} \footnotesize 7049 \normalsize  & $2$ \textsc{tesla} \textsc{p}\footnotesize $100$, $1$ $16$-core \textsc{intel xeon} & $126$\\
  $13$, \textsc{supermicro} \footnotesize 7049 \normalsize  & $4$ \textsc{rtx} \footnotesize $8000$, $1$ $24$-core \textsc{intel xeon} & $768$\\
  \bottomrule\\
\end{tabular}}
\caption{Specifications of the physical nodes.}\label{tab:compute_cluster}
\end{table}
\subsection{The Physical Layer}\label{sec:physical_layer}
The physical layer contains a cluster with $13$ nodes connected through an Ethernet network. Specifications of the nodes can be found in Table \ref{tab:compute_cluster}.

Nodes communicate via message passing over authenticated channels. Each node runs (\textit{i}) a service replica in a Docker container \cite{docker}; (\textit{ii}) a node controller (\S \ref{sec:system_arch}); and (\textit{iii}) the \snort \ids with ruleset v2.9.17.1 \cite{snort}.
\subsection{The Virtualization Layer}\label{sec:virt_layer}
Each service replica runs a web service \cite{state_machine_reference}. The service offers two deterministic operations: (\textit{i}) a \textit{read operation}, which returns the current state of the service; and (\textit{ii}) a \textit{write operation}, which updates the state. To coordinate these operations, replicas run reconfigurable \minbft \cite[\S 4.2]{giuliana_thesis}. The throughput of our implementation of \minbft is shown in Fig. \ref{fig:tolerance_24}. The source code of our implementation is available at \cite{csle_docs,supplementary} and a description of \minbft is available in Appendix \ref{appendix_minbft}.

\begin{figure}
  \centering
  \scalebox{0.9}{
    \input{tikz/tolerance_24.tex}
  }
  \caption{Average throughput of our implementation of \minbft; error bars indicate the $95\%$ confidence interval based on $1000$ samples.}
  \label{fig:tolerance_24}
\end{figure}
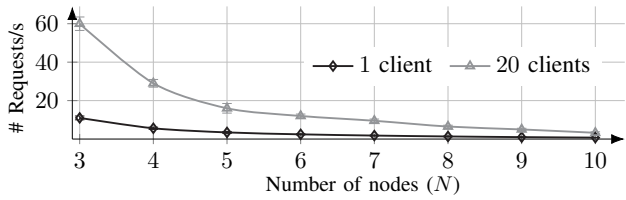

Clients access the service by issuing requests that are sent to all replicas. Each request has a unique identifier that is digitally signed. After sending a request, the client waits for a quorum of $f+1$ identical replies with valid signatures \cite{giuliana_thesis}. (Remark: a quorum is necessary to guarantee that the response is correct since the client does not know which replicas are compromised (Prop. \ref{prop:correctness}).)
\subsection{The Control Layer}\label{sec:controllers_21}
Node controllers collect \ids alerts and decide when to recover service replicas. When a replica is recovered, it starts with a new container and its state is initialized with the (identical) state from $f+1$ other replicas \cite{lazarus,4365686}.

The system controller is implemented by a crash-tolerant system that runs the \raft protocol \cite{raft}. When it decides to evict or add a node, it triggers a view change in \minbft. Figure \ref{fig:strategy_structure}.a illustrates the strategy of the system controller.
\section{Evaluation of \tolerancee}\label{sec:solution_approach}
In this section, we evaluate our implementation of \tolerancee and compare it with state-of-the-art intrusion-tolerant systems.

\subsection{Evaluation Setup}\label{sec:evaluation_setup}
\begin{figure*}
\centering
\scalebox{0.65}{
      \includegraphics{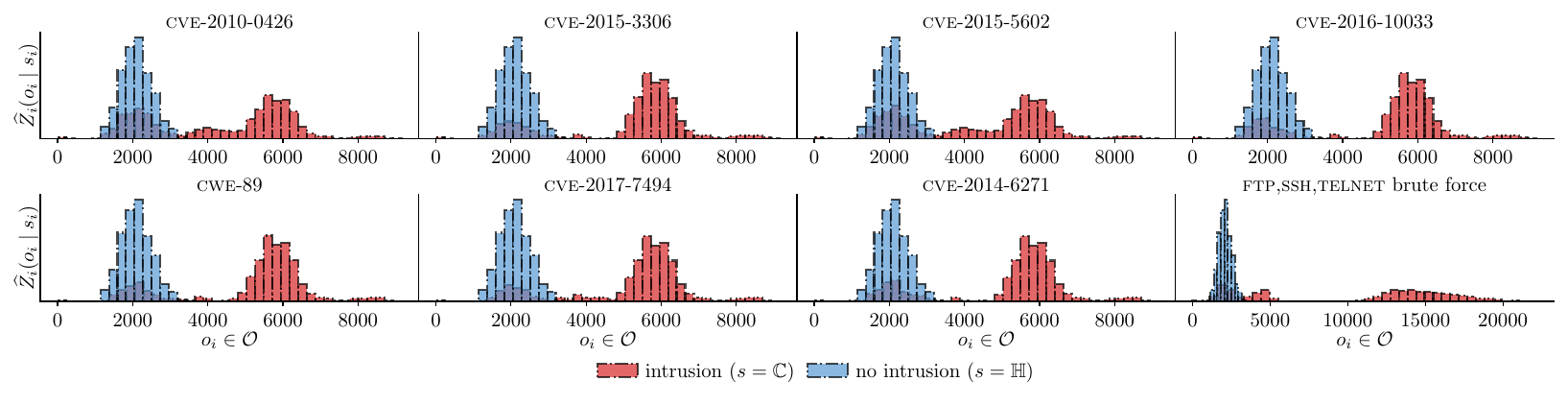}
}
\caption{Empirical distributions $\widehat{Z}_1(\cdot \mid s_i), \hdots, \widehat{Z}_{10}(\cdot \mid s_i)$ as estimates of $Z_{1}, \hdots, Z_{10}$ (\ref{eq:obs_function}) for the containers in Table \ref{tab:replica_types}.}\label{fig:tolerance_obs_dists}
\end{figure*}
An evaluation run evolves in time-steps of $60$ seconds. It starts with $N_1$ nodes from Table \ref{tab:compute_cluster}, each of which runs a service replica (see Table \ref{tab:replica_types}). At each time-step, one or more replicas may be recovered by the node controllers and a new node may be added by the system controller. When a replica is recovered, its container is replaced with a container selected randomly from the list in Table \ref{tab:replica_types}. Similarly, when a new node is added, a node from the list in Table \ref{tab:compute_cluster} is started.

Each replica has one or more vulnerabilities that can be exploited by the attacker using the steps listed in Table \ref{tab:attacker_actions}. After compromising a replica, the attacker randomly chooses between: a) participating in the consensus protocol; b) not participating; and c) participating with randomly selected messages.

Replicas are interconnected through Gbit/s connections with $0.05\%$ packet loss (emulated with \netem \cite{netem}). They receive a stream of service requests, which are sent by a client over $100$ Mbit/s connections with $0.1\%$ packet loss.

To emulate \ids events for a realistic system, each replica runs a set of background services in addition to the replicated service (see Table \ref{tab:client_profiles}). These background services are consumed by a population of background clients, who arrive with a Poisson rate $\lambda=20$ and have exponentially distributed service times with mean $\mu=4$ time-steps.

All parameters for the evaluation are listed in Appendix \ref{appendix:hyperparameters} except for $Z_i$ (\ref{eq:obs_function}), which we estimate with the empirical distribution $\widehat{Z}_i$ (see Fig. \ref{fig:tolerance_obs_dists}). We compute $\widehat{Z}_i$ based on $M=25,000$ samples \cite{supplementary,csle_docs}, knowing that $\widehat{Z}_i \overset{\text{a.s.}}{\rightarrow} Z_i$ as $M \rightarrow \infty$ (Glivenko-Cantelli theorem).

In practice, $\widehat{Z}_i$ may be implemented using any statistical intrusion detection method (e.g., anomaly detection \cite{int_prevention}). Similarly, the model parameters (e.g., the probability that a node crashes) can be defined based on domain knowledge or based on system measurements \cite{google_failure_model,facebook_faulure_model}.

\begin{table}
\centering
\begin{tabular}{lll} \toprule
  {\textit{Replica ID}} & {\textit{Operating system}} & {\textit{Vulnerabilities}} \\ \midrule
  $1$  & \ubuntu 14 & \ftp weak password\\
  $2$  & \ubuntu 20 & \ssh weak password\\
  $3$  & \ubuntu 20 & \telnet weak password\\
  $4$  & \debian 10.2 & \cve-2017-7494\\
  $5$  & \ubuntu 20 & \cve-2014-6271\\
  $6$  & \debian 10.2 & \cwe-89 on \dvwa \cite{dvwa}\\
  $7$  & \debian 10.2 & \cve-2015-3306\\
  $8$  & \debian 10.2 & \cve-2016-10033\\
  $9$  & \debian 10.2 & \cve-2010-0426, \ssh weak password\\
  $10$  & \debian 10.2 & \cve-2015-5602, \ssh weak password\\
  \bottomrule\\
\end{tabular}
\caption{Containers running the service replicas.}\label{tab:replica_types}
\end{table}

\begin{table}
\centering
\begin{tabular}{ll} \toprule
  {\textit{Background services}} & {\textit{Replica ID(s)}} \\ \midrule
  \ftp, \ssh, \textsc{mongodb}, \http, \textsc{teamspeak} & $1$\\
  \ssh, \dns, \http & $2$\\
  \ssh, \telnet, \http & $3$\\
  \ssh, \textsc{samba}, \textsc{ntp} & $4$\\
  \ssh & $5,7,8, 10$\\
  \dvwa, \irc, \ssh & $6$\\
  \textsc{teamspeak}, \http, \ssh & $9$\\
  \bottomrule\\
\end{tabular}
\caption{Background services of the service replicas.}\label{tab:client_profiles}
\end{table}

\begin{table}
\centering
\begin{tabular}{ll} \toprule
  {\textit{Replica ID}} & {\textit{Intrusion steps}} \\ \midrule
  $1$  & \tcpp \syn scan, \ftp brute force\\
  $2$  & \tcpp \syn scan, \ssh brute force\\
  $3$  & \tcpp \syn scan, \telnet brute force\\
  $4$  & \icmp scan, exploit of \cve-2017-7494\\
  $5$  & \icmp scan, exploit of \cve-2014-6271\\
  $6$  & \icmp scan, exploit of \cwe-89 on on \dvwa \cite{dvwa}\\
  $7$  & \icmp scan, exploit of \cve-2015-3306\\
  $8$  & \icmp scan, exploit of \cve-2016-10033\\
  $9$  & \icmp scan, \ssh brute force, exploit of \cve-2010-0426\\
  $10$  & \icmp scan, \ssh brute force, exploit of \cve-2015-5602\\
  \bottomrule\\
\end{tabular}
\caption{Steps to compromise service replicas.}\label{tab:attacker_actions}
\end{table}
\begin{table*}
\resizebox{1\textwidth}{!}{%
  \centering
  \begin{tabular}{l|lll|lll|lll}
    \toprule
    \multirow{2}{*}{Control strategy} &
     \multicolumn{3}{c|}{$\Delta_{\mathrm{R}}=15$} &
     \multicolumn{3}{c|}{$\Delta_{\mathrm{R}}=25$} &
     \multicolumn{3}{c}{$\Delta_{\mathrm{R}}=\infty$}\\
      & {$T^{(\mathrm{A})}$} & {$T^{(\mathrm{R})}$} & {$F^{(\mathrm{R})}$} & {$T^{(\mathrm{A})}$} & {$T^{(\mathrm{R})}$} & {$F^{(\mathrm{R})}$} & {$T^{(\mathrm{A})}$} & {$T^{(\mathrm{R})}$} & {$F^{(\mathrm{R})}$}\\
    \midrule
    \multicolumn{9}{l}{$\quad\quad\quad$$\quad\quad\quad\quad$$\quad\quad\quad\quad$$\quad\quad\quad\quad$$\quad\quad\quad\quad$$\quad\quad\quad\quad$$\quad\quad\quad\quad$$\quad\quad\quad\quad$$\quad\quad\quad\quad$$\quad\quad\quad\quad\quad\quad N_1=3$}\\
    \midrule
    \tolerancee & $\bm{0.99 \pm 0.01}$ & $\bm{1.43 \pm 0.09}$& $0.09 \pm 0.01$ & $\bm{0.99 \pm 0.01}$ & $\bm{1.43 \pm 0.09}$ & $0.09 \pm 0.01$ & $\bm{0.99 \pm 0.01}$ & $\bm{1.43 \pm 0.09}$& $0.09 \pm 0.01$\\
    \textsc{no-recovery} & $0.08 \pm 0.06$ & $10^3 \pm 0.00$& $\bm{0.00 \pm 0.00}$ & $0.08 \pm 0.06$ & $10^3 \pm 0.00$ & $\bm{0.00 \pm 0.00}$ & $0.08 \pm 0.06$ & $10^3 \pm 0.00$& $\bm{0.00 \pm 0.00}$\\
    \textsc{periodic} & $0.97 \pm 0.01$ & $6.06 \pm 1.16$& $0.065 \pm 0.01$ & $0.93 \pm 0.01$ & $8.64 \pm 1.48$ & $0.04 \pm 0.01$ & $0.08 \pm 0.06$ & $10^3 \pm 0.00$& $\bm{0.00 \pm 0.00}$\\
    \textsc{periodic-adaptive} & $0.95 \pm 0.02$ & $5.42 \pm 0.93$& $0.05 \pm 0.01$ & $0.94 \pm 0.02$ & $6.57 \pm 1.01$ & $0.03 \pm 0.01$ & $0.09 \pm 0.04$ & $10^3 \pm 0.00$& $\bm{0.00 \pm 0.00}$\\
    \midrule
    \multicolumn{9}{l}{$\quad\quad\quad$$\quad\quad\quad\quad$$\quad\quad\quad\quad$$\quad\quad\quad\quad$$\quad\quad\quad\quad$$\quad\quad\quad\quad$$\quad\quad\quad\quad$$\quad\quad\quad\quad$$\quad\quad\quad\quad$$\quad\quad\quad\quad\quad\quad N_1=6$}\\
    \midrule
    \tolerancee & $\bm{0.99 \pm 0.01}$ & $\bm{1.47 \pm 0.07}$& $0.07 \pm 0.01$ & $\bm{0.99 \pm 0.01}$ & $\bm{1.47 \pm 0.07}$ & $0.07 \pm 0.01$ & $\bm{0.99 \pm 0.01}$ & $\bm{1.47 \pm 0.07}$ & $0.07 \pm 0.01$\\
    \textsc{no-recovery} & $0.16 \pm 0.06$ & $10^3 \pm 0.00$& $\bm{0.00 \pm 0.00}$ & $0.16 \pm 0.06$ & $10^3 \pm 0.00$ & $\bm{0.00 \pm 0.00}$ & $0.16 \pm 0.06$ & $10^3 \pm 0.00$& $\bm{0.00 \pm 0.00}$\\
    \textsc{periodic} & $0.98 \pm 0.01$ & $5.96 \pm 1.16$& $0.065 \pm 0.01$ & $0.95 \pm 0.02$ & $8.13 \pm 1.48$ & $0.04 \pm 0.01$ & $0.16 \pm 0.03$ & $10^3 \pm 0.00$& $\bm{0.00 \pm 0.00}$\\
    \textsc{periodic-adaptive} & $\bm{0.99 \pm 0.01}$ & $5.02 \pm 0.34$& $0.06 \pm 0.01$ & $0.97 \pm 0.02$ & $6.16 \pm 0.54$ & $0.03 \pm 0.01$ & $0.17 \pm 0.03$ & $10^3 \pm 0.00$& $\bm{0.00 \pm 0.00}$\\
    \midrule
    \multicolumn{9}{l}{$\quad\quad\quad$$\quad\quad\quad\quad$$\quad\quad\quad\quad$$\quad\quad\quad\quad$$\quad\quad\quad\quad$$\quad\quad\quad\quad$$\quad\quad\quad\quad$$\quad\quad\quad\quad$$\quad\quad\quad\quad$$\quad\quad\quad\quad\quad\quad N_1=9$}\\
    \midrule
    \tolerancee & $\bm{1.00 \pm 0.00}$ & $\bm{1.44 \pm 0.05}$& $0.07 \pm 0.01$ & $\bm{1.00 \pm 0.00}$ & $\bm{1.44 \pm 0.05}$ & $0.07 \pm 0.01$ & $\bm{1.00 \pm 0.00}$ & $\bm{1.44 \pm 0.05}$& $0.07 \pm 0.01$\\
    \textsc{no-recovery} & $0.17 \pm 0.04$ & $10^3 \pm 0.00$& $\bm{0.00 \pm 0.00}$ & $0.17 \pm 0.04$ & $10^3 \pm 0.00$ & $\bm{0.00 \pm 0.00}$ & $0.17 \pm 0.04$ & $10^3 \pm 0.00$& $\bm{0.00 \pm 0.00}$\\
    \textsc{periodic} & $0.99 \pm 0.00$ & $5.37 \pm 0.34$& $0.06 \pm 0.01$ & $0.98 \pm 0.01$ & $7.74 \pm 0.51$ & $0.04 \pm 0.01$ & $0.17 \pm 0.04$ & $10^3 \pm 0.00$& $\bm{0.00 \pm 0.00}$\\
    \textsc{periodic-adaptive} & $\bm{1.00 \pm 0.00}$ & $4.44 \pm 0.25$& $0.06 \pm 0.01$ & $0.99 \pm 0.01$ & $6.01 \pm 0.39$ & $0.04 \pm 0.01$ & $0.18 \pm 0.02$ & $10^3 \pm 0.00$ & $\bm{0.00 \pm 0.00}$\\
    \bottomrule
  \end{tabular}}
  \caption{Comparison between \tolerancee and the baselines (\S \ref{sec:baselines}); columns indicate values of $\Delta_{\mathrm{R}}$; subcolumns represent performance metrics; row groups relate to the number of initial nodes $N_1$; numbers indicate the mean and the $95\%$ confidence interval from evaluations with $20$ random seeds.}\label{tab:stota_eval_results}
\end{table*}
\begin{figure*}
  \centering
  \scalebox{0.93}{
    \input{tikz/tolerance_41.tex}
  }
  \caption{Comparison between \tolerancee and the baselines (\S \ref{sec:baselines}); columns represent performance metrics; x-axes indicate values of $\Delta_{\mathrm{R}}$; rows relate to the number of initial nodes $N_1$; bars show the mean and error bars indicate the $95\%$ confidence interval from evaluations with $20$ random seeds.}
  \label{fig:tolerance_41}
\end{figure*}
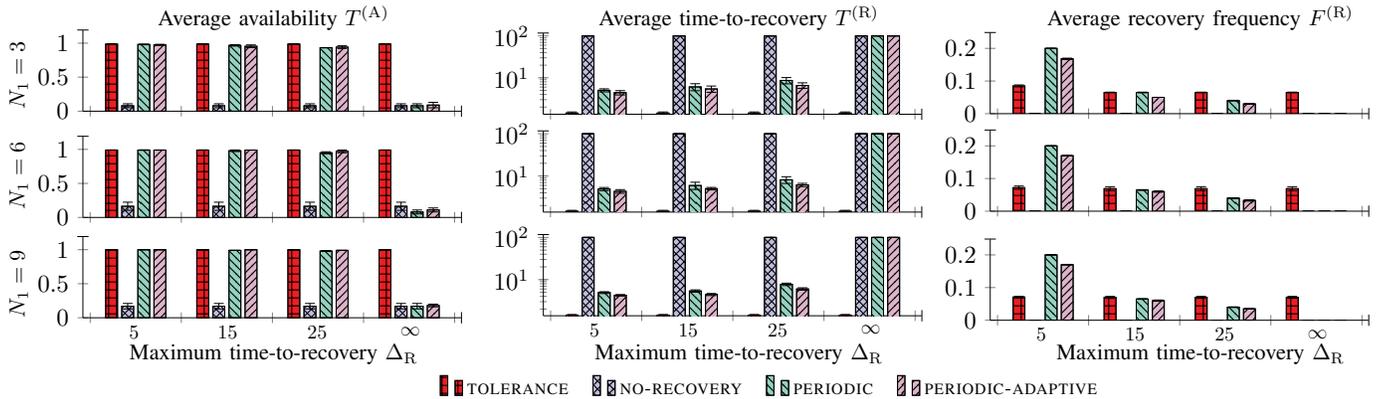
\begin{figure}
  \centering
  \scalebox{1.08}{
    \input{tikz/strategy_structure.tex}
  }
  \caption{Illustration of the replication and recovery strategies for $\Delta_{\mathrm{R}}=\infty$, $N_1=6$, and $f=1$.}
  \label{fig:strategy_structure}
\end{figure}
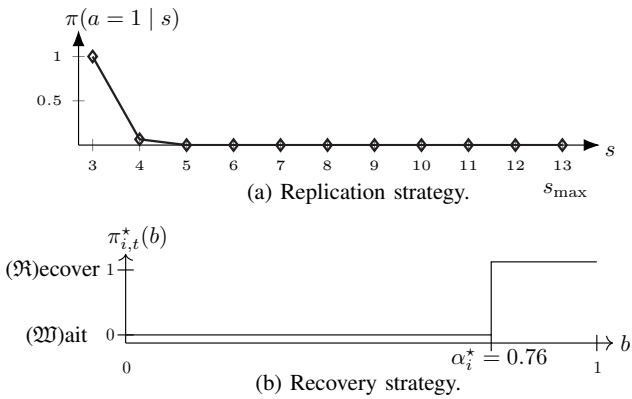
\begin{figure}
  \centering
  \scalebox{0.9}{
    \input{tikz/kl_1.tex}
  }
 \caption{Optimal recovery cost (\ref{eq:objective_recovery}) in function of the accuracy of the intrusion detection model $Z_i$ (\ref{eq:obs_function}); $D_{\text{\textsc{kl}}}$ refers to the Kullback-Leibler (\kl) divergence \cite{kl_divergence}.}\label{fig:kls}
\end{figure}
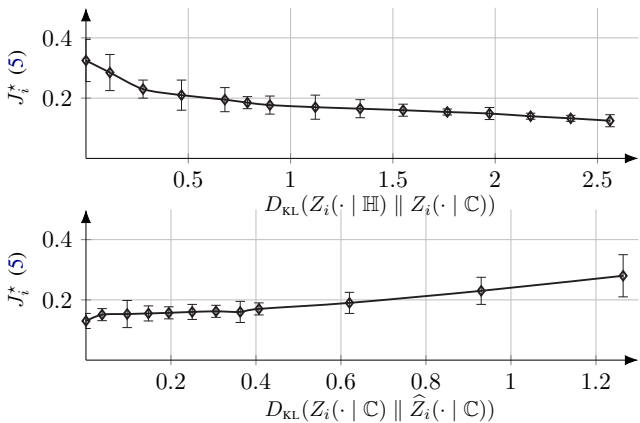
\subsection{Baseline Control Strategies}\label{sec:baselines}
We compare the control strategies of \tolerancee with those used in current intrusion-tolerant systems, whereby we choose three baseline strategies: \textsc{no-recovery}, \textsc{periodic} and \textsc{periodic-adaptive}. The first baseline, \textsc{no-recovery}, does not recover or add any nodes, which corresponds to the strategy used in traditional intrusion-tolerant systems, such as \rampart \cite{rampart} and \securering \cite{sercurering}. The second baseline, \textsc{periodic}, recovers nodes every $\Delta_{\mathrm{R}}$ time-steps but does not add any new nodes. This is the strategy used in most of the intrusion-tolerant systems proposed in prior work, including \pbft \cite{pbft}, \vmfit \cite{virt_replica, 4365686}, \wormit \cite{wormit}, \prrw \cite{4459685,5010435}, \maftia \cite{1311915}, \recover \cite{forever_service}, \scit \cite{scit, 5958797}, \coca \cite{coca}, \spire \cite{spire}, \itcisprr \cite{itcis_prr}, \crutial \cite{crutial}, \sbft \cite{sbft}, \bft-\smart \cite{bft_smart}, \uprightt \cite{upright}, and \skynet \cite{skynet}. The third baseline, \textsc{periodic-adaptive}, recovers nodes every $\Delta_{\mathrm{R}}$ time-steps and adds a node when $o_{i,t} \geq 2\mathbb{E}[O_t]$ (\ref{eq:obs_function}), which approximates the heuristic strategies used in \cite{4429182}, \sitar \cite{sitar}, \itsi \cite{itsi}, and \itua \cite{itua,1209971}.

\subsection{Evaluation Results}\label{sec:results}
The results are summarized in Fig. \ref{fig:tolerance_41} and Table \ref{tab:stota_eval_results}. The control strategies are illustrated in Fig. \ref{fig:strategy_structure} and the sensitivity of the controllers with respect to the intrusion detection model $Z_i$ (\ref{eq:obs_function}) is shown in Fig. \ref{fig:kls}.

The red bars in Fig. \ref{fig:tolerance_41} relate to \tolerancee. The blue, green, and pink bars relate to the baselines. The leftmost column in Fig. \ref{fig:tolerance_41} shows the average availability for different values of $\Delta_{\mathrm{R}}$. We observe that \tolerancee has close to $100$\% service availability in all of the cases we studied. By contrast, \textsc{no-recovery} has close to $0$\% availability. The availability achieved by \textsc{periodic} and \textsc{periodic-adaptive} is in-between; they perform similar to \tolerancee when $\Delta_{\mathrm{R}}$ is small (i.e., when recoveries are frequent) and similar to \textsc{no-recovery} when $\Delta_{\mathrm{R}} \rightarrow \infty$. We note that increasing $N_1$ from $3$ to $9$ doubles the availability of \textsc{no-recovery} but has a negligible impact on the availability of \tolerancee, \textsc{periodic}, and \textsc{periodic-adaptive}.

The second leftmost column in Fig. \ref{fig:tolerance_41} shows the average time-to-recovery $T^{(\mathrm{R})}$. We observe that $T^{(\mathrm{R})}$ of \tolerancee is an order of magnitude smaller than that of \textsc{periodic} and \textsc{periodic-adaptive} and two orders of magnitude smaller than that of \textsc{no-recovery}. This result demonstrates the benefit of feedback control, which allows the system to react promptly to intrusions.

Finally, the rightmost column in Fig. \ref{fig:tolerance_41} shows the average frequency of recoveries $F^{(\mathrm{R})}$. We note that $F^{(\mathrm{R})}$ of \tolerancee is about the same as \textsc{periodic} and \textsc{periodic-adaptive} when $\Delta_{\mathrm{R}}=15$. Interestingly, when $\Delta_{\mathrm{R}}=5$, \tolerancee has both a smaller $F^{(\mathrm{R})}$ and a smaller $T^{(\mathrm{R})}$ than \textsc{periodic} and \textsc{periodic-adaptive}.

\subsection{Discussion of the Evaluation Results}\label{sec:discussion}
The key findings from our evaluation of the \tolerancee architecture can be summarized as follows:
\begin{enumerate}[(i)]
\item \tolerancee can achieve a lower time-to-recovery and a higher service availability than state-of-the-art intrusion-tolerant systems (Table \ref{tab:stota_eval_results}, Fig. \ref{fig:tolerance_41}). The \btr constraint (\ref{eq:recovery_constraint}) guarantees that \tolerancee never has a worse time-to-recovery than current systems. The performance of \tolerancee depends on the accuracy of the intrusion detection model $\widehat{Z}_i$ (see Fig. \ref{fig:kls} and Fig. \ref{fig:tolerance_obs_dists}).
\item The solutions to both control problems of the \tolerancee architecture have threshold properties (Thms. \ref{thm:stopping_policy}--\ref{thm:structure_respone}, Cor. \ref{corr:increasing_thresholds}), which enable efficient computation of optimal strategies (Figs. \ref{fig:recovery_curves}--\ref{fig:cmdp_times}, Algs. \ref{alg:solver_prob1}--\ref{alg:solver_prob2}).
\item The benefit of using an adaptive replication strategy as opposed to a static strategy is mainly prominent when node crashes are frequent (see Fig. \ref{fig:strategy_structure} and cf. the results of \textsc{periodic} and \textsc{periodic-adaptive} in Fig. \ref{fig:tolerance_41}.)
\end{enumerate}
While the results demonstrate clear benefits of \tolerancee compared to current intrusion-tolerant systems, \tolerancee has two drawbacks. First, the performance of \tolerancee depends on the accuracy of the intrusion detection model $\widehat{Z}_i$ (see Fig. \ref{fig:tolerance_obs_dists} and Fig. \ref{fig:kls}). This means that practical deployments of \tolerancee require a statistical intrusion detection model for estimating the probability of intrusion (\ref{eq:belief_def}). This model can be realized in many ways. It can for example be based on anomaly detection methods or machine learning techniques. Further, the detection model can use different types of data sources, e.g., log files, \ids alerts, threat intelligence sources, etc. Our proof-of-concept implementation of \tolerancee uses the \snort \ids as the data source and obtains the distribution of \ids alerts using maximum likelihood estimation, which allows to compute the probability of intrusion (\ref{eq:belief_def}).

Second, \tolerancee is vulnerable to an attack where a large amount of false \ids alerts trigger excess recoveries. Analysis of such attacks requires a game-theoretic treatment, whereby problems \ref{prob_1}--\ref{prob_2} are modified to take into account how an attacker may exploit the control strategies (e.g., minimax problem formulations \cite{kim_gamesec23}). We plan to investigate such problem formulations in future work (see \S \ref{sec:solution_approach}).
\section{Related Work}\label{sec:related_work}
Intrusion tolerance is studied in several broad areas of research, including: Byzantine fault tolerance \cite{bft_systems_survey}, dependability \cite{goyal1992unified,burns1991framework}, reliability \cite{birman1993process}, survivability \cite{Kreidl2004FeedbackCA}, and cyber resilience \cite{resilience_nature,10.1145/2602087.2602116,Kott_2022,kott_resilience_3,Ellis_2022,control_rl_reviews}. This research effort has led to many mechanisms for implementing intrusion-tolerant systems, such as: intrusion-tolerant consensus protocols \cite{bft_systems_survey,goyal1992unified,burns1991framework,birman1993process,7307998,hotstuff,sbft,prime,upright,bft_smart}, software diversification schemes \cite{lazarus}, geo-replication schemes \cite{9505127}, cryptographic mechanisms \cite{270414,6038579}, and defenses against denial of service \cite{8692706}. These mechanisms provide the foundation for \tolerancee, which adds automated recovery and replication control.

While \tolerancee builds on all of the above works, we limit the following discussion to explain how \tolerancee differs from current intrusion-tolerant systems and how it relates to prior work that uses feedback control.
\subsection{Intrusion-Tolerant Systems}
Existing intrusion-tolerant systems include \pbft \cite{pbft}, \zyzzyva \cite{zyzzyva}, \hq \cite{hq}, \hotstuff \cite{hotstuff}, \vmfit \cite{virt_replica, 4365686}, \wormit \cite{wormit}, \prrw \cite{4459685}, \recover \cite{forever_service}, \scit \cite{scit, 5958797}, \coca \cite{coca}, \cite{4429182}, \spire \cite{spire}, \itcisprr \cite{itcis_prr}, \crutial \cite{crutial}, \uprightt \cite{upright}, \bft-\smart \cite{bft_smart}, \sbft \cite{sbft}, \sitar \cite{sitar}, \itua \cite{itua,1209971}, \maftia \cite{1311915}, \itsi \cite{itsi}, and \skynet \cite{skynet}. All of them are based on intrusion-tolerant consensus protocols and support recovery, either directly or indirectly through external recovery services, like \phoenix \cite{phoenix}. \tolerancee differs from these systems in two main ways.

First, \tolerancee uses feedback control to decide when to perform intrusion recovery. This contrasts with all of the referenced systems, which either use periodic or heuristic recovery schemes. (\prrw, \recover, \crutial, \scit, \sitar, \itsi, and \itua can be implemented with feedback-based recovery but they do not specify how to implement such recovery strategies.)

Second, \tolerancee uses an adaptive replication strategy. In comparison, all of the referenced systems use static replication strategies except \sitar, \cite{4429182}, \itua, and \itsi, who implement adaptive replication based on time-outs and static rules as opposed to feedback control. The benefit of feedback control is that it allows the system to adapt promptly to intrusions, not having to wait for a time-out.
\subsection{Intrusion Response through Feedback Control}
Intrusion response through feedback control is an active area of research that uses concepts and methods from various emergent and traditional fields. Most notably from reinforcement learning (see examples \cite{control_rl_reviews,4725362,janisch2023nasimemu,hammar_stadler_cnsm_20,hammar_stadler_cnsm_21, hammar_stadler_cnsm_22,hammar_stadler_tnsm,hammar_stadler_tnsm_23,kim_gamesec23,9833086,kunz2023multiagent,ko2020cyber,foley2023inroads}), control theory (see examples \cite{feedback_control_computing_systems,Miehling_control_theoretic_approaches_summary,7011201,5542603,https://doi.org/10.1002/cplx.20011,Kreidl2004FeedbackCA,9505132,tifs_hlsz_extended,li2024conjectural}), causal inference (see example \cite{causal_neil_agent}), game theory (see examples \cite{nework_security_alpcan,tambe,carol_book_intrusion_detection,levente_book,9087864,5270307}), natural language processing (see example \cite{rigaki2023cage}), evolutionary computation (see example \cite{hemberg_oreily_evo}), and general optimization (see examples \cite{7127023,6514999}). While these works have obtained promising results, none of them consider the integration with intrusion-tolerant systems as we do in this paper. Another drawback of the existing solutions is that many of them are inefficient and lack safety guarantees. Finally, and most importantly, nearly all of the above works are limited to simulation environments and it is not clear how they generalize to practical systems. In contrast, \tolerancee is practical: it can be integrated with existing intrusion-tolerant systems, it satisfies safety constraints, and it is computationally efficient.

\section{Conclusion and Future Work}\label{sec:solution_approach} %
This paper presents \tolerancee: a novel control architecture for intrusion-tolerant systems that uses two levels of control to decide when to perform recovery and when to increase the replication factor. These control problems can be formulated as two classical problems in operations research, namely, the machine replacement problem and the inventory replenishment problem. Using this formulation, we prove that the optimal control strategies have threshold structure (Thms. \ref{thm:stopping_policy}--\ref{thm:structure_respone}, Cor. \ref{corr:increasing_thresholds}) and we design efficient algorithms for computing them (Algs. \ref{alg:solver_prob1}--\ref{alg:solver_prob2}, Table \ref{tab:results_prob_1}). We evaluate \tolerancee in an emulation environment where we run $10$ types of network intrusions. The results demonstrate that \tolerancee improves service availability and reduces operational cost when compared with state-of-the-art intrusion-tolerant systems in the scenarios we studied (Fig. \ref{fig:tolerance_41}, Table \ref{tab:stota_eval_results}). The improvement of \tolerancee with respect to current systems comes at the expense of a training phase, where we first fit an intrusion detection model (Fig. \ref{fig:tolerance_obs_dists}) and then train the controllers based on this model (Figs. \ref{fig:recovery_curves}--\ref{fig:cmdp_times}).

We plan to continue this work in several directions. First, we will improve our implementation by integrating a high-performance consensus protocol, e.g., \hotstuff-\textsc{m} \cite{ittai_hybrid_failures}. Second, we intend to extend our control-theoretic model of the \tolerancee architecture to a game-theoretic model, which allows us to study defenses against dynamic attackers. Third, we plan to investigate methods for online learning of intrusion detection models and online adaptation of control strategies.

\section{Acknowledgments}
The authors would like to thank Forough Shahab Samani, Xiaoxuan Wang, and Duc Huy Le for their constructive comments on a draft of this paper.

\appendices

\section{Belief Computation}\label{appendix_belief_comp}
The belief state for each node $i \in \mathcal{N}_t$ can be computed recursively as
\begin{align*}
&b_{i,t}(s_{i,t}) \numeq{a} \mathbb{P}\big[s_{i,t} \mid \overbrace{o_{i,t},a_{i,t-1},o_{i,t-1},\hdots,a_{i,1}, o_{i,1}, b_{i,1}}^{\mathbf{h}_{i,t}}\big]\\
  &\numeq{b}\frac{\mathbb{P}\big[o_{i,t} \mid s_{i,t},a_{i,t-1},\mathbf{h}_{i,t-1}]\mathbb{P}[s_{i,t} \mid a_{i,t-1},\mathbf{h}_{i,t-1}\big]}{\mathbb{P}\big[o_{i,t} \mid a_{i,t-1},\mathbf{h}_{i,t-1}\big]}\nonumber\\
  &\numeq{c}\frac{Z_i(o_{i,t} \mid s_{i,t})\mathbb{P}[s_{i,t} \mid a_{i,t-1},\mathbf{h}_{i,t-1}\big]}{\mathbb{P}\big[o_{i,t} \mid a_{i,t-1}\mathbf{h}_{i,t-1}\big]}\nonumber\\
  &\numeq{d}\frac{Z_i(o_{i,t} \mid s_{i,t})\sum_{s \in \mathcal{S}}b_{i,t-1}(s)f_{\mathrm{N},i}(s_{i,t} \mid s, a_{t-1})}{\mathbb{P}\big[o_{i,t} \mid a_{i,t-1}\mathbf{h}_{i,t-1}\big]}\nonumber\\
  &\numeq{e}\frac{Z_i(o_{i,t} \mid s_{i,t})\sum_{s \in \mathcal{S}}b_{i,t-1}(s)f_{\mathrm{N},i}(s_{i,t} \mid s, a_{t-1})}{\sum_{s^{\prime},s\in \mathcal{S}}Z_i(o_{i,t}\mid s^{\prime})f_{\mathrm{N},i}(s^{\prime} \mid s, a_{i,t-1})b_{i,t-1}(s)}.\nonumber
\end{align*}
(a) follows from the definition of $b_{i,t}$ (\ref{eq:belief_def}); (b) is an expansion of the conditional probability using Bayes rule; (c) follows from the Markov property of $Z_i$ (\ref{eq:obs_function}); (d) follows from the Markov property of $f_{\mathrm{N},i}$ (\ref{eq:recovery_dynamics}); and (e) follows by definition of $Z_i$ (\ref{eq:obs_function}) and $f_{\mathrm{N},i}$ (\ref{eq:recovery_dynamics}). Computing the belief state through the expression in (e) requires $O(|\mathcal{S}|^2)$ scalar multiplications.

\section{Proof of Theorem \ref{thm:stopping_policy}}\label{appendix:theorem_1}
Solving (\ref{eq:recovery_problem}) corresponds to solving $N_t$ finite, stationary, and constrained Partially Observed Markov Decision Processes (\pomdp{}s) with bounded costs and the average cost optimality criterion. Since the \pomdp{}s are equivalent except for the parameters $p_{\mathrm{A},i},p_{\mathrm{C}_1,i},p_{\mathrm{C}_2,i},p_{\mathrm{U},i}$ (\ref{eq:recovery_dynamics}) and the observation distribution $Z_{i}$ (\ref{eq:obs_function}) it suffices to prove the statement for a single \pomdp.

It follows from (\ref{eq:recovery_dynamics}) and assumption \ref{eq:assumption_A} that $f_{\mathrm{N},i}(s^{\prime} \mid s, a) > 0$ for all $t,s^{\prime},s,a$. Given this property and assumption \ref{eq:assumption_D}, we know that there exists a deterministic optimal strategy $\pi_{i,t}^{\star}$ for which the limit in (\ref{eq:objective_recovery}) exists and which satisfies
\begin{align}
&\pi_{i,t}^{\star}(b_{i,t}) \in \argmin_{a \in \{\mathfrak{W}, \mathfrak{R}\}}\bigg[c_{\mathrm{N}}(b_{i,t}, a) + \sum_{o \in \mathcal{O}}\mathbb{P}\left[o| b_{i,t}, a\right]V_{i,t}^{\star}(b_{i,t+1})\bigg]\label{eq:belief_bellman}
\end{align}
for all $b_{i,t}$ and $t \geq 1$ \cite[Prop. 1]{xiong2022sublinear}, where $c_{\mathrm{N}}(b_{i,t}, a)$ is the expected immediate cost of $a$ given $b_{i,t}$ and $V_{i,t}^{\star}$ is the value function \cite[Thm. 7.4.1]{krishnamurthy_2016}.

Each of the constrained \pomdp{}s with infinite horizons defined in (\ref{eq:recovery_problem}) can be converted into a sequence of unconstrained \pomdps $(\mathcal{M}_{i,k})_{k=1,2,\hdots}$ with horizon $T=\Delta_{\mathrm{R}}$, where $a_{i,T}=\mathfrak{R}$ ensures that (\ref{eq:recovery_constraint}) is satisfied. This sequence is equivalent to the original \pomdp because
\begin{align}
  &\argmin_{\pi_{i,t}}\left[\lim_{T\rightarrow \infty}\mathbb{E}_{\pi_{i,t}}\left[\frac{1}{T}\sum_{t=1}^{T}C_{i,t} \mid B_{i,1}=p_{\mathrm{A},i}\right]\right]\nonumber\\
  &\numeq{a}\argmin_{\pi_{i,t}}\Bigg[\lim_{T\rightarrow \infty}\frac{1}{T}\bigg(\mathbb{E}_{\pi_{i,t}}\bigg[\sum_{t=1}^{\Delta_{\mathrm{R}}}C_{i,t} \mid B_{i,1}=p_{\mathrm{A},i}\bigg] + \nonumber\\
  &\quad\quad\quad\quad\quad\mathbb{E}_{\pi_{i,t}}\bigg[\sum_{t=\Delta_{\mathrm{R}}}^{2\Delta_{\mathrm{R}}}C_{i,t} \mid B_{i,\Delta_{\mathrm{R}}}=p_{\mathrm{A},i}\bigg] + \hdots\bigg)\Bigg]\nonumber\\
  &\numeq{b}\argmin_{\pi_{i,t}}\left[\lim_{T\rightarrow \Delta_{\mathrm{R}}}\mathbb{E}_{\pi_{i,t}}\left[\frac{1}{T}\sum_{t=1}^{\Delta_{\mathrm{R}}}C_{i,t} \mid B_{i,1}=p_{\mathrm{A},i}\right]\right],\label{eq:decomposed_recovery_prob}
\end{align}
where $C_{i,t}$ is a random variable representing the cost of node $i$ at time $t$; (a) follows from linearity of $\mathbb{E}$; (b) follows because all elements inside the parentheses are equivalent, which means that a strategy that minimizes one element minimizes the whole expression.

Consider the threshold structure in (\ref{eq:threhsold_structure}). We know that a strategy $\pi_i^{\star}$ that achieves the minimization in (\ref{eq:decomposed_recovery_prob}) induces a partition of $[0,1]$ into two regions at each time $t$: a wait region $\mathcal{W}_t$ where $\pi_i^{\star}(b)=\mathfrak{W}$, and a recovery region $\mathcal{R}_t$ where $\pi_i^{\star}(b)=\mathfrak{R}$. The idea behind the proof of (\ref{eq:threhsold_structure}) is to show that $\mathcal{R}_t = [\alpha_{i,t}^{\star}, 1]$ for all $t$ and some thresholds $(\alpha_{i,t}^{\star})_{t=1,\hdots,T}$. Towards this deduction, note that $\mathcal{W}_t$ and $\mathcal{R}_t$ are connected sets \cite[Thm. 12.3.4]{krishnamurthy_2016}. This follows because (\textit{i}) the transition and observation matrices are \tpp \cite[Def. 10.2.1]{krishnamurthy_2016} (it is a consequence of assumptions \ref{eq:assumption_A}--\ref{eq:assumption_C}, and \ref{eq:assumption_E}); (\textit{ii}) $c_{\mathrm{N}}(s_{i,t},a_{i,t})$ (\ref{eq:objective_recovery}) is submodular \cite[Def. 12.3.2]{krishnamurthy_2016}; and (\textit{iii}) $c_{\mathrm{N}}(s_{i,t},a_{i,t})$ is weakly increasing in $s_{i,t}$ for each $a_{i,t}$.

As $\mathcal{R}_t$ is connected, $1 \in \mathcal{R}_t \iff \mathcal{R}_t=[\alpha_{i,t}^{\star}, 1]$. Hence it suffices to show that $1 \in \mathcal{R}_t$. We obtain from (\ref{eq:belief_bellman}) that
\begin{align*}
1 \in \mathcal{R}_t \iff&  \mathbb{E}_{B_{i,t+1}}\left[V_{i,t+1}^{\star}(B_{i,t+1}) \mid A_{i,t}=\mathfrak{R}, B_{i,t}=1\right] \leq\\
  &\mathbb{E}_{B_{i,t+1}}\left[V_{i,t+1}^{\star}(B_{i,t+1}) \mid A_{i,t}=\mathfrak{W}, B_{i,t}=1\right].
\end{align*}
Clearly
\begin{align*}
\mathbb{E}_{B_{i,t+1}}\left[B_{i,t+1} \mid A_{i,t}=\mathfrak{R}, B_{i,t}=1\right] \leq\\
 \mathbb{E}_{B_{i,t+1}}\left[B_{i,t+1} \mid A_{i,t}=\mathfrak{W}, B_{i,t}=1\right].
\end{align*}
Further $B^{\prime} \leq B \implies V_{i,t+1}^{\star}(B^{\prime}) \leq V_{i,t+1}^{\star}(B)$ for all $i \in \mathcal{N}_t$ and $t \geq 1$ \cite[Thm. 11.2.1]{krishnamurthy_2016}, which implies that $1 \in \mathcal{R}_t$ for all $t$ (see Fig. \ref{fig:thm1_illustration}). \qed

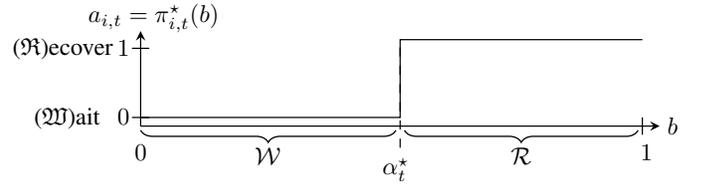
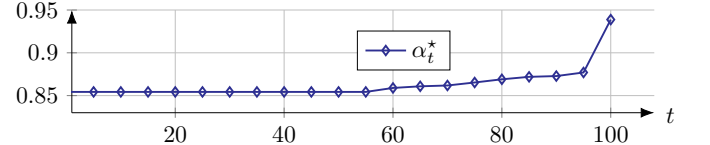
\begin{figure}
  \centering
  \begin{subfigure}[t]{1\columnwidth}
    \centering
    \scalebox{1.15}{
      \input{tikz/threshold_strategy.tex}
    }
    \caption{Structure of an optimal threshold recovery strategy $\pi^{\star}_{i,t}$; $\mathcal{W}$ and $\mathcal{R}$ denote the wait and recovery sets.}
    \label{fig:threshold_strategy_structure}
  \end{subfigure}
  \hfill
  \begin{subfigure}[t]{1\columnwidth}
    \centering
    \scalebox{0.9}{
      \input{tikz/thresholds.tex}
    }
    \caption{Optimal recovery thresholds $\alpha^{\star}_{i,t}$, where $\Delta_R=100$.}
    \label{fig:recovery_thresholds}
  \end{subfigure}
  \caption{Illustration of Thm. \ref{thm:stopping_policy} and Cor. \ref{corr:increasing_thresholds}; the parameters for computing the figures are listed in Appendix \ref{appendix:hyperparameters}.}\label{fig:thm1_illustration}
\end{figure}

\section{Proof of Corollary \ref{corr:increasing_thresholds}}\label{appendix:corollary_1}
When $\Delta_{\mathrm{R}} \rightarrow \infty$, (\ref{eq:recovery_problem}) reduces to a set of unconstrained, stationary, and finite \pomdp{}s, which means that there exists an optimal deterministic stationary strategy for each node \cite[Prop. 1]{xiong2022sublinear}. Such a strategy partitions the belief space into two time-independent regions $\mathcal{W}$ and $\mathcal{R}$, which means that the recovery threshold $\alpha_i^{\star}$ is time-independent (Thm. \ref{thm:stopping_policy}).

When $\Delta_{\mathrm{R}} < \infty$, it follows from (\ref{eq:belief_bellman}) that it is optimal to recover node $i \in \mathcal{N}_t$ at time $t$ iff
\begin{align*}
  & c_{\mathrm{N}}(b_{i,t}, \mathfrak{R}) + \mathbb{E}_{B_{i,t+1}}[V^{\star}_{i,t}(B_{i,t+1}) \mid a_{i,t}=\mathfrak{R}, b_{i,t}]  \leq\\
  &c_{\mathrm{N}}(b_{i,t}, \mathfrak{W}) + \mathbb{E}_{B_{i,t+1}}[V^{\star}_{i,t}(B_{i,t+1}) \mid a_{i,t}=\mathfrak{W},b_{i,t}]\\
&\iff 1 \leq \eta b_{i,t} + W_{i,t}(b_{i,t}) \iff b_{i,t} \geq \underbrace{\frac{1 -W_{i,t}(b_{i,t})}{\eta}}_{\alpha^{\star}_{i,t}},
\end{align*}
where $W_{i,t}(b_{i,t}) \triangleq \mathbb{E}_{B_{i,t+1}}[V_{i,t+1}^{\star}(B_{i,t+1}) \mid a_{i,t}=\mathfrak{W}, b_{i,t}] - \mathbb{E}_{B_{i,t+1}}[V^{\star}_{i,t+1}(B_{i,t+1}) \mid a_{i,t}=\mathfrak{R}, b_{i,t}]$

Hence $\alpha^{\star}_{i,t}\leq \alpha^{\star}_{i,t+1}$ iff $W_{i,t}$ is non-increasing in $t$ for all $b_{i,t}$ and $i$. We prove this using mathematical induction on $k=T,T-2,\hdots,1$. For $k=T-1$ we have $W_{i,T-1}(b_{i,T-1})=0$ for all $b$ and $i$. Next, for $k=T-2$ we have
\begin{align*}
  &W_{i,T-2}(b_{i,T-2}) = \min\big[1 + \mathbb{E}_{B_{i,T}}[V^{\star}_{i,T}(B_{i,T}) \mid a_{i,T-1}=\mathfrak{R}], \\
  &\mathbb{E}_{B_{i,T-1},B_{i,T}}[\eta B_{i,T-1} + V^{\star}_{i,T}(B_{i,T}) \mid a_{i,T-2}=a_{i, T-1}=\\
  & \quad\quad \mathfrak{W}, b_{T-2}]\big]-\min\big[1 -\mathbb{E}_{B_{i,T}}[V^{\star}_{i,T}(B_{i,T}) \mid a_{i,T-1}=\mathfrak{R}], \\
  &\mathbb{E}_{B_{i,T-1},B_{i,T}}[\eta B_{i,T-1} + V^{\star}_{i,T}(B_{i,T}) \\
   &\quad\quad\quad\quad\quad\quad\text{ }\mid a_{i,T-2}=\mathfrak{R},a_{i,T-1}=\mathfrak{W}, b_{i,T-2}]\big]\\
  &\numgeq{a} 0 = W_{i,T-1}(b_{i,T-2}),
\end{align*}
where (a) follows because $\mathbb{E}[B_{i,t+1} \mid a_{i,t}=\mathfrak{W}, b_{i,t}] \geq \mathbb{E}[B_{i,t} \mid a_{i,t}=\mathfrak{R}, b_{i,t}]$ by definition (\ref{eq:recovery_dynamics}).

Assume by induction that $W_{i,k}(b_i) \geq W_{i,k+1}(b_i)$ for $k=t,t+1,\hdots,T-3$ and all $b_i$ and $i$. We will show that this assumption implies $W_{i,k-1}(b_i) \geq W_{i,k}(b_i)$ for all $b_i$ and $i$.

There are three cases to consider:
\begin{enumerate}
\item If $B_{i,k} \in \mathcal{R}$ both when $a_{i,k-1} = \mathfrak{W}$ and when  $a_{i,k-1} = \mathfrak{R}$, then
\begin{align*}
&W_{i,k-1}(b_{i,k-1}) = \mathbb{E}_{B_{i,k}}[V_{i,k}^{\star}(B_{i,k}) \mid a_{i,k-1}=\mathfrak{W}, b_{i,k-1}] -\\
  &\quad\quad\quad\quad\quad\quad\quad\text{ }\text{ }\mathbb{E}_{B_{i,k}}[V_{i,k}^{\star}(B_{i,k}) \mid a_{i,k-1}=\mathfrak{R}, b_{i,k-1}]\\
  &\numeq{a} \mathbb{E}_{B_{i,k+1}}[1+ V_{i,k+1}^{\star}(B_{i,k+1}) \mid a_{i,k}=\mathfrak{R}] -\\
  &\quad\text{ }\mathbb{E}_{B_{i,k+1}}[1+ V_{i,k+1}^{\star}(B_{i,k+1}) \mid a_{i,k}=\mathfrak{R}]\\
& = \mathbb{E}_{B_{i,k+1}}[V_{i,k+1}^{\star}(B_{i,k+1}) \mid a_{i,k}=\mathfrak{R}] -\\
  &\quad\text{ }\mathbb{E}_{B_{i,k+1}}[V_{i,k+1}^{\star}(B_{i,k+1}) \mid a_{i,k}=\mathfrak{R}]= W_{i,k}(b_{i,k-1}),
\end{align*}
where (a) follows from (\ref{eq:belief_bellman}).
\item If $B_{i,k} \in \mathcal{W}$ both when $a_{i,k-1} = \mathfrak{W}$ and when $a_{i,k-1} = \mathfrak{R}$, then
\begin{align*}
&W_{i,k-1}(b_{i,k-1}) = \mathbb{E}_{B_{i,k}}[V_{i,k}^{\star}(B_{i,k}) \mid a_{i,k-1}=\mathfrak{W}, b_{i,k-1}] -\\
  &\quad\quad\quad\quad\quad\quad\quad\text{ }\text{ }\mathbb{E}_{B_{i,k}}[V_{i,k}^{\star}(B_{i,k}) \mid a_{i,k-1}=\mathfrak{R}, b_{i,k-1}]\\
  &= \mathbb{E}_{B_{i,k},B_{i,k+1}}[\eta B_{i,k} + V_{i,k+1}^{\star}(B_{i,k+1}) \mid \\
  &\quad\quad\quad\quad\quad\quad a_{i,k}=a_{i,k-1}=\mathfrak{W}, b_{i,k-1}] - \\
  &\quad\text{ }\mathbb{E}_{B_{i,k},B_{i,k+1}}[\eta B_{i,k} + V_{i,k+1}^{\star}(B_{i,k+1})\mid\\
  & \quad\quad\quad\quad\quad\quad a_{i,k}=\mathfrak{W},a_{i,k-1}=\mathfrak{R}, b_{i,k-1}]\\
& \numgeq{a} \mathbb{E}_{B_{i,k+1}}[V_{i,k+1}^{\star}(B_{i,k+1}) \mid a_{i,k}=a_{i,k-1}=\mathfrak{W},b_{i,k-1}] -\\
  &\quad\text{ }\mathbb{E}_{B_{i,k+1}}[V_{i,k+1}^{\star}(B_{i,k+1}) \mid a_{i,k}=\mathfrak{W},a_{i,k-1}=\mathfrak{R},b_{i,k-1}]\\
  &= W_{i,k}(b_{i,k-1}),
\end{align*}
where (a) follows because $\mathbb{E}[B_i^{\prime} \mid a_i=\mathfrak{W}, b_i] \geq \mathbb{E}[B_i^{\prime} \mid a_i=\mathfrak{R}, b_i]$ by definition (\ref{eq:recovery_dynamics}).
\item If $B_{i,k} \in \mathcal{R}$ when $a_{i,k-1} = \mathfrak{W}$, and $B_{i,k} \in \mathcal{W}$ when $a_{i,k-1} = \mathfrak{R}$, then
\begin{align*}
&W_{i,k-1}(b_{i,k-1}) = \mathbb{E}_{B_{i,k}}[V_{i,k}^{\star}(B_{i,k}) \mid a_{i,k-1}=\mathfrak{W}, b_{i,k-1}] -\\
  &\quad\quad\quad\quad\quad\quad\quad\text{ }\text{ }\mathbb{E}_{B_{i,k}}[V_{i,k}^{\star}(B_{i,k}) \mid a_{i,k-1}=\mathfrak{R}, b_{i,k-1}]\\
  &= \mathbb{E}_{B_{i,k+1}}[1+ V_{i,k+1}^{\star}(B_{i,k+1}) \mid a_{i,k}=\mathfrak{R}] -\\
  &\mathbb{E}_{B_{i,k}, B_{i,k+1}}[\eta B_{i,k} + V_{i,k+1}^{\star}(B_{i,k+1})\mid a_{i,k}=\mathfrak{W},a_{i,k-1}=\mathfrak{R}]\\
& \numgeq{a} \mathbb{E}_{B_{i,k+1}}[1+ V_{i,k+1}^{\star}(B_{i,k+1}) \mid a_{i,k}=\mathfrak{R}] -\\
  &\quad\text{ }\mathbb{E}_{B_{i,k+1}}[1+ V_{i,k+1}^{\star}(B_{i,k+1}) \mid a_{i,k}=\mathfrak{R}]\\
& \geq \mathbb{E}_{B_{i,k+1}}[V_{i,k+1}^{\star}(B_{i,k+1}) \mid a_{i,k}=\mathfrak{R}] -\\
  &\quad\text{ }\mathbb{E}_{B_{i,k+1}}[V_{i,k+1}^{\star}(B_{i,k+1}) \mid a_{i,k}=\mathfrak{R}]\numeq{d} W_{i,k}(b_{i,k-1}),
\end{align*}
where (a) follows from (\ref{eq:belief_bellman}).
\end{enumerate}
The case where $a_{i,k-1} = \mathfrak{R} \implies B_{i,k} \in \mathcal{R}$ and $a_{i,k-1} = \mathfrak{W} \implies B_{i,k} \in \mathcal{W}$ can be discarded due to Thm. \ref{thm:stopping_policy} since $\mathbb{E}[B_i^{\prime} \mid a_i=\mathfrak{W}, b_i] \geq \mathbb{E}[B_i^{\prime} \mid a_i=\mathfrak{R}, b_i]$, which means that if $B_{i,k} \in \mathcal{R}$ when $a_{i,k-1} = \mathfrak{R}$, then $B_{i,k} \in \mathcal{R}$ also when $a_{i,k-1} = \mathfrak{W}$. It follows by induction that $W_{i,t}(b) \geq W_{i,t+1}(b)$ for all $t$, $b$, and $i$. \qed
\section{Proof of Theorem \ref{thm:structure_respone}}\label{appendix:theorem_2}
Solving (\ref{eq:response_problem}) corresponds to solving a finite and stationary Constrained Markov Decision Process (\cmdp) with bounded costs and the average cost optimality criterion. Assumption \ref{eq:feasibility_assumption} implies that the \cmdp is feasible and assumption \ref{eq:unichain_assumption} implies that the \cmdp is \textit{unichain} \cite[Def. 6.5.1]{krishnamurthy_2016}, which means that there exists an optimal stationary strategy $\pi^{\star}$ for which the limit in (\ref{eq:objective_response}) exists \cite[Thm. 8.4.5]{puterman}\cite[Thms. 6.5.2--6.5.4]{krishnamurthy_2016}.

By introducing the Lagrange multiplier $\lambda\geq 0$ and defining the immediate cost to be $c_{\lambda}(s_t) = s_t + \lambda \llbracket s_t < f+1 \rrbracket$ we can reformulate the \cmdp as an unconstrained \mdp through Lagrangian relaxation \cite[Thm. 3.7]{altman-constrainedMDP}. The optimal strategy in the unconstrained \mdp satisfies
\begin{align}
&\pi_{\lambda}^{\star}(s_{t}) \in \argmin_{a \in \{0,1\}}\bigg[c_{\lambda}(s_t) + \mathbb{E}_{S_{t+1}}\left[V_{\lambda}^{\star}(S_{t+1})\right]\bigg],\label{eq:bellman_eq_state}
\end{align}
where $V_{\lambda}^{\star}$ is the value function \cite[Thm. 3.6]{altman-constrainedMDP}.

Since $c_{\lambda}(s_t)$ is non-decreasing in $s$ it follows from assumptions \ref{eq:dominance_transitions} and \ref{eq:tail_sum} that the \mdp has an optimal threshold strategy for any $\lambda$ \cite[Thm. 9.3.1]{krishnamurthy_2016}\cite[Prop. 4.7.3]{puterman}. Further, we know from Lagrangian dynamic programming theory that there exists an optimal strategy in the \cmdp which is a randomized mixture of two optimal deterministic strategies of the \mdp with different Lagrange multipliers $\lambda_1$ and $\lambda_2$ \cite[Thm. 6.6.2]{krishnamurthy_2016}, \cite[Thm. 12.7]{altman-constrainedMDP}. When combined, these two properties imply Thm. \ref{thm:structure_respone}. \qed

\section{Hyperparameters}\label{appendix:hyperparameters}
Hyperparameters for the experimental results and figures reported in the paper are listed in Table \ref{tab:hyperparams}.

\begin{table}
\centering
\resizebox{1\columnwidth}{!}{%
  \begin{tabular}{ll} \toprule
  \textbf{Intrusion recovery parameters} & {\textbf{Values}} \\
    \hline
    Confidence levels & Confidence levels for all figures were \\
                                         &computed based on the Student-t distribution\\
    Fig. \ref{fig:reliability_curve_1}--\ref{fig:reliability_curves_3}, Fig. \ref{fig:strategy_structure}   & $p_{\mathrm{C}_1,i}=10^{-5}$,$p_{\mathrm{C}_2,i}=10^{-3}$, $k=1$\\
                                         & $\eta=2$, $\mathcal{O}=\{0,\hdots,9\}$,\\
                                         & $Z_i(\cdot \mid 0) = \mathrm{BetaBin}(n=10,\alpha=0.7,\beta=3)$,\\
                                         & $Z_i(\cdot \mid 1) = \mathrm{BetaBin}(n=10,\alpha=1,\beta=0.7$\\\\
    Figs. \ref{fig:reliability_curve_1}--\ref{fig:reliability_curves_3} & no recoveries, $p_{\mathrm{U},i}=0$, $\Delta_{\mathrm{R}}=100$, $p_{\mathrm{A},i}=0.1$, $k=1$\\\\
    Fig. \ref{fig:value_funs} & $p_{\mathrm{A},i}=0.01$ \\\\
    Fig. \ref{fig:value_funs}, Fig. \ref{fig:strategy_structure} & $p_{\mathrm{U},i}=2\times 10^{-2}$, $\Delta_{\mathrm{R}}=100$, $k=1$\\\\
    Figs. \ref{fig:recovery_curves}-\ref{fig:pomdp_times} & $\eta=2$, $p_{\mathrm{A},i}=0.1$, $p_{\mathrm{C}_1,i}=10^{-5}$, $p_{\mathrm{C}_2,i}=10^{-3}$,\\
     & $p_{\mathrm{U},i}=2\times 20^{-2}$, $k=1$,\\
     & $Z_i(\cdot \mid 0) = \mathrm{BetaBin}(n=10,\alpha=0.7,\beta=3)$,\\
     & $Z_i(\cdot \mid 1) = \mathrm{BetaBin}(n=10,\alpha=1,\beta=0.7$,\\\\
    Fig. \ref{fig:cmdp_times} & $\epsilon_{\mathrm{A}}=0.9$, $N=10$, $f=3$,\\
                                         & see Fig. \ref{fig:pmf_1} for $f_{\mathrm{S}}$\\
    \\
    Evaluation in \S \ref{sec:solution_approach} & $p_{\mathrm{U},i}=2\times 10^{-2}$, $p_{\mathrm{A},i}=10^{-1}$, $p_{\mathrm{C}_1,i}=10^{-5}$,\\
                                         &  $p_{\mathrm{C}_2,i}=10^{-3}$, $\Delta_{\mathrm{R}}=\infty$, $\epsilon_{\mathrm{A}}=0.9$,\\
                                         & $s_{\mathrm{max}}=13$, $\eta = 2$, $N_1=3$, $f=\min[\frac{N_1-1}{2}, 2]$\\
                                         & $f_{\mathrm{S}}$ estimated from simulations of Prob \ref{prob_1},\\                                                                                                                                                                      & $\mathrm{PO}=\text{\cem}$ in Alg. \ref{alg:solver_prob1}\\
    Fig. \ref{fig:kls} & $\eta=2$, $p_{\mathrm{A},i}=0.1$, $p_{\mathrm{C}_1,i}=10^{-5}$, $p_{\mathrm{C}_2,i}=10^{-3}$,\\
                                         & $p_{\mathrm{U},i}=2\times 20^{-2}$, $k=1$,$\mathrm{PO}=\text{\cem}$ in Alg. \ref{alg:solver_prob1}\\
    {\textbf{\minbft \cite[\S 4.2]{giuliana_thesis} parameters}} &   \\
  \hline
    \usig implementation & \rsa with key lengths $1024$ bits \cite{rsa_citation}\\
    $T_{\mathrm{exec}}$, $T_{\mathrm{vc}}$, $\mathrm{cp}$, $L$ & $30$ seconds, $280$ seconds, $10^2$, $10^3$\\
  {\textbf{\ppo \cite[Alg. 1]{ppo} parameters}} &   \\
  \hline
  lr $\alpha$, batch, \# layers, \# neurons, clip $\epsilon$ & $10^{-5}$, $4\cdot 10^{3}t$, $4$, $64$, $0.2$,\\
  GAE $\lambda$, ent-coef, activation & $0.95$, $10^{-4}$, ReLU \\
  {\textbf{\spsa parameters \cite[Fig. 1]{spsa_impl}}} & \\
  \hline
    $c, \epsilon, \lambda, A, a, N, \delta$ & $10$, $0.101$, $0.602$, $100$, $1$, $50$, $0.2$\\
    $M$ number of samples for each evaluation & $50$ \\
  {\textbf{Incremental pruning parameters \cite[Fig. 4]{incremental_pruning_pomdp}}} &  \\
  \hline
    Variation, $\epsilon$ & normal, $0$\\
  {\textbf{Cross-entropy method \cite{cem_rubinstein}\cite[Alg. 1]{moss2020crossentropy}}} &   \\
  \hline
    $\lambda$ (fraction of samples to keep) & $0.15,100$\\
    $K$ population size & $100$\\
    $M$ number of samples for each evaluation & $50$ \\
  {\textbf{Differential evolution \cite[Fig. 3]{differential_evolution}}} &  \\
  \hline
    Population size $K$, mutate step  & $10, 0.2$\\
    Recombination rate  & $0.7$\\
    $M$ number of samples for each evaluation & $50$ \\
  {\textbf{Bayesian optimization \cite[Alg. 1]{bayesian_opt}}} &  \\
  \hline
    Acquisition function  & lower confidence bound \cite[Alg. 1]{gp_ucb}\\
    $\beta$, Kernel  & $2.5$, $\operatorname{Matern}(2.5)$\\
    $M$ number of samples for each evaluation & $50$ \\
  {\textbf{Linear Programming}} &  \\
  \hline
    Solver  & \cbc \cite{forrest2005cbc}\\
  \bottomrule\\
\end{tabular}
}
\caption{Hyperparameters.}\label{tab:hyperparams}
\end{table}
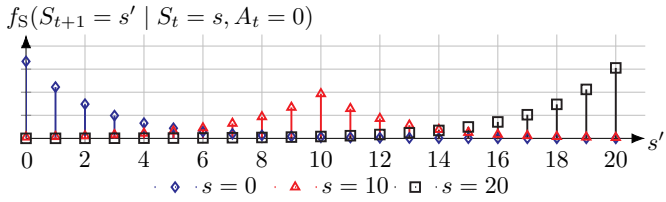
\begin{figure}
  \centering
    \scalebox{0.91}{
      \input{tikz/pmf_1.tex}
    }
    \caption{Example transition function for Prob. \ref{prob_2}.}
    \label{fig:pmf_1}
\end{figure}
\section{Computation of \mttf and Reliability Functions}\label{appendix_mttf}
The \mttf and the reliability function $R(t)$ can be calculated using numerical methods for Markov chains. Specifically, the number of healthy nodes in the system can be modeled as a Markov chain with state space $\mathcal{S} \triangleq \{0,1,\hdots,N\}$ and transition matrix $\mathbf{P} \in [0,1]^{|\mathcal{S}|^2}$. In this Markov chain, the subset of states $\mathcal{F} \triangleq \{0,1,\hdots,f\} \subset \mathcal{S}$ represents the states where service is unavailable. ($f$ is a fixed tolerance threshold and service is guaranteed if $S \geq f+1$ (Prop. \ref{prop:correctness}).) When calculating the \mttf, we assume that there are no recoveries, which means that $\mathcal{F}$ is absorbing. As a consequence, the mean time to failure ($\mttf$) can be defined as
\begin{align*}
\mathbb{E}[T^{(f)} \mid S_1=s_1] \triangleq \mathbb{E}_{(S_t)_{t\geq 1}}\left[\inf \left\{t \geq 1 \mid S_t \in \mathcal{F}\right\}\mid S_1=s_1\right],
\end{align*}
i.e., the \mttf is the mean hitting time of $\mathcal{F}$ in the Markov chain starting at $s_1 \in \mathcal{S}$.

By standard Markov chain calculations:
\begin{align*}
  &\mathbb{E}[T^{(f)} \mid S_1=s_1]
  \\&=
                                        \begin{dcases}
                                          0 & \text{if } s_1 \in \mathcal{F}\\
                                          1 + \sum_{s^{\prime} \in \mathcal{S} \setminus \mathcal{F}}\mathbf{P}_{s_1,s^{\prime}}\mathbb{E}[T^{(f)} \mid S_1=s^{\prime}] & \text{if }s_1 \not\in \mathcal{F},
                                          \end{dcases}
\end{align*}
which defines a set of $|\mathcal{S}|$ linear equations that can be solved using Gaussian elimination.

Similarly, since the reliability function is defined as $R(t) \triangleq \mathbb{P}[T^{(f)}> t]=\mathbb{P}[S_t > f]$, we have from the Chapman-Kolmogorov equation that
\begin{align}
R(t) &= \sum_{s \in \mathcal{S} \setminus \mathcal{F}}\left(\mathbf{e}^T_{s_1}\mathbf{P}^{t}\right)_{s},
\end{align}
where $\mathbf{e}_{s_1}$ is the $s_1$-basis vector.

\section{The \minbft Consensus Protocol \cite[\S 4.2]{giuliana_thesis}}\label{appendix_minbft}
\tolerancee is based on a reconfigurable consensus protocol for the partially synchronous system model with hybrid failures, a reliable network, and authenticated communication links (see \S \ref{sec:implementation} and Prop. \ref{prop:correctness}). Examples of such protocols include \minbft \cite[\S 4.2]{giuliana_thesis}, \minzyzzyva \cite[\S 4.3]{giuliana_thesis}, \reminbft \cite[\S 5]{7307998}, and \cheapbft \cite[\S 3]{10.1145/2168836.2168866}. Our implementation uses \minbft. Correctness of \minbft is proven in \cite[Appendix A]{giuliana_thesis}.

\minbft is based on \pbft \cite{pbft} with one crucial difference. While \pbft assumes Byzantine failures and tolerates up to $f = \frac{N-1}{3}$ failures, \minbft assumes hybrid failures \cite{wormit} and tolerates up to $f=\frac{N-1}{2}$ failures. The improved resilience of \minbft is achieved by assuming access to a trusted component that provides certain functions for the protocol. In particular, \minbft relies on a tamperproof service at each node that can assert whether a given sequence number was assigned to a message. This service allows \minbft to prevent equivocation \cite{ampbea} and imposes a first-in-first-out (\fifo) order on requests issued by clients. In \tolerancee, the tamperproof service is provided by the virtualization layer (see Fig. \ref{fig:tolerance_6}).

We extend the \minbft protocol \cite[\S 4.2]{giuliana_thesis} to be \textit{reconfigurable} \cite{reconfigurable_consensus}, where the reconfiguration procedure is based on the method described in \cite[\S IV.B]{8433150}. The different stages of the protocol are illustrated in Fig. \ref{fig:tolerance_28} and the throughput of our implementation is shown in Fig. \ref{fig:tolerance_24}. The source code is available at \cite{csle_docs,supplementary} and the hyperparameters are listed in Table \ref{tab:hyperparams}.

\begin{figure}
  \centering
  \scalebox{1.35}{
    \input{tikz/tolerance_28.tex}
  }
  \caption{Time-space diagrams illustrating the message patterns of the \minbft consensus protocol \cite[\S 4.2]{giuliana_thesis}.}
  \label{fig:tolerance_28}
\end{figure}
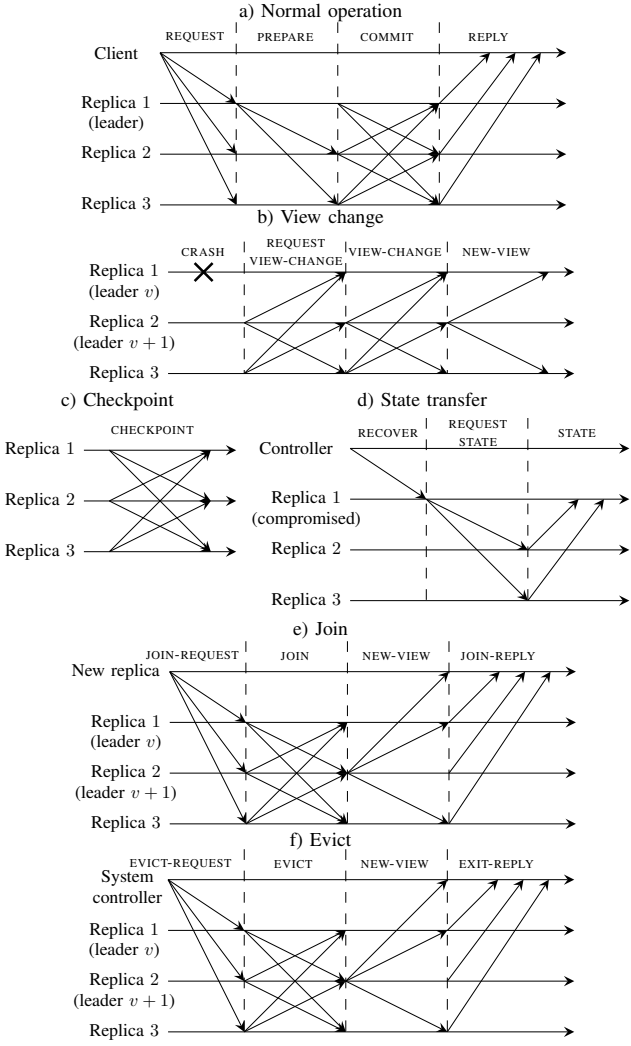

\section{Distributions of System Metrics}\label{appendix:infrastructure_metrics}
Our testbed implementation of \tolerancee collects hundreds of metrics every time step. To measure the information that a metric provides for detecting intrusions, we calculate the Kullback-Leibler (\kl) divergence $D_{\text{\textsc{kl}}}(\widehat{Z}_{O|0} \parallel \widehat{Z}_{O | s>0})$ between the distribution of the metric when no intrusion occurs $\widehat{Z}_{O|\mathbb{H}} \triangleq \widehat{Z}(\cdot \mid S_i=\mathbb{H})$ and during an intrusion $\widehat{Z}_{O|\mathbb{C}} \triangleq \widehat{Z}(\cdot \mid S_i=\mathbb{C})$:
\begin{align}
D_{\text{\textsc{kl}}}(\widehat{Z}_{O|\mathbb{H}} \parallel \widehat{Z}_{O \mid \mathbb{C}})= \sum_{o \in \mathcal{O}}\widehat{Z}_{O \mid \mathbb{H}}\log\left(\frac{\widehat{Z}_{O \mid \mathbb{H}}}{\widehat{Z}_{O \mid \mathbb{C}}}\right).\nonumber
\end{align}
Here $o\in \mathcal{O}$ realizes the random variable $O$ (\ref{eq:obs_function}), which represents the value of the metric. ($\mathcal{O}$ is the domain of $O$.)

Figure \ref{fig:observations} shows empirical distributions of the collected metrics with the largest \kl divergence. We see that the \textsc{ids} alerts have the largest \kl divergence and thus provide the most information for detecting the type of intrusions that we consider in this paper (see Table \ref{tab:attacker_actions}).
\begin{figure}
  \centering
    \scalebox{0.45}{
      \includegraphics{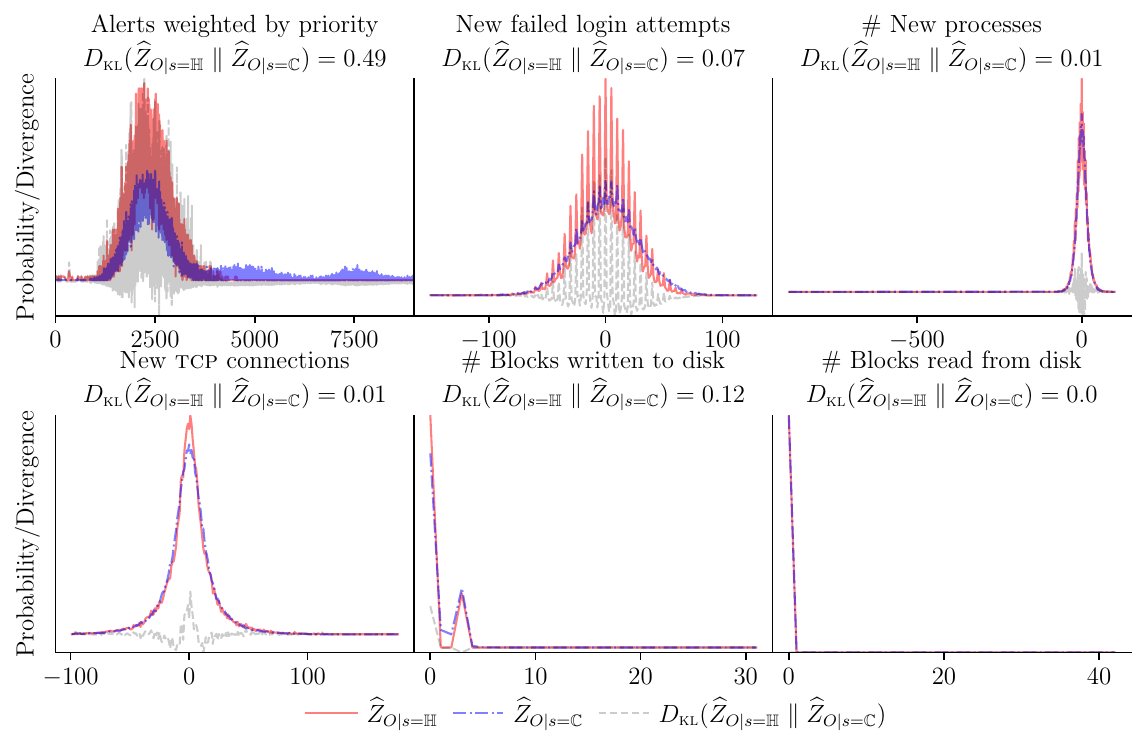}
    }
    \caption{Empirical distributions of selected infrastructure metrics; the red and blue lines show the distributions when no intrusion occurs and during an intrusion, respectively.}
    \label{fig:observations}
\end{figure}

\hypersetup{
  colorlinks,
  linkcolor={black},
  citecolor={black},
  urlcolor={black}
}
\bibliographystyle{IEEEtran}
\bibliography{references,url}

\end{document}

%% file: config.tex
\usepackage[noadjust]{cite}
\usepackage[bookmarks=false]{hyperref}
\hypersetup{
  colorlinks,
  linkcolor={blue!50!black},
  citecolor={blue!50!black},
  urlcolor={blue!80!black}
}
\usepackage{subcaption}
\usepackage{caption}
\usepackage[utf8]{inputenc} 
\usepackage{graphicx,url}
\usepackage{booktabs}       
\graphicspath{ {./images/} }
\usepackage{multirow}
\usepackage{hhline}
\usepackage{xspace}
\usepackage{graphicx}
\usepackage{amssymb,amsmath,mathtools} %
\usepackage{amsthm}
\newtheorem{theorem}{Theorem}   %
\newtheorem{proposition}{Proposition}
\newtheorem{problem}{Problem}

\newtheorem{corollary}{Corollary}
\usepackage{bm}                 %

\usepackage{color,soul}
\usepackage[dvipsnames]{xcolor}
\usepackage{bbm}
\usepackage{stmaryrd}
\usepackage{textcomp}
\usepackage{lipsum,lineno}
\usepackage{float}
\makeatletter
\newif\if@restonecol
\makeatother
\usepackage{amsmath}
\usepackage{kbordermatrix}
\usepackage{mathrsfs}
\usepackage[shortlabels]{enumitem}
\newcommand\numeq[1]%
{\stackrel{\scriptscriptstyle(\mkern-1.5mu#1\mkern-1.5mu)}{=}}
\newcommand\numeqq[1]%
{\stackrel{\scriptscriptstyle(\mkern-1.5mu#1\mkern-1.5mu)}{\triangleq}}
\newcommand\numleq[1]%
{\stackrel{\scriptscriptstyle(\mkern-1.5mu#1\mkern-1.5mu)}{\leq}}
\newcommand\numgeq[1]%
{\stackrel{\scriptscriptstyle(\mkern-1.5mu#1\mkern-1.5mu)}{\geq}}
\newcommand\numimp[1]%
{\stackrel{\scriptscriptstyle(\mkern-1.5mu#1\mkern-1.5mu)}{\implies}}
\usepackage[titlenumbered,ruled,linesnumbered]{algorithm2e}
\SetAlCapNameFnt{\footnotesize}
\SetAlCapFnt{\footnotesize}

\usepackage{fancyhdr} 
\usepackage{tabularx}
\usepackage{multirow}
\usepackage{dsfont}
\usepackage{listings}
\usepackage{tcolorbox}
\newcommand{\acro}[1]{\textsc{#1}\xspace}
\newcommand{\acros}[1]{\textsc{#1}s\xspace}
\newcommand{\acrop}[1]{\textsc{#1}\xspace}

\newcommand{\bft}{\acro{bft}}
\newcommand{\rampart}{\acro{rampart}}
\newcommand{\hotstuff}{\acro{hotstuff}}
\newcommand{\uprightt}{\acro{upright}}
\newcommand{\securering}{\acro{secure-ring}}
\newcommand{\scada}{\acro{scada}}
\newcommand{\vmfit}{\acro{vm-fit}}
\newcommand{\smart}{\acro{smart}}
\newcommand{\phoenix}{\acro{phoenix}}
\newcommand{\sbft}{\acro{sbft}}
\newcommand{\raft}{\acro{raft}}
\newcommand{\wormit}{\acro{worm-it}}
\newcommand{\mttf}{\acro{mttf}}

\newcommand{\crutial}{\acro{crutial}}
\newcommand{\itcisprr}{\acro{itcis-prr}}
\newcommand{\recover}{\acro{recover}}
\newcommand{\spire}{\acro{spire}}

\newcommand{\skynet}{\acro{skynet}}
\newcommand{\reminbft}{\acro{reminbft}}
\newcommand{\minbft}{\acro{minbft}}
\newcommand{\fifo}{\acro{fifo}}
\newcommand{\cheapbft}{\acro{cheapbft}}
\newcommand{\usig}{\acro{usig}}
\newcommand{\rsa}{\acro{rsa}}
\newcommand{\minzyzzyva}{\acro{minzyzzyva}}
\newcommand{\hq}{\acro{hq}}
\newcommand{\flp}{\acro{flp}}
\newcommand{\pbft}{\acro{pbft}}
\newcommand{\zyzzyva}{\acro{zyzzyva}}
\newcommand{\prrw}{\acro{prrw}}
\newcommand{\scit}{\acro{scit}}
\newcommand{\coca}{\acro{coca}}
\newcommand{\itua}{\acro{itua}}
\newcommand{\maftia}{\acro{maftia}}
\newcommand{\itsi}{\acro{itsi}}
\newcommand{\sitar}{\acro{sitar}}
\newcommand{\btr}{\acro{btr}}
\newcommand{\kl}{\acro{kl}}

\newcommand{\mdp}{\acro{mdp}}
\newcommand{\tpp}{\acrop{tp-2}}
\newcommand{\pomdp}{\acro{pomdp}}
\newcommand{\cmdp}{\acro{cmdp}}

\newcommand{\pomdps}{\acros{pomdp}}

\newcommand{\ids}{\acro{ids}}

\newcommand{\tolerancee}{\acro{tolerance}}

\newcommand{\ssh}{\acro{ssh}}
\newcommand{\irc}{\acro{irc}}

\newcommand{\debian}{\acro{debian}}

\newcommand{\capp}{\acro{cap}}
\newcommand{\dns}{\acro{dns}}
\newcommand{\snort}{\acro{snort}}
\newcommand{\http}{\acro{http}}

\newcommand{\netem}{\acro{netem}}

\newcommand{\tcpp}{\acro{tcp}}

\newcommand{\syn}{\acro{syn}}

\newcommand{\ubuntu}{\acro{ubuntu}}

\newcommand{\ftp}{\acro{ftp}}

\newcommand{\cve}{\acro{cve}}
\newcommand{\icmp}{\acro{icmp}}
\newcommand{\cwe}{\acro{cwe}}

\newcommand{\dvwa}{\acro{dvwa}}
\newcommand{\telnet}{\acro{telnet}}
\newcommand{\pspace}{\acro{pspace}}

\newcommand{\ppo}{\acro{ppo}}
\newcommand{\cem}{\acro{cem}}
\newcommand{\cbc}{\acro{cbc}}
\newcommand{\ipp}{\acro{ip}}
\newcommand{\de}{\acro{de}}
\newcommand{\spsa}{\acro{spsa}}
\newcommand{\bo}{\acro{bo}}

\lstdefinestyle{mystyle}{
  backgroundcolor=\color{backcolour},
  commentstyle=\color{codegreen},
  keywordstyle=\color{magenta},
  numberstyle=\tiny\color{codegray},
  stringstyle=\color{codepurple},
  basicstyle=\ttfamily\footnotesize,
  breakatwhitespace=false,
  breaklines=true,
  captionpos=b,
  keepspaces=true,
  showspaces=false,
  showstringspaces=false,
  showtabs=false,
  tabsize=2,
  xleftmargin=50pt,
  xrightmargin=50pt
}
\usepackage{tikz,pgfplots}
\usepackage{pgfplots}
\usepackage{pgfplotstable}
\definecolor{gray2}{HTML}{ededed}
\definecolor{gray3}{HTML}{F5F5F5}
\usetikzlibrary{shapes.geometric,backgrounds,patterns, trees}
\usetikzlibrary{3d,decorations.text,shapes.arrows,positioning,fit,backgrounds}
\usetikzlibrary{positioning, decorations.pathmorphing, shapes}
\usetikzlibrary{decorations.pathreplacing}
\usetikzlibrary{shapes.geometric,backgrounds,patterns, trees}
\usetikzlibrary{spy}
\usetikzlibrary{arrows.meta,
  bending,
  intersections,
  quotes,
  shapes.geometric}
\usetikzlibrary{automata, positioning}
\usepgfplotslibrary{fillbetween}
\pgfmathdeclarefunction{gauss}{2}{%
  \pgfmathparse{1/(#2*sqrt(2*pi))*exp(-((x-#1)^2)/(2*#2^2))}%
}
\usetikzlibrary{shapes,arrows}
\usetikzlibrary{arrows.meta}
\usetikzlibrary{positioning}
\tikzset{set/.style={draw,circle,inner sep=0pt,align=center}}
\usetikzlibrary{automata, positioning}
\usetikzlibrary{shapes,shadows}
\tikzstyle{abstractbox} = [draw=black, fill=white, rectangle,
inner sep=10pt, style=rounded corners, drop shadow={fill=black,
  opacity=1}]
\tikzstyle{abstracttitle} =[fill=white]
\usetikzlibrary{calc,positioning,shapes.geometric}
\usetikzlibrary{arrows.meta,arrows}

\DeclareMathOperator*{\argmin}{arg\,min}

\DeclareMathOperator*{\minimize}{\text{\normalfont minimize}}

\usetikzlibrary{matrix}
\tikzstyle{cblue}=[circle, draw, thin,fill=cyan!20, scale=0.8]
\tikzstyle{qgre}=[rectangle, draw, thin,fill=green!20, scale=0.8]
\tikzstyle{rpath}=[ultra thick, red, opacity=0.4]
\tikzstyle{legend_isps}=[rectangle, rounded corners, thin,
fill=gray!20, text=blue, draw]

\tikzstyle{legend_overlay}=[rectangle, rounded corners, thin,
top color= white,bottom color=green!25,
minimum width=2.5cm, minimum height=0.8cm,
pinegreen]
\tikzstyle{legend_phytop}=[rectangle, rounded corners, thin,
top color= white,bottom color=cyan!25,
minimum width=2.5cm, minimum height=0.8cm,
royalblue]
\tikzstyle{legend_general}=[rectangle, rounded corners, thin,
top color= white,bottom color=lavander!25,
minimum width=2.5cm, minimum height=0.8cm,
violet]
\usetikzlibrary{matrix}

\colorlet{myRed}{red!20}
\tikzset{
  rows/.style 2 args={/utils/temp/.style={row ##1/.append style={nodes={#2}}},
    /utils/temp/.list={#1}},
  columns/.style 2 args={/utils/temp/.style={column ##1/.append style={nodes={#2}}},
    /utils/temp/.list={#1}}}
\usetikzlibrary{backgrounds,calc,shadings,shapes.arrows,shapes.symbols,shadows}
\definecolor{switch}{HTML}{006996}

\makeatletter
\pgfkeys{/pgf/.cd,
  parallelepiped offset x/.initial=2mm,
  parallelepiped offset y/.initial=2mm
}
\pgfdeclareshape{parallelepiped}
{
  \inheritsavedanchors[from=rectangle]
  \inheritanchorborder[from=rectangle]
  \inheritanchor[from=rectangle]{north}
  \inheritanchor[from=rectangle]{north west}
  \inheritanchor[from=rectangle]{north east}
  \inheritanchor[from=rectangle]{center}
  \inheritanchor[from=rectangle]{west}
  \inheritanchor[from=rectangle]{east}
  \inheritanchor[from=rectangle]{mid}
  \inheritanchor[from=rectangle]{mid west}
  \inheritanchor[from=rectangle]{mid east}
  \inheritanchor[from=rectangle]{base}
  \inheritanchor[from=rectangle]{base west}
  \inheritanchor[from=rectangle]{base east}
  \inheritanchor[from=rectangle]{south}
  \inheritanchor[from=rectangle]{south west}
  \inheritanchor[from=rectangle]{south east}
  \backgroundpath{
    \southwest \pgf@xa=\pgf@x \pgf@ya=\pgf@y
    \northeast \pgf@xb=\pgf@x \pgf@yb=\pgf@y
    \pgfmathsetlength\pgfutil@tempdima{\pgfkeysvalueof{/pgf/parallelepiped
        offset x}}
    \pgfmathsetlength\pgfutil@tempdimb{\pgfkeysvalueof{/pgf/parallelepiped
        offset y}}
    \def\ppd@offset{\pgfpoint{\pgfutil@tempdima}{\pgfutil@tempdimb}}
    \pgfpathmoveto{\pgfqpoint{\pgf@xa}{\pgf@ya}}
    \pgfpathlineto{\pgfqpoint{\pgf@xb}{\pgf@ya}}
    \pgfpathlineto{\pgfqpoint{\pgf@xb}{\pgf@yb}}
    \pgfpathlineto{\pgfqpoint{\pgf@xa}{\pgf@yb}}
    \pgfpathclose
    \pgfpathmoveto{\pgfqpoint{\pgf@xb}{\pgf@ya}}
    \pgfpathlineto{\pgfpointadd{\pgfpoint{\pgf@xb}{\pgf@ya}}{\ppd@offset}}
    \pgfpathlineto{\pgfpointadd{\pgfpoint{\pgf@xb}{\pgf@yb}}{\ppd@offset}}
    \pgfpathlineto{\pgfpointadd{\pgfpoint{\pgf@xa}{\pgf@yb}}{\ppd@offset}}
    \pgfpathlineto{\pgfqpoint{\pgf@xa}{\pgf@yb}}
    \pgfpathmoveto{\pgfqpoint{\pgf@xb}{\pgf@yb}}
    \pgfpathlineto{\pgfpointadd{\pgfpoint{\pgf@xb}{\pgf@yb}}{\ppd@offset}}
  }
}

\makeatletter
\tikzset{anchor/.append code=\let\tikz@auto@anchor\relax,
  add font/.code=%
  \expandafter\def\expandafter\tikz@textfont\expandafter{\tikz@textfont#1},
  left delimiter/.style 2 args={append after command={\tikz@delimiter{south east}
      {south west}{every delimiter,every left delimiter,#2}{south}{north}{#1}{.}{\pgf@y}}}}
\tikzstyle{sms} = [rectangle callout, draw,very thick, rounded corners, minimum height=20pt]
\makeatletter
\tikzset{anchor/.append code=\let\tikz@auto@anchor\relax,
  add font/.code=%
  \expandafter\def\expandafter\tikz@textfont\expandafter{\tikz@textfont#1},
  left delimiter/.style 2 args={append after command={\tikz@delimiter{south east}
      {south west}{every delimiter,every left delimiter,#2}{south}{north}{#1}{.}{\pgf@y}}}}
\tikzstyle{sms} = [rectangle callout, draw,very thick, rounded corners, minimum height=20pt]
\usetikzlibrary{positioning,calc}
\tikzstyle{block} = [rectangle, draw,
text width=10.5em, text centered, rounded corners, minimum height=4em]
\tikzstyle{line} = [draw, -latex]
\tikzset{
  mybackground9/.style={execute at end picture={
      \begin{scope}[on background layer]
        \draw[black,fill=black!5,rounded corners=6ex] (current bounding box.south west)
        rectangle (current bounding box.north east);
        \node[draw,fill=white,ellipse,anchor=west,inner sep=1pt,minimum width=4ex] at (current bounding box.north
        west){#1};
      \end{scope}
    }},
}
\newcommand{\Crosss}{$\mathbin{\tikz [x=1.4ex,y=1.4ex,line width=.2ex] \draw (0,0) -- (1,1) (0,1) -- (1,0);}$}%
\tikzset{
  mybackground13/.style={execute at end picture={
      \begin{scope}[on background layer]
        \draw[black, fill=gray2, rounded corners=4ex] (current bounding box.south west)
        rectangle (current bounding box.north east);
        \node[draw,fill=white,ellipse,anchor=west,inner sep=1pt,minimum width=4ex] at (current bounding box.north
        west){#1};
      \end{scope}
    }},
}
\tikzset{
  mybackground14/.style={execute at end picture={
      \begin{scope}[on background layer]
        \draw[black, rounded corners=2ex] (current bounding box.south west)
        rectangle (current bounding box.north east);
        \node[draw,fill=white,ellipse,anchor=west,inner sep=1pt,minimum width=4ex] at (current bounding box.north
        west){#1};
      \end{scope}
    }},
}

\tikzset{
  mybackground6/.style={execute at end picture={
      \begin{scope}[on background layer]
        \draw[black,rounded corners=1ex, line width=0.15mm] (current bounding box.south west)
        rectangle (current bounding box.north east);
        \node[draw,fill=white,ellipse,anchor=west,inner sep=1pt,minimum width=4ex] at (current bounding box.north
        west){#1};
      \end{scope}
    }},
}

\tikzset{
  mybackground11/.style={execute at end picture={
      \begin{scope}[on background layer]
        \draw[black, fill=Black!80!Sepia!9, rounded corners=6ex] (current bounding box.south west)
        rectangle (current bounding box.north east);
        \node[draw,fill=white,ellipse,anchor=west,inner sep=1pt,minimum width=4ex] at (current bounding box.north
        west){#1};
      \end{scope}
    }},
}

\tikzset{
  mybackground15/.style={execute at end picture={
      \begin{scope}[on background layer]
        \draw[black, fill=Black!80!Sepia!9, rounded corners=3ex] (current bounding box.south west)
        rectangle (current bounding box.north east);
        \node[draw,fill=white,ellipse,anchor=west,inner sep=1pt,minimum width=4ex] at (current bounding box.north
        west){#1};
      \end{scope}
    }},
}

\tikzset{
  mybackground12/.style={execute at end picture={
      \begin{scope}[on background layer]
        \draw[black, fill=Black!40!Emerald!30, rounded corners=3ex, line width=0.3mm] (current bounding box.south west)
        rectangle (current bounding box.north east);
      \end{scope}
    }},
}
\tikzset{
  mybackground18/.style={execute at end picture={
      \begin{scope}[on background layer]
        \draw[black, draw, fill=gray3, rounded corners=0ex] (current bounding box.south west)
        rectangle (current bounding box.north east);
        \node[draw,fill=white,rectangle,anchor=west,inner sep=2pt,minimum width=4ex, rounded corners=0.7ex] at (current bounding box.north
        west){#1};
      \end{scope}
    }},
}
\tikzset{
  mybackground58/.style={execute at end picture={
      \begin{scope}[on background layer]
        \draw[black, fill=blue!40!black!5, rounded corners=1ex] (current bounding box.south west)
        rectangle (current bounding box.north east);
        \node[draw,fill=white,ellipse,anchor=west,inner sep=1pt,minimum width=4ex, rounded corners=1ex] at (current bounding box.north
        west){#1};
      \end{scope}
    }},
}
\tikzset{l3 switch/.style={
    parallelepiped,fill=switch, draw=white,
    minimum width=0.75cm,
    minimum height=0.75cm,
    parallelepiped offset x=1.75mm,
    parallelepiped offset y=1.25mm,
    path picture={
      \node[fill=white,
      circle,
      minimum size=6pt,
      inner sep=0pt,
      append after command={
        \pgfextra{
          \foreach \angle in {0,45,...,360}
          \draw[-latex,fill=white] (\tikzlastnode.\angle)--++(\angle:2.25mm);
        }
      }
      ]
      at ([xshift=-0.75mm,yshift=-0.5mm]path picture bounding box.center){};
    }
  },
  ports/.style={
    line width=0.3pt,
    top color=gray!20,
    bottom color=gray!80
  },
  rack switch/.style={
    parallelepiped,fill=white, draw,
    minimum width=1.25cm,
    minimum height=0.25cm,
    parallelepiped offset x=2mm,
    parallelepiped offset y=1.25mm,
    xscale=-1,
    path picture={
      \draw[top color=gray!5,bottom color=gray!40]
      (path picture bounding box.south west) rectangle
      (path picture bounding box.north east);
      \coordinate (A-west) at ([xshift=-0.2cm]path picture bounding box.west);
      \coordinate (A-center) at ($(path picture bounding box.center)!0!(path
      picture bounding box.south)$);
      \foreach \x in {0.275,0.525,0.775}{
        \draw[ports]([yshift=-0.05cm]$(A-west)!\x!(A-center)$)
        rectangle +(0.1,0.05);
        \draw[ports]([yshift=-0.125cm]$(A-west)!\x!(A-center)$)
        rectangle +(0.1,0.05);
      }
      \coordinate (A-east) at (path picture bounding box.east);
      \foreach \x in {0.085,0.21,0.335,0.455,0.635,0.755,0.875,1}{
        \draw[ports]([yshift=-0.1125cm]$(A-east)!\x!(A-center)$)
        rectangle +(0.05,0.1);
      }
    }
  },
  server/.style={
    parallelepiped,
    fill=white, draw,
    minimum width=0.35cm,
    minimum height=0.75cm,
    parallelepiped offset x=3mm,
    parallelepiped offset y=2mm,
    xscale=-1,
    path picture={
      \draw[top color=gray!5,bottom color=gray!40]
      (path picture bounding box.south west) rectangle
      (path picture bounding box.north east);
      \coordinate (A-center) at ($(path picture bounding box.center)!0!(path
      picture bounding box.south)$);
      \coordinate (A-west) at ([xshift=-0.575cm]path picture bounding box.west);
      \draw[ports]([yshift=0.1cm]$(A-west)!0!(A-center)$)
      rectangle +(0.2,0.065);
      \draw[ports]([yshift=0.01cm]$(A-west)!0.085!(A-center)$)
      rectangle +(0.15,0.05);
      \fill[black]([yshift=-0.35cm]$(A-west)!-0.1!(A-center)$)
      rectangle +(0.235,0.0175);
      \fill[black]([yshift=-0.385cm]$(A-west)!-0.1!(A-center)$)
      rectangle +(0.235,0.0175);
      \fill[black]([yshift=-0.42cm]$(A-west)!-0.1!(A-center)$)
      rectangle +(0.235,0.0175);
    }
  },
}

\usetikzlibrary{calc, shadings, shadows, shapes.arrows}
\tikzset{cross/.style={cross out, draw=black, minimum size=2*(#1-\pgflinewidth), inner sep=0pt, outer sep=0pt},
  cross/.default={1pt}}
\tikzset{%
  interface/.style={draw, rectangle, rounded corners, font=\LARGE\sffamily},
  ethernet/.style={interface, fill=yellow!50},
  serial/.style={interface, fill=green!70},
  speed/.style={sloped, anchor=south, font=\large\sffamily},
  route/.style={draw, shape=single arrow, single arrow head extend=4mm,
    minimum height=1.7cm, minimum width=3mm, white, fill=switch!20,
    drop shadow={opacity=.8, fill=switch}, font=\tiny}
}
\makeatletter
\pgfdeclareradialshading[tikz@ball]{cloud}{\pgfpoint{-0.275cm}{0.4cm}}{%
  color(0cm)=(tikz@ball!75!white);
  color(0.1cm)=(tikz@ball!85!white);
  color(0.2cm)=(tikz@ball!95!white);
  color(0.7cm)=(tikz@ball!89!black);
  color(1cm)=(tikz@ball!75!black)
}
\tikzoption{cloud color}{\pgfutil@colorlet{tikz@ball}{#1}%
  \def\tikz@shading{cloud}\tikz@addmode{\tikz@mode@shadetrue}}
\makeatother

\tikzset{my cloud/.style={
    cloud, draw, aspect=2,
    cloud color={gray!5!white}
  }
}
\renewcommand\thetable{\arabic{table}}
\usepackage{etoolbox}
\makeatletter
\patchcmd{\@makecaption}
{\scshape}
{}
{}
{}
\makeatother

\allowdisplaybreaks

%% file: tikz/tolerance_13.tex
      \begin{tikzpicture}[fill=white, >=stealth,
    node distance=3cm,
    database/.style={
      cylinder,
      cylinder uses custom fill,
      shape border rotate=90,
      aspect=0.25,
      draw}]

    \tikzset{
node distance = 9em and 4em,
sloped,
   box/.style = {%
    shape=rectangle,
    rounded corners,
    draw=blue!40,
    fill=blue!15,
    align=center,
    font=\fontsize{12}{12}\selectfont},
 arrow/.style = {%
    line width=0.1mm,
    -{Triangle[length=5mm,width=2mm]},
    shorten >=1mm, shorten <=1mm,
    font=\fontsize{8}{8}\selectfont},
}

\draw[fill=gray2, rounded corners] (-1.27,-0.35) rectangle node (m1){} (5.55,1.2);

\node [scale=0.75] (node1) at (2,0.75) {
\begin{tikzpicture}[fill=white, >=stealth,
    node distance=3cm,
    database/.style={
      cylinder,
      cylinder uses custom fill,
      shape border rotate=90,
      aspect=0.25,
      draw}]
    \tikzset{
node distance = 9em and 4em,
sloped,
   box/.style = {%
    shape=rectangle,
    rounded corners,
    draw=blue!40,
    fill=blue!15,
    align=center,
    font=\fontsize{12}{12}\selectfont},
 arrow/.style = {%
    line width=0.1mm,
    shorten >=1mm, shorten <=1mm,
    font=\fontsize{8}{8}\selectfont},
}

\node[scale=1] (emulation_system) at (-0.07,-2.7)
{
\begin{tikzpicture}

\node[scale=0.9] (cloud_1) at (3,0.5)
{
  \begin{tikzpicture}

\node[server, scale=0.8](n1) at (-4,0) {};
\node[server, scale=0.8](n2) at (-2,0) {};
\node[server, scale=0.8](n3) at (0,0) {};
\node[server, scale=0.8](n4) at (2,0) {};
\node[server, scale=0.8](n5) at (4,0) {};
\node[inner sep=0pt,align=center, scale=1.8, color=black] (hacker) at (3,0) {
$\hdots$
};

\draw[-, color=black] (-3.86,0) to (-2.4,0);
\draw[-, color=black] (-1.86,0) to (-0.4,0);
\draw[-, color=black] (0.14,0) to (1.6,0);
\draw[-, color=black] (2.14,0) to (2.4,0);
\draw[-, color=black] (3.3,0) to (3.6,0);

    \end{tikzpicture}
  };
    \end{tikzpicture}
  };
\end{tikzpicture}
};
\node[scale=1] (emulation_system) at (-0.57,0.35)
{
  \begin{tikzpicture}

\node[inner sep=0pt,align=center, scale=0.75, color=black] (hacker) at (-0.52,0.22) {
  $\pi_1(b_1)$
};
\draw[->, color=black, rounded corners] (-0.94,0.15) to (-1.15,0.15) to (-1.15, 0.6) to (-0.9, 0.6);
\draw[->, color=black, rounded corners] (-0.51,0.6) to (0,0.6) to (0, 0.15) to (-0.25, 0.15);
    \end{tikzpicture}
};

\node[scale=1] (emulation_system) at (0.75,0.35)
{
  \begin{tikzpicture}

\node[inner sep=0pt,align=center, scale=0.75, color=black] (hacker) at (-0.52,0.22) {
  $\pi_2(b_2)$
};
\draw[->, color=black, rounded corners] (-0.94,0.15) to (-1.15,0.15) to (-1.15, 0.6) to (-0.9, 0.6);
\draw[->, color=black, rounded corners] (-0.51,0.6) to (0,0.6) to (0, 0.15) to (-0.25, 0.15);
    \end{tikzpicture}
};

\node[scale=1] (emulation_system) at (2.1,0.35)
{
  \begin{tikzpicture}

\node[inner sep=0pt,align=center, scale=0.75, color=black] (hacker) at (-0.52,0.22) {
  $\pi_3(b_3)$
};
\draw[->, color=black, rounded corners] (-0.94,0.15) to (-1.15,0.15) to (-1.15, 0.6) to (-0.9, 0.6);
\draw[->, color=black, rounded corners] (-0.51,0.6) to (0,0.6) to (0, 0.15) to (-0.25, 0.15);
    \end{tikzpicture}
};

%

\node[scale=1] (emulation_system) at (3.45,0.35)
{
  \begin{tikzpicture}

\node[inner sep=0pt,align=center, scale=0.75, color=black] (hacker) at (-0.52,0.22) {
  $\pi_4(b_4)$
};
\draw[->, color=black, rounded corners] (-0.94,0.15) to (-1.15,0.15) to (-1.15, 0.6) to (-0.9, 0.6);
\draw[->, color=black, rounded corners] (-0.51,0.6) to (0,0.6) to (0, 0.15) to (-0.25, 0.15);
    \end{tikzpicture}
  };

\node[scale=1] (emulation_system) at (4.8,0.35)
{
  \begin{tikzpicture}

\node[inner sep=0pt,align=center, scale=0.75, color=black] (hacker) at (-0.52,0.22) {
  $\pi_{N_t}(b_{N_t})$
};
\draw[->, color=black, rounded corners] (-1.09,0.15) to (-1.2,0.15) to (-1.2, 0.6) to (-0.85, 0.6);
\draw[->, color=black, rounded corners] (-0.51,0.6) to (0.16,0.6) to (0.16, 0.15) to (-0.06, 0.15);
    \end{tikzpicture}
};
%

  \node[cloud, fill=gray2, cloud puffs=20, cloud puff arc=100,
  minimum width=4cm, minimum height=1.75cm, aspect=1, draw=black, scale=0.8] at (2,2.25) {};

\node [scale=0.75] (node1) at (2,2.5) {
\begin{tikzpicture}[fill=white, >=stealth,
    node distance=3cm,
    database/.style={
      cylinder,
      cylinder uses custom fill,
      shape border rotate=90,
      aspect=0.25,
      draw}]
    \tikzset{
node distance = 9em and 4em,
sloped,
   box/.style = {%
    shape=rectangle,
    rounded corners,
    draw=blue!40,
    fill=blue!15,
    align=center,
    font=\fontsize{12}{12}\selectfont},
 arrow/.style = {%
    line width=0.1mm,
    shorten >=1mm, shorten <=1mm,
    font=\fontsize{8}{8}\selectfont},
}

\node[scale=1] (emulation_system) at (-0.07,-1.85)
{
\begin{tikzpicture}

\node[scale=0.9] (cloud_1) at (3,1.15)
{
  \begin{tikzpicture}

\node[server, scale=0.8](n2) at (1,0) {};


%

    \end{tikzpicture}
  };
    \end{tikzpicture}
  };
\end{tikzpicture}
};

\node[inner sep=0pt,align=center, scale=0.75, color=black] (hacker) at (-0.6,1.85) {
  Belief\\
  transmissions
};


\draw [decorate,decoration={brace,amplitude=4pt,mirror,raise=4.5pt},yshift=0pt,line width=0.15mm]
(-1.1,0.25) -- (5.4,0.25) node [black,midway,xshift=0.2cm] {};

\node[inner sep=0pt,align=center, scale=0.75, color=black] (hacker) at (2.25,-0.17) {
  Node controllers
};

\node[inner sep=0pt,align=center, scale=0.75, color=black] (hacker) at (-1.9,0.27) {
  Replicated\\
  system
};
\node[inner sep=0pt,align=center, scale=0.75, color=black] (hacker) at (2,2.1) {
  System controller
};

\node[inner sep=0pt,align=center, scale=0.75, color=black] (hacker) at (2,1.85) {
  $\pi(b_1,\hdots,b_{N_t})$
};

\draw[->, black, thick, line width=0.15mm] (-0.7,1.05) to (0.7,1.8);
\draw[->, black, thick, line width=0.15mm] (0.7,1.05) to (1.2,1.65);
\draw[->, black, thick, line width=0.15mm] (2,1.05) to (2,1.55);
\draw[->, black, thick, line width=0.15mm] (3.3,1.05) to (2.9,1.65);
\draw[->, black, thick, line width=0.15mm] (4.6,1.05) to (3.5,1.9);

\node[inner sep=0pt,align=center, scale=0.75, color=black] (hacker) at (-0.5,1.4) {
  $b_1$
};
\node[inner sep=0pt,align=center, scale=0.75, color=black] (hacker) at (0.6,1.4) {
  $b_2$
};
\node[inner sep=0pt,align=center, scale=0.75, color=black] (hacker) at (1.8,1.4) {
  $b_3$
};
\node[inner sep=0pt,align=center, scale=0.75, color=black] (hacker) at (2.9,1.4) {
  $b_4$
};
\node[inner sep=0pt,align=center, scale=0.75, color=black] (hacker) at (3.95,1.4) {
  $b_{N_t}$
};

\draw[->, black, thick, line width=0.15mm, rounded corners] (3.59,2.45) to (5.8,2.45) to (5.8, 0.5) to (5.55, 0.5);





\end{tikzpicture}

%% file: tikz/tolerance_6.tex
\begin{tikzpicture}[fill=white, >=stealth,
    node distance=3cm,
    database/.style={
      cylinder,
      cylinder uses custom fill,
      shape border rotate=90,
      aspect=0.25,
      draw}]

    \tikzset{
node distance = 9em and 4em,
sloped,
   box/.style = {%
    shape=rectangle,
    rounded corners,
    draw=blue!40,
    fill=blue!15,
    align=center,
    font=\fontsize{12}{12}\selectfont},
 arrow/.style = {%
    line width=0.1mm,
    -{Triangle[length=5mm,width=2mm]},
    shorten >=1mm, shorten <=1mm,
    font=\fontsize{8}{8}\selectfont},
}

\node [scale=1] (node1) at (4.05,0) {
\begin{tikzpicture}[fill=white, >=stealth,mybackground58={\tolerancee},
    node distance=3cm,
    database/.style={
      cylinder,
      cylinder uses custom fill,
      shape border rotate=90,
      aspect=0.25,
      draw}]
    \tikzset{
node distance = 9em and 4em,
sloped,
   box/.style = {%
    shape=rectangle,
    rounded corners,
    draw=blue!40,
    fill=blue!15,
    align=center,
    font=\fontsize{12}{12}\selectfont},
 arrow/.style = {%
    line width=0.1mm,
    shorten >=1mm, shorten <=1mm,
    font=\fontsize{8}{8}\selectfont},
}

\node [scale=0.75] (node1) at (0,0) {
\begin{tikzpicture}[fill=white, >=stealth,mybackground18={\small $\quad$ Node $1$ $\quad$},
    node distance=3cm,
    database/.style={
      cylinder,
      cylinder uses custom fill,
      shape border rotate=90,
      aspect=0.25,
      draw}]
    \tikzset{
node distance = 9em and 4em,
sloped,
   box/.style = {%
    shape=rectangle,
    rounded corners,
    draw=blue!40,
    fill=blue!15,
    align=center,
    font=\fontsize{12}{12}\selectfont},
 arrow/.style = {%
    line width=0.1mm,
    shorten >=1mm, shorten <=1mm,
    font=\fontsize{8}{8}\selectfont},
}

\draw[rounded corners=0.5ex, fill=gray2] (0.5,1.1) rectangle node (m1){} (4.5,2.45);
\node[inner sep=0pt,align=center, scale=0.9, color=black] (hacker) at (1.4,3.55) {
  \underline{Privileged domain}
};
\node[inner sep=0pt,align=center] (hacker2) at (2.5,1.6)
  {\scalebox{0.1}{
     \includegraphics{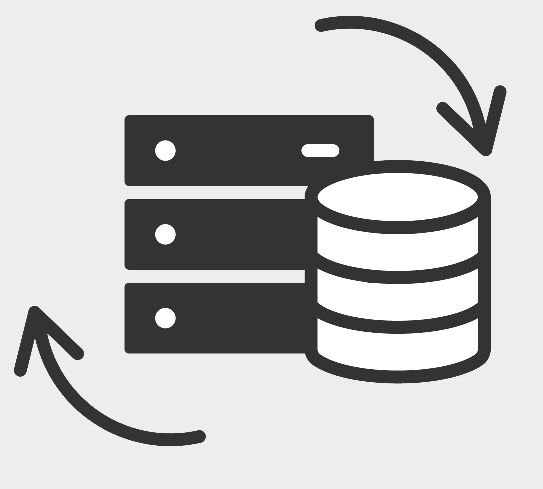}
   }
 };
\node[inner sep=0pt,align=center, scale=0.9, color=black] (hacker) at (2.5,2.2) {
  \underline{Application domain}
};

\node[inner sep=0pt,align=center, scale=0.9, color=black] (hacker) at (3.45,1.55) {
  Service\\
  replica
};
\draw[rounded corners=0.5ex, fill=white] (1.2,2.75) rectangle node (m1){} (3.8,3.25);
\node[inner sep=0pt,align=center, scale=0.9, color=black] (hacker) at (2.55,3) {
  \textit{Node controller}
};
\node[inner sep=0pt,align=center, scale=0.9, color=black] (hacker) at (0.4,2.82) {
  \ids\\alerts
};
\node[inner sep=0pt,align=center, scale=0.9, color=black] (hacker) at (4.6,2.82) {
  reco-\\
  very
};
\path[fill=white, draw=none] (0,0.25) rectangle node (m1){} (5,0.8);
\path[fill=white, draw=none] (0,-0.3) rectangle node (m1){} (5,0.25);
\draw[-, black] (0,0.8) to (5,0.8);
\draw[-, black] (0,0.25) to (5,0.25);
\draw[-, black] (0,-0.3) to (5,-0.3);

\node[inner sep=0pt,align=center, scale=0.9, color=black] (hacker) at (2.55,0.5) {
Virtualization layer
};
\node[inner sep=0pt,align=center, scale=0.9, color=black] (hacker) at (2.55,0) {
Hardware
};
\draw[->, black, thick, line width=0.2mm, rounded corners] (0.85,2.45) to (0.85,3) to (1.2, 3);
\draw[->, black, thick, line width=0.2mm, rounded corners] (3.8,3) to (4.2, 3) to (4.2, 2.45);
\draw[draw=none] (0,-0.3) rectangle node (m1){} (5,4);
\end{tikzpicture}
};

\node [scale=0.75] (node1) at (4.5,0) {
\begin{tikzpicture}[fill=white, >=stealth,mybackground18={\small $\quad$ Node $2$ $\quad$},
    node distance=3cm,
    database/.style={
      cylinder,
      cylinder uses custom fill,
      shape border rotate=90,
      aspect=0.25,
      draw}]
    \tikzset{
node distance = 9em and 4em,
sloped,
   box/.style = {%
    shape=rectangle,
    rounded corners,
    draw=blue!40,
    fill=blue!15,
    align=center,
    font=\fontsize{12}{12}\selectfont},
 arrow/.style = {%
    line width=0.1mm,
    shorten >=1mm, shorten <=1mm,
    font=\fontsize{8}{8}\selectfont},
}

\draw[rounded corners=0.5ex, fill=gray2] (0.5,1.1) rectangle node (m1){} (4.5,2.45);
\node[inner sep=0pt,align=center, scale=0.9, color=black] (hacker) at (1.4,3.55) {
  \underline{Privileged domain}
};
\node[inner sep=0pt,align=center] (hacker2) at (2.5,1.6)
  {\scalebox{0.1}{
     \includegraphics{idps_1.eps}
   }
 };
\node[inner sep=0pt,align=center, scale=0.9, color=black] (hacker) at (2.5,2.2) {
  \underline{Application domain}
};

\node[inner sep=0pt,align=center, scale=0.9, color=black] (hacker) at (3.45,1.55) {
  Service\\
  replica
};

\draw[rounded corners=0.5ex, fill=white] (1.2,2.75) rectangle node (m1){} (3.8,3.25);
\node[inner sep=0pt,align=center, scale=0.9, color=black] (hacker) at (2.55,3) {
  \textit{Node controller}
};
\node[inner sep=0pt,align=center, scale=0.9, color=black] (hacker) at (0.4,2.82) {
  \ids\\alerts
};
\node[inner sep=0pt,align=center, scale=0.9, color=black] (hacker) at (4.6,2.82) {
  reco-\\
  very
};
\path[fill=white] (0,0.25) rectangle node (m1){} (5,0.8);
\path[fill=white] (0,-0.3) rectangle node (m1){} (5,0.25);
\draw[-, black] (0,0.8) to (5,0.8);
\draw[-, black] (0,0.25) to (5,0.25);
\draw[-, black] (0,-0.3) to (5,-0.3);
\node[inner sep=0pt,align=center, scale=0.9, color=black] (hacker) at (2.55,0.5) {
Virtualization layer
};
\node[inner sep=0pt,align=center, scale=0.9, color=black] (hacker) at (2.55,0) {
Hardware
};
\draw[->, black, thick, line width=0.2mm, rounded corners] (0.85,2.45) to (0.85,3) to (1.2, 3);
\draw[->, black, thick, line width=0.2mm, rounded corners] (3.8,3) to (4.2, 3) to (4.2, 2.45);
\draw[draw=none] (0,-0.3) rectangle node (m1){} (5,4);
\end{tikzpicture}
};

\node [scale=0.75] (node1) at (10,0) {
\begin{tikzpicture}[fill=white, >=stealth,mybackground18={\small $\quad$ Node $N_t$ $\quad$},
    node distance=3cm,
    database/.style={
      cylinder,
      cylinder uses custom fill,
      shape border rotate=90,
      aspect=0.25,
      draw}]
    \tikzset{
node distance = 9em and 4em,
sloped,
   box/.style = {%
    shape=rectangle,
    rounded corners,
    draw=blue!40,
    fill=blue!15,
    align=center,
    font=\fontsize{12}{12}\selectfont},
 arrow/.style = {%
    line width=0.1mm,
    shorten >=1mm, shorten <=1mm,
    font=\fontsize{8}{8}\selectfont},
}

\draw[rounded corners=0.5ex, fill=gray2] (0.5,1.1) rectangle node (m1){} (4.5,2.45);
\node[inner sep=0pt,align=center, scale=0.9, color=black] (hacker) at (1.4,3.55) {
  \underline{Privileged domain}
};
\node[inner sep=0pt,align=center] (hacker2) at (2.5,1.6)
  {\scalebox{0.1}{
     \includegraphics{idps_1.eps}
   }
 };
\node[inner sep=0pt,align=center, scale=0.9, color=black] (hacker) at (2.5,2.2) {
  \underline{Application domain}
};

\node[inner sep=0pt,align=center, scale=0.9, color=black] (hacker) at (3.45,1.55) {
  Service\\
  replica
};
\draw[rounded corners=0.5ex, fill=white] (1.2,2.75) rectangle node (m1){} (3.8,3.25);
\node[inner sep=0pt,align=center, scale=0.9, color=black] (hacker) at (2.55,3) {
  \textit{Node controller}
};
\node[inner sep=0pt,align=center, scale=0.9, color=black] (hacker) at (0.4,2.82) {
  \ids\\alerts
};
\node[inner sep=0pt,align=center, scale=0.9, color=black] (hacker) at (4.6,2.82) {
  reco-\\
  very
};

\path[fill=white] (0,0.25) rectangle node (m1){} (5,0.8);
\path[fill=white] (0,-0.3) rectangle node (m1){} (5,0.25);
\draw[-, black] (0,0.8) to (5,0.8);
\draw[-, black] (0,0.25) to (5,0.25);
\draw[-, black] (0,-0.3) to (5,-0.3);
\node[inner sep=0pt,align=center, scale=0.9, color=black] (hacker) at (2.55,0.5) {
Virtualization layer
};
\node[inner sep=0pt,align=center, scale=0.9, color=black] (hacker) at (2.55,0) {
Hardware
};
\draw[->, black, thick, line width=0.2mm, rounded corners] (0.85,2.45) to (0.85,3) to (1.2, 3);
\draw[->, black, thick, line width=0.2mm, rounded corners] (3.8,3) to (4.2, 3) to (4.2, 2.45);
\draw[draw=none] (0,-0.3) rectangle node (m1){} (5,4);
\end{tikzpicture}
};

\node[inner sep=0pt,align=center, scale=2, color=black] (hacker) at (7.3,0) { %
  $\hdots$
};

\draw[rounded corners=0.5ex, fill=OliveGreen!50!black!2, line width=0.1mm] (-2,-1.825) rectangle node (m1){} (12.2,-2.125);
\node[inner sep=0pt,align=center, scale=0.7, color=black] (hacker) at (5.1,-2) {
  Consensus protocol
};
\draw[->, black, thick, line width=0.2mm, rounded corners] (-1.2,-0.65) to (-1.2,-1.825);
\draw[<-, black, thick, line width=0.2mm, rounded corners] (1.2,-0.65) to (1.2,-1.825);

\draw[->, black, thick, line width=0.2mm, rounded corners] (3.25,-0.65) to (3.25,-1.825);
\draw[<-, black, thick, line width=0.2mm, rounded corners] (5.7,-0.65) to (5.7,-1.825);

\draw[->, black, thick, line width=0.2mm, rounded corners] (8.8,-0.65) to (8.8,-1.825);
\draw[<-, black, thick, line width=0.2mm, rounded corners] (11.2,-0.65) to (11.2,-1.825);

\draw[rounded corners=0.5ex, fill=white, line width=0.1mm] (-2,2) rectangle node (m1){} (12.2,2.3);
\node[inner sep=0pt,align=center, scale=0.7, color=black] (hacker) at (5.05,2.13) {
  \textit{System controller}
};

\draw[->, black, thick, line width=0.2mm, rounded corners] (0.1,1) to (0.1,2);
\draw[<-, black, thick, line width=0.2mm, rounded corners] (0.7,1) to (0.7,2);

\draw[->, black, thick, line width=0.2mm, rounded corners] (4.7,1) to (4.7,2);
\draw[<-, black, thick, line width=0.2mm, rounded corners] (5.3,1) to (5.3,2);

\draw[->, black, thick, line width=0.2mm, rounded corners] (10.2,1) to (10.2,2);
\draw[<-, black, thick, line width=0.2mm, rounded corners] (10.8,1) to (10.8,2);

\node[inner sep=0pt,align=center, scale=0.7, color=black] (hacker) at (-0.7,1.85) {
  State estimate
};
\node[inner sep=0pt,align=center, scale=0.7, color=black] (hacker) at (1.5,1.85) {
  Evict or add
};

\node[inner sep=0pt,align=center, scale=0.7, color=black] (hacker) at (3.87,1.85) {
  State estimate
};
\node[inner sep=0pt,align=center, scale=0.7, color=black] (hacker) at (6.1,1.85) {
  Evict or add
};

\node[inner sep=0pt,align=center, scale=0.7, color=black] (hacker) at (9.42,1.85) {
  State estimate
};
\node[inner sep=0pt,align=center, scale=0.7, color=black] (hacker) at (11.6,1.8) {
  Evict or add
};

\draw[draw=none] (-2.4,-2.26) rectangle node (m1){} (12.6,2.63);
\end{tikzpicture}
};

\node[inner sep=0pt,align=center] (client1) at (3,-3.15)
  {\scalebox{0.055}{
     \includegraphics{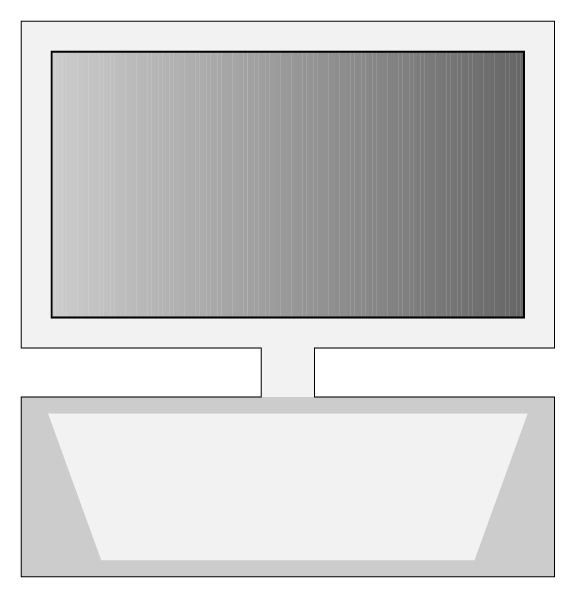}
   }
 };
\node[inner sep=0pt,align=center] (client1) at (3.75,-3.15)
  {\scalebox{0.055}{
     \includegraphics{laptop3.eps}
   }
 };

\node[inner sep=0pt,align=center] (client1) at (5,-3.15)
  {\scalebox{0.055}{
     \includegraphics{laptop3.eps}
   }
 };
\node[inner sep=0pt,align=center, scale=1.2, color=black] (hacker) at (4.4,-3.15) {
  $\hdots$
};

\draw[rounded corners=0.5ex, line width=0.1mm] (2.5,-3.45) rectangle node (m1){} (5.5,-2.85);

\draw[->, black, thick, line width=0.2mm, rounded corners] (3.6,-2.85) to (3.6,-2.395);

\node[inner sep=0pt,align=center, scale=0.7, color=black] (hacker) at (2.7,-2.7) {
  Service requests
};
\draw[<-, black, thick, line width=0.2mm, rounded corners] (4.3,-2.85) to (4.3,-2.395);
\node[inner sep=0pt,align=center, scale=0.7, color=black] (hacker) at (5.05,-2.7) {
  Responses
};

\node[inner sep=0pt,align=center, scale=0.7, color=black] (hacker) at (4.05,-3.6) {
  Clients
};

\node[inner sep=0pt,align=center] (hacker2) at (8,-3.15)
  {\scalebox{0.12}{
     \includegraphics{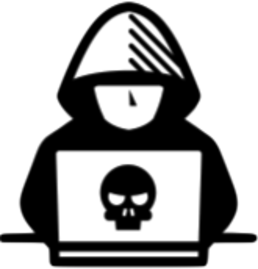}
   }
 };
\node[inner sep=0pt,align=center, scale=0.7, color=black] (hacker) at (8,-3.6) {
  Attacker
};
\draw[->, black, thick, line width=0.2mm, rounded corners] (7.925,-2.88) to (7.925,-2.395);
\node[inner sep=0pt,align=center, scale=0.7, color=black] (hacker) at (9,-2.7) {
  Intrusion attempts
};
\end{tikzpicture}

%% file: tikz/state_transition_node.tex
\begin{tikzpicture}[fill=white, >=stealth,
    node distance=3cm,
    database/.style={
      cylinder,
      cylinder uses custom fill,
      shape border rotate=90,
      aspect=0.25,
      draw}]

\node[scale=0.8] (kth_cr) at (0,2.15)
{
  \begin{tikzpicture}

\node[scale=1] (level1) at (-1.7,-5.6)
{
  \begin{tikzpicture}
\node[draw,circle, minimum width=15mm, scale=0.6](s0) at (0,0) {};
\node[draw,circle, minimum width=15mm, scale=0.6](s1) at (4,0) {};
\node[draw,circle, minimum width=15mm, scale=0.6](s2) at (2,-1.2) {};
\node[draw,circle, minimum width=15mm, scale=0.5](s4) at (2,-1.2) {};

\node[inner sep=0pt,align=center, scale=1] (time) at (0.07,0)
{
$\mathbb{H}$
};

\node[inner sep=0pt,align=center, scale=1] (time) at (4.07,0)
{
$\mathbb{C}$
};

\node[inner sep=0pt,align=center, scale=1] (time) at (2.07,-1.2)
{
$\emptyset$
};

\node[inner sep=0pt,align=center, scale=1] (time) at (2.07,-1.86)
{
Crashed
};

\node[inner sep=0pt,align=center, scale=1] (time) at (0.07,0.65)
{
Healthy
};

\node[inner sep=0pt,align=center, scale=1] (time) at (4.07,0.65)
{
Compromised
};

\node[inner sep=0pt,align=center, scale=1] (time) at (0.55,-0.68)
{
$p_{\mathrm{C}_1,i}$
};

\node[inner sep=0pt,align=center, scale=1] (time) at (3.52,-0.7)
{
$p_{\mathrm{C}_2,i}$
};

\node[inner sep=0pt,align=center, scale=1] (time) at (2.145,-0.35)
{
$p_{\mathrm{U},i}$ or $a_i=\mathfrak{R}$
};

\node[inner sep=0pt,align=center, scale=1] (time) at (2.15,0.35)
{
$p_{\mathrm{A},i}$
};

\draw[thick,-{Latex[length=2mm]}] (0.43,0.1) to (3.55,0.1);
\draw[thick,-{Latex[length=2mm]}] (3.55,-0.1) to (0.43,-0.1);
\draw[thick,-{Latex[length=2mm]}] (s1) to (s2);
\draw[thick,-{Latex[length=2mm]}] (s0) to (s2);

    \end{tikzpicture}
  };
    \end{tikzpicture}
  };

\end{tikzpicture}

%% file: tikz/value_fun.tex
\pgfplotstableread{
0.0 0.22717
0.05263157894736842 0.4054973684210526
0.10526315789473684 0.5659931578947368
0.15789473684210525 0.7159199999999999
0.21052631578947367 0.8575747368421053
0.2631578947368421 0.9920815789473684
0.3157894736842105 1.1199873684210526
0.3684210526315789 1.2270852631578946
0.42105263157894735 1.227073157894737
0.47368421052631576 1.227061052631579
0.5263157894736842 1.227048947368421
0.5789473684210527 1.227036842105263
0.631578947368421 1.2270247368421052
0.6842105263157894 1.2270126315789474
0.7368421052631579 1.2270005263157895
0.7894736842105263 1.2269884210526316
0.8421052631578947 1.2269763157894735
0.894736842105263 1.2269642105263157
0.9473684210526315 1.2269521052631578
1.0 1.22694
}\valuefunzeroone
\begin{tikzpicture}[
    dot/.style={
        draw=black,
        fill=blue!90,
        circle,
        minimum size=3pt,
        inner sep=0pt,
        solid,
    },
    ]

\node[scale=1] (kth_cr) at (0,0)
{
  \begin{tikzpicture}
    \begin{axis}[
      xmin=0,
      grid=major,
      grid style={dashed},
      xmax=1,
      ymin=-0.6,
      ymax=4,
        axis lines=center,
        xlabel style={below right},
        ylabel style={above left},
        axis line style={-{Latex[length=2mm]}},
        smooth,
        legend style={at={(0.56,0.97)}},
        legend columns=2,
        legend style={
            /tikz/column 2/.style={
                column sep=5pt,
              },
              draw=none
            },
            width=11cm,
            height=3.5cm
            ]
            \addplot[Red!50!white,mark repeat=10,mark size=1.3pt,samples=100,smooth, name path=l1, domain=0:1] (x, {(1-x)*0.22717 + x*3.93854});
\addplot[Black, thick, mark repeat=10,mark size=1.3pt,samples=100,smooth] table [x index=0, y index=1] {\valuefunzeroone};
\addplot[Red!50!white,mark repeat=10,mark size=1.3pt,samples=100,smooth, name path=l1, domain=0:1] (x, {(1-x)*0.23795 + x*3.42135});
\addplot[Red!50!white,mark repeat=10,mark size=1.3pt,samples=100,smooth, name path=l1, domain=0:1] (x, {(1-x)*0.25623 + x*3.19898});
\addplot[Red!50!white,mark repeat=10,mark size=1.3pt,samples=100,smooth, name path=l1, domain=0:1] (x, {(1-x)*0.27909 + x*3.04568});
\addplot[Red!50!white,mark repeat=10,mark size=1.3pt,samples=100,smooth, name path=l1, domain=0:1] (x, {(1-x)*0.30592 + x*2.92628});
\addplot[Red!50!white,mark repeat=10,mark size=1.3pt,samples=100,smooth, name path=l1, domain=0:1] (x, {(1-x)*0.3359 + x*2.82939});
\addplot[Red!50!white,mark repeat=10,mark size=1.3pt,samples=100,smooth, name path=l1, domain=0:1] (x, {(1-x)*0.3724 + x*2.73976});
\addplot[Red!50!white,mark repeat=10,mark size=1.3pt,samples=100,smooth, name path=l1, domain=0:1] (x, {(1-x)*1.22717 + x*1.22694});
\addplot[Red!50!white,mark repeat=10,mark size=1.3pt,samples=100,smooth, name path=l1, domain=0:1] (x, {(1-x)*1.22717 + x*1.22694});
\addplot[Red!50!white,mark repeat=10,mark size=1.3pt,samples=100,smooth, name path=l1, domain=0:1] (x, {(1-x)*1.22717 + x*1.22694});
\addplot[Red!50!white,mark repeat=10,mark size=1.3pt,samples=100,smooth, name path=l1, domain=0:1] (x, {(1-x)*1.22717 + x*1.22694});
\addplot[Red!50!white,mark repeat=10,mark size=1.3pt,samples=100,smooth, name path=l1, domain=0:1] (x, {(1-x)*1.22717 + x*1.22694});
\addplot[Red!50!white,mark repeat=10,mark size=1.3pt,samples=100,smooth, name path=l1, domain=0:1] (x, {(1-x)*1.22717 + x*1.22694});
\addplot[Red!50!white,mark repeat=10,mark size=1.3pt,samples=100,smooth, name path=l1, domain=0:1] (x, {(1-x)*1.22717 + x*1.22694});
\addplot[Red!50!white,mark repeat=10,mark size=1.3pt,samples=100,smooth, name path=l1, domain=0:1] (x, {(1-x)*1.22717 + x*1.22694});
\addplot[Red!50!white,mark repeat=10,mark size=1.3pt,samples=100,smooth, name path=l1, domain=0:1] (x, {(1-x)*1.22717 + x*1.22694});
\addplot[Red!50!white,mark repeat=10,mark size=1.3pt,samples=100,smooth, name path=l1, domain=0:1] (x, {(1-x)*1.22717 + x*1.22694});
\addplot[Red!50!white,mark repeat=10,mark size=1.3pt,samples=100,smooth, name path=l1, domain=0:1] (x, {(1-x)*1.22717 + x*1.22694});
\addplot[Red!50!white,mark repeat=10,mark size=1.3pt,samples=100,smooth, name path=l1, domain=0:1] (x, {(1-x)*1.22717 + x*1.22694});
\addplot[Red!50!white,mark repeat=10,mark size=1.3pt,samples=100,smooth, name path=l1, domain=0:1] (x, {(1-x)*1.22717 + x*1.22694});
\addplot[Black, thick, mark repeat=10,mark size=1.3pt,samples=100,smooth] table [x index=0, y index=1] {\valuefunzeroone};
\legend{alpha vectors, $V^{\star}_{i,1}(b_{i,1})$}
\end{axis}
\node[inner sep=0pt,align=center, scale=1, rotate=0, opacity=1] (obs) at (9.825,0.23)
{
  $b_{i,1}$
};
\end{tikzpicture}
};

  \end{tikzpicture}

%% file: tikz/reliability_curve_1.tex
\begin{tikzpicture}[
    dot/.style={
        draw=black,
        fill=blue!90,
        circle,
        minimum size=3pt,
        inner sep=0pt,
        solid,
    },
    ]
\tikzset{
        hatch distance/.store in=\hatchdistance,
        hatch distance=10pt,
        hatch thickness/.store in=\hatchthickness,
        hatch thickness=2pt
      }
\pgfdeclarepatternformonly[\hatchdistance,\hatchthickness]{flexible hatch}
    {\pgfqpoint{0pt}{0pt}}
    {\pgfqpoint{\hatchdistance}{\hatchdistance}}
    {\pgfpoint{\hatchdistance-1pt}{\hatchdistance-1pt}}%
    {
        \pgfsetcolor{\tikz@pattern@color}
        \pgfsetlinewidth{\hatchthickness}
        \pgfpathmoveto{\pgfqpoint{0pt}{0pt}}
        \pgfpathlineto{\pgfqpoint{\hatchdistance}{\hatchdistance}}
        \pgfusepath{stroke}
      }
\node[scale=1] (kth_cr) at (0,0)
{
  \begin{tikzpicture}[declare function={sigma(\x)=1/(1+exp(-\x));
sigmap(\x)=sigma(\x)*(1-sigma(\x));}]
\begin{axis}[
        xmin=1,
        xmax=100,
        ymin=0,
        ymax=1.05,
        width =1.2\columnwidth,
        height = 0.4\columnwidth,
        axis lines=center,
        xmajorgrids=true,
        ymajorgrids=true,
        major grid style = {lightgray},
        minor grid style = {lightgray!25},
        scaled y ticks=false,
        yticklabel style={
        /pgf/number format/fixed,
        /pgf/number format/precision=5
        },
        xlabel style={below right},
        ylabel style={above left},
        axis line style={-{Latex[length=2mm]}},
        smooth,
        legend style={at={(1.05,-0.245)}},
        legend columns=5,
        legend style={
          draw=none,
            /tikz/column 2/.style={
                column sep=0pt,
              }
              }
              ]
              \addplot[Blue,mark=diamond,smooth, name path=l1, thick, domain=1:100]   (x,{1-(1-0.10001)^(x)});

              \addplot[Red,mark=triangle,mark repeat=1, smooth, name path=l1, thick, domain=1:100]   (x,{1-(1-0.05001)^(x)});

              \addplot[Black,mark=square,mark repeat=1, smooth, name path=l1, thick, domain=1:100]   (x,{1-(1-0.025001)^(x)});

              \addplot[OliveGreen,mark=pentagon,mark repeat=1, smooth, name path=l1, thick, domain=1:100]   (x,{1-(1-0.010001)^(x)});


\legend{$p_{\mathrm{A},i}=0.1$,$p_{\mathrm{A},i}=0.05$, $p_{\mathrm{A},i}=0.025$, $p_{\mathrm{A},i}=0.01$}
\end{axis}
\end{tikzpicture}
};

\node[inner sep=0pt,align=center, scale=1, rotate=0, opacity=1] (obs) at (4.95,-0.5)
{
  $t$
};

\node[inner sep=0pt,align=center, scale=1, rotate=0, opacity=1] (obs) at (-1.23,1.64)
{
  $\mathbb{P}[S_{i,t}=\mathbb{C} \cup S_{i,t}=\emptyset \mid \pi_{i,t}(b)=\mathfrak{W} \text{ }\forall t,b]$
};

  \end{tikzpicture}

%% file: tikz/mttf.tex
\begin{tikzpicture}[
    dot/.style={
        draw=black,
        fill=blue!90,
        circle,
        minimum size=3pt,
        inner sep=0pt,
        solid,
    },
    ]

\node[scale=1] (kth_cr) at (0,2.15)
{
  \begin{tikzpicture}[declare function={sigma(\x)=1/(1+exp(-\x));
      sigmap(\x)=sigma(\x)*(1-sigma(\x));}]
\pgfplotstableread{
4 2.9076001572726677
5 4.770933711089396
6 6.3521771699648415
7 7.7081170247921085
8 8.894417204131274
9 9.948895076503222
10 10.897926861504418
11 11.760683337035962
12 12.551543417442916
13 13.28156809623083
14 13.959448154950596
15 14.59213621019097
16 15.1852812620361
17 15.743535428462012
18 16.270775474524967
19 16.770266044479538
20 17.244782085936883
21 17.696702125420135
22 18.128080344926843
23 18.54070298967239
24 18.936133024220194
25 19.315745857386098
26 19.6807581969687
27 20.032251561011204
28 20.371191590623617
29 20.698444033008006
30 21.014788060646254
31 21.320927442231664
32 21.617499968142525
33 21.90508544781366
34 22.184212531023878
35 22.45536455471381
36 22.71898457774569
37 22.97547973529022
38 23.225225020267786
39 23.468566579989517
40 23.705824600718206
41 23.937295840453515
42 24.163255860195125
43 24.38396099575668
44 24.59965010551004
45 24.81054612393553
46 25.016857446308315
47 25.2187791660774
48 25.41649418335132
49 25.6101742002727
50 25.799980616855652
51 25.986065338995793
52 26.168571508787082
53 26.347634165940818
54 26.523380847962066
55 26.695932135764753
56 26.865402150570958
57 27.03189900722267
58 27.195525228414866
59 27.356378123824147
60 27.514550137643266
61 27.67012916762929
62 27.82319885842199
63 27.973838871583062
64 28.12212513453849
65 28.26813007037153
66 28.411922810207095
67 28.553569389746624
68 28.69313293135173
69 28.830673812933583
70 28.96624982477853
71 29.099916315329896
72 29.231726326845845
73 29.361730721765667
74 29.48997830053793
75 29.616515911593225
76 29.741388554082025
77 29.864639473941093
78 29.98631025380194
79 30.10644089720887
80 30.22506990757322
81 30.34223436225407
82 30.45796998212173
83 30.5723111969307
84 30.68529120680153
85 30.79694204008563
86 30.907294607866397
87 31.01637875532788
88 31.124223310204545
89 31.230856128509554
90 31.336304137722323
91 31.440593377603065
92 31.54374903878944
93 31.645795499317902
94 31.746756359202458
95 31.846654473193507
96 31.945511981830442
97 32.043350340893824
98 32.140190349354505
99 32.23605217591156
100 32.33095538420302
101 32.42491895676883
102 32.51796131783894
103 32.6101003550151
104 32.70135343991076
105 32.791737447807414
106 32.88126877638426
107 32.96996336357257
108 33.0578367045832
109 33.14490386815335
110 33.23117951205469
111 33.31667789790287
112 33.40141290530596
113 33.48539804538692
114 33.56864647371279
115 33.651171002661876
116 33.73298411325799
117 33.81409796649857
118 33.89452441420322
119 33.974275009406135
120 34.053361016315684
121 34.13179341986235
122 34.20958293485537
123 34.28674001476714
124 34.36327486016349
125 34.439197426796646
126 34.51451743337723
127 34.5892443690398
128 34.663387500517516
129 34.736955879038035
130 34.80995834695455
131 34.88240354412361
132 34.95429991404138
133 35.02565570974928
134 35.09647899951902
135 35.16677767232753
136 35.23655944313009
137 35.30583185794137
138 35.37460229873231
139 35.442877988150634
140 35.5106659940731
141 35.57797323399616
142 35.644806479271814
143 35.71117235919595
144 35.77707736495391
145 35.842527853430774
146 35.9075300508907
147 35.972090056531144
148 36.03621384591728
149 36.09990727430082
150 36.16317607982849
151 36.22602588664403
152 36.28846220788842
153 36.35049044860177
154 36.41211590853131
155 36.4733437848484
156 36.53417917477885
157 36.59462707814917
158 36.65469239985266
159 36.71437995223724
160 36.7736944574194
161 36.83264054952592
162 36.891222776866336
163 36.94944560403899
164 37.007313413972824
165 37.06483050990707
166 37.12200111731158
167 37.17882938574959
168 37.235319390684985
169 37.29147513523614
170 37.34730055187822
171 37.40279950409546
172 37.457975787985866
173 37.51283313381908
174 37.56737520754984
175 37.62160561228777
176 37.675527889726105
177 37.72914552152923
178 37.78246193068172
179 37.83548048279988
180 37.88820448740623
181 37.94063719916951
182 37.99278181910987
183 38.044641495771884
184 38.09621932636504
185 38.14751835787397
186 38.198541588138184
187 38.24929196690372
188 38.29977239684595
189 38.34998573456633
190 38.39993479156185
191 38.44962233516995
192 38.49905108948842
193 38.54822373627157
194 38.597142915803275
195 38.64581122774761
196 38.69423123197795
197 38.74240544938478
198 38.79033636266332
199 38.838026417081146
200 38.88547802122689
201 38.93269354774003
202 38.97967533402295
203 39.02642568293501
204 39.07294686347004
205 39.11924111141713
206 39.1653106300052
207 39.21115759053247
208 39.25678413298031
209 39.3021923666126
210 39.34738437056092
211 39.39236219439572
212 39.43712785868415
213 39.481683355534635
214 39.526030649128785
215 39.57017167624109
216 39.6141083467464
217 39.65784254411576
218 39.70137612590084
219 39.744710924207425
220 39.78784874615812
221 39.830791374344294
222 39.87354056726839
223 39.916098059775805
224 39.958465563477326
225 40.00064476716241
226 40.04263733720291
227 40.08444491794804
228 40.12606913211097
229 40.16751158114656
230 40.208773845621124
231 40.24985748557411
232 40.29076404087213
233 40.3314950315552
234 40.372051958175504
235 40.41243630212929
236 40.45264952598165
237 40.49269307378394
238 40.53256837138543
239 40.57227682673751
240 40.61181983019229
241 40.65119875479458
242 40.69041495656788
243 40.72946977479486
244 40.76836453229135
245 40.80710053567562
246 40.845679075631494
247 40.88410142716652
248 40.92236884986465
249 40.96048258813435
250 40.99844387145096
251 41.03625391459496
252 41.07391391788522
253 41.11142506740755
254 41.148788535238815
255 41.18600547966682
256 41.22307704540573
257 41.2600043638071
258 41.29678855306735
259 41.333430718430826
260 41.369931952389095
261 41.406293334876246
262 41.44251593346078
263 41.47860080353361
264 41.51454898849251
265 41.550361519923236
266 41.586039417777194
267 41.621583690545506
268 41.656995335430395
269 41.69227533851272
270 41.72742467491697
271 41.762444308972874
272 41.79733519437412
273 41.832098274334385
274 41.86673448174006
275 41.90124473930055
276 41.93562995969604
277 41.969891045721845
278 42.00402889043102
279 42.03804437727383
280 42.07193838023506
281 42.10571176396866
282 42.13936538393019
283 42.172900086506694
284 42.20631670914453
285 42.23961608047488
286 42.272799020436956
287 42.30586634039912
288 42.33881884327813
289 42.371657323655796
290 42.40438256789424
291 42.4369953542487
292 42.469496452978646
293 42.50188662645699
294 42.534166629277195
295 42.56633720835905
296 42.598399103052124
297 42.63035304523778
298 42.662199759429576
299 42.69393996287156
300 42.72393996287156
}\datatablee
\pgfplotstableread{
4 5.390070516064278
5 9.271400207055208
6 12.519998514938413
7 15.30458139815138
8 17.741072393679865
9 19.90684183425978
10 21.85603439554437
11 23.628027635483335
12 25.25235477174424
13 26.75173366673031
14 28.14401406921991
15 29.443475778210527
16 30.661721130389182
17 31.80830499126321
18 32.891189748755366
19 33.91708057164267
20 34.891676853385604
21 35.81986378837885
22 36.70586040814515
23 37.55333543574769
24 38.365499003866795
25 39.14517602926114
26 39.89486547675571
27 40.616788648417135
28 41.3129288496621
29 41.98506421638135
30 42.63479507087662
31 43.26356686554949
32 43.87268954163883
33 44.46335395481635
34 45.03664588525335
35 45.59355804624935
36 46.13500042499542
37 46.661809225937546
38 47.17475463738122
39 47.67454760237756
40 48.16184574324904
41 48.637258563611425
42 49.10135203110807
43 49.55465262726753
44 49.99765093715071
45 50.430804840147566
46 50.854542353948816
47 51.26926417596713
48 51.67534596002666
49 52.07314036073808
50 52.462978873435276
51 52.84517349372659
52 53.22001821747386
53 53.587790399263646
54 53.948751985094376
55 54.30315063300089
56 54.65122073362338
57 54.993184341252466
58 55.32925202461209
59 55.6596236455419
60 55.984489072789536
61 56.30402883729542
62 56.61841473463184
63 56.92781037962962
64 57.23237171767426
65 57.532247496672085
66 57.82757970326086
67 58.11850396646769
68 58.40514993168622
69 58.68764160755373
70 58.966097688051725
71 59.24063185192296
72 59.51135304129593
73 59.77836572122553
74 60.04177012169658
75 60.30166246349468
76 60.55813516921657
77 60.811277060578355
78 61.06117354307652
79 61.30790677896083
80 61.55155584939653
81 61.792196906617036
82 62.02990331679824
83 62.26474579432662
84 62.49679252807495
85 62.72610930024975
86 62.95275959832952
87 63.17680472056926
88 63.398303875510834
89 63.61731427590246
90 63.83389122740088
91 64.04808821239939
92 64.2599569693
93 64.46954756752427
94 64.6769084785334
95 64.88208664311091
96 65.08512753514067
97 65.28607522209795
98 65.48497242245365
99 65.68186056017949
100 65.87677981652811
101 66.06976917924943
102 66.26086648939514
103 66.45010848585005
104 66.63753084772371
105 66.82316823472237
106 67.00705432561729
107 67.189221854915
108 67.36970264783034
109 67.54852765365476
110 67.72572697760805
111 67.90132991125542
112 68.07536496156662
113 68.24785987868924
114 68.41884168250387
115 68.58833668802434
116 68.75637052970416
117 68.92296818470294
118 69.08815399516784
119 69.25195168957846
120 69.41438440320225
121 69.57547469770519
122 69.73524457995813
123 69.89371552007893
124 70.05090846874718
125 70.20684387382602
126 70.3615416963249
127 70.51502142573324
128 70.66730209475556
129 70.81840229347543
130 70.96834018297429
131 71.11713350843124
132 71.26479961172562
133 71.4113554435667
134 71.55681757517011
135 71.70120220950237
136 71.84452519211166
137 71.98680202156312
138 72.12804785949692
139 72.26827754032325
140 72.40750558057215
141 72.54574618791163
142 72.6830132698472
143 72.81932044211902
144 72.9546810368055
145 73.08910811014935
146 73.22261445011411
147 73.3552125836846
148 73.48691478392017
149 73.61773307677157
150 73.74767924767063
151 73.87676484790143
152 74.00500120076238
153 74.13239940752612
154 74.258970353207
155 74.38472471214168
156 74.5096729533908
157 74.63382534596948
158 74.75719196391162
159 74.87978269117482
160 75.0016072263927
161 75.12267508747873
162 75.242995616089
163 75.362577981947
164 75.4814311870377
165 75.59956406967316
166 75.7169853084374
167 75.83370342601138
168 75.94972679288553
169 76.06506363096165
170 76.17972201704903
171 76.29370988625868
172 76.40703503529856
173 76.51970512567347
174 76.63172768679343
175 76.74311011899265
176 76.85385969646343
177 76.96398357010663
178 77.07348877030243
179 77.18238220960336
180 77.29067068535257
181 77.39836088223022
182 77.50545937472941
183 77.61197262956478
184 77.71790700801509
185 77.82326876820355
186 77.92806406731572
187 78.03229896375878
188 78.13597941926332
189 78.23911130092922
190 78.34170038321795
191 78.4437523498926
192 78.54527279590748
193 78.64626722924878
194 78.74674107272743
195 78.8466996657267
196 78.94614826590451
197 79.04509205085311
198 79.14353611971599
199 79.24148549476554
200 79.33894512293982
201 79.43591987734209
202 79.5324145587028
203 79.62843389680557
204 79.72398255187838
205 79.81906511595085
206 79.91368611417835
207 80.00785000613419
208 80.101561187071
209 80.19482398915167
210 80.28764268265097
211 80.38002147712893
212 80.47196452257643
213 80.56347591053353
214 80.65455967518237
215 80.74521979441425
216 80.835460190872
217 80.9252847329681
218 81.01469723588028
219 81.10370146252347
220 81.19230112450006
221 81.28049988302891
222 81.36830134985257
223 81.45570908812545
224 81.54272661328103
225 81.62935739388044
226 81.7156048524418
227 81.80147236625173
228 81.88696326815894
229 81.97208084735054
230 82.0568283501108
231 82.14120898056473
232 82.2252259014046
233 82.30888223460143
234 82.39218106210083
235 82.47512542650445
236 82.55771833173695
237 82.63996274369836
238 82.72186159090364
239 82.80341776510815
240 82.88463412191996
241 82.96551348140079
242 83.04605862865219
243 83.12627231439237
244 83.2061572555188
245 83.28571613566105
246 83.36495160572152
247 83.44386628440516
248 83.52246275873928
249 83.60074358458209
250 83.67871128712156
251 83.75636836136398
252 83.8337172726134
253 83.91076045694086
254 83.98750032164507
255 84.06393924570332
256 84.14007958021452
257 84.21592364883264
258 84.29147374819254
259 84.36673214832715
260 84.44170109307662
261 84.51638280048988
262 84.59077946321833
263 84.66489324890215
264 84.73872630054936
265 84.81228073690734
266 84.88555865282788
267 84.95856211962501
268 85.03129318542686
269 85.10375387551927
270 85.1759461926854
271 85.24787211753727
272 85.31953360884185
273 85.39093260384132
274 85.46207101856707
275 85.5329507481484
276 85.60357366711527
277 85.67394162969597
278 85.7440564701092
279 85.81392000285062
280 85.88353402297517
281 85.9529003063732
282 86.02202061004292
283 86.09089667235688
284 86.15953021332464
285 86.22792293485051
286 86.29607652098636
287 86.36399263818103
288 86.43167293552433
289 86.49911904498747
290 86.56633258165944
291 86.63331514397846
292 86.70006831396083
293 86.76659365742447
294 86.83289272420973
295 86.89896704839572
296 86.96481814851346
297 87.03044752775547
298 87.09585667418112
299 87.16104706091983
300 87.22602014636928
}\datatableee

\pgfplotstableread{
4 10.379015022169002
5 18.26703701347727
6 24.847412827956617
7 30.487743445251684
8 35.42303038559195
9 39.80995210036714
10 43.75818164570327
11 47.34748123239048
12 50.63767252018415
13 53.67477217045519
14 56.49493613142117
15 59.12708916165607
16 61.59473262750131
17 63.917220595355595
18 66.11068145388465
19 68.18869700407008
20 70.16281177674628
21 72.04292108405683
22 73.83757087739876
23 75.55419241885629
24 77.19928806275307
25 78.77857988089399
26 80.29712970602947
27 81.75943694504878
28 83.16951892553168
29 84.53097738944628
30 85.84705390456367
31 87.12067633854828
32 88.35449807147084
33 89.55093126703208
34 90.7121752509592
35 91.84024083534553
36 92.93697126461004
37 94.00406033092143
38 95.04306810601416
39 96.05543465610445
40 97.04249204244248
41 98.00547485838204
42 98.94552951203735
43 99.86372242956105
44 100.761047326232
45 101.63843166964357
46 102.49674244037234
47 103.33679127980895
48 104.15933910175733
49 104.96510023346188
50 105.75474614253233
51 106.52890879848366
52 107.2881837110513
53 108.03313268187253
54 108.7642863013822
55 109.48214621871891
56 110.18718720896037
57 110.87985905902218
58 111.56058829097942
59 112.22977973934418
60 112.88781799690283
61 113.53506874204257
62 114.1718799590349
63 114.7985830614717
64 115.41549392793297
65 116.02291385798718
66 116.62113045576777
67 117.21041844761143
68 117.79104043957494
69 118.36324762006079
70 118.927280412254
71 119.48336908061344
72 120.03173429524571
73 120.57258765762275
74 121.10613219077845
75 121.63256279682537
76 122.15206668437172
77 122.66482376818368
78 123.17100704322878
79 123.6707829350456
80 124.16431162821458
81 124.65174737455436
82 125.1332387825242
83 125.60892908919304
84 126.07895641602072
85 126.54345400959164
86 127.00255046835342
87 127.45636995632503
88 127.90503240466046
89 128.348653701891
90 128.78734587359685
91 129.22121725220703
92 129.65037263757137
93 130.07491344889957
94 130.49493786861785
95 130.91054097865495
96 131.32181488962908
97 131.72884886337675
98 132.13172942922913
99 132.53054049441607
100 132.92536344895132
101 133.31627726532287
102 133.7033585932986
103 134.08668185012885
104 134.46631930641283
105 134.8423411678749
106 135.2148156532855
107 135.58380906873893
108 135.9493858784938
109 136.3116087725629
110 136.67053873123123
111 137.02623508666838
112 137.3787555817892
113 137.72815642651057
114 138.07449235154152
115 138.41781665983297
116 138.7581812758116
117 139.09563679250837
118 139.43023251669078
119 139.76201651209848
120 140.09103564087786
121 140.41733560330366
122 140.74096097587352
123 141.06195524785338
124 141.38036085634943
125 141.6962192199777
126 142.00957077119617
127 142.3204549873656
128 142.6289104205963
129 142.93497472643753
130 143.23868469146458
131 143.5400762598121
132 143.83918455870239
133 144.13604392301463
134 144.43068791893646
135 144.7231493667403
136 145.01346036272216
137 145.30165230033901
138 145.58775589058192
139 145.87180118161453
140 146.15381757771104
141 146.4338338575233
142 146.71187819170302
143 146.9879781599094
144 147.2621607672256
145 147.53445246000857
146 147.80487914119695
147 148.0734661850986
148 148.34023845167636
149 148.6052203003578
150 148.86843560338127
151 149.1299077587026
152 149.3896597024757
153 149.64771392112627
154 149.90409246303233
155 150.15881694982923
156 150.41190858735175
157 150.66338817622776
158 150.91327612213612
159 151.16159244574317
160 151.4083567923278
161 151.65358844110736
162 151.89730631427722
163 152.13952898577116
164 152.38027468975613
165 152.61956132886831
166 152.8574064822028
167 153.0938274130623
168 153.32884107647607
169 153.56246412649685
170 153.7947129232824
171 154.02560353996964
172 154.2551517693505
173 154.4833731303537
174 154.7102828743394
175 154.9358959912166
176 155.1602272153844
177 155.38329103150593
178 155.60510168012124
179 155.82567316310192
180 156.0450192489548
181 156.26315347797967
182 156.48008916728472
183 156.69583941566475
184 156.91041710834696
185 157.12383492160916
186 157.3361053272734
187 157.54724059707823
188 157.75725280693734
189 157.96615384108307
190 158.1739553961016
191 158.38066898486343
192 158.58630594035066
193 158.79087741938443
194 158.99439440625815
195 159.19686771627633
196 159.3983079992024
197 159.59872574262124
198 159.79813127521487
199 159.99653476995616
200 160.19394624722366
201 160.3903755778382
202 160.58583248602412
203 160.78032655229748
204 160.9738672162853
205 161.16646377947325
206 161.35812540788845
207 161.54886113471696
208 161.73867986285904
209 161.9275903674212
210 162.11560129815226
211 162.30272118181833
212 162.4889584245236
213 162.67432131397675
214 162.85881802170357
215 163.04245660520832
216 163.22524501008567
217 163.40719107208352
218 163.58830251911812
219 163.76858697324366
220 163.9480519525778
221 164.12670487318204
222 164.3045530509006
223 164.48160370315847
224 164.65786395071902
225 164.83334081940112
226 165.00804124176182
227 165.18197205873784
228 165.3551400212533
229 165.52755179179263
230 165.6992139459385
231 165.87013297387574
232 166.04031528186502
233 166.2097671936829
234 166.3784949520312
235 166.54650471991846
236 166.71380258200975
237 166.88039454594866
238 167.04628654365254
239 167.21148443257945
240 167.37599399696913
241 167.5398209490584
242 167.70297093027116
243 167.86544951238457
244 168.0272621986695
245 168.1884144250103
246 168.34891156100036
247 168.50875891101458
248 168.66796171526258
249 168.8265251508191
250 168.98445433263305
251 169.14175431451977
252 169.29843009012905
253 169.45448659389777
254 169.6099287019825
255 169.76476123317292
256 169.91898894978814
257 170.07261655855666
258 170.22564871147733
259 170.37809000666468
260 170.52994498917818
261 170.68121815183522
262 170.83191393600916
263 170.98203673241042
264 171.1315908818555
265 171.2805806760198
266 171.429010358176
267 171.5768841239194
268 171.72420612188048
269 171.87098045442127
270 172.0172111783231
271 172.16290230545812
272 172.30805780344892
273 172.45268159631894
274 172.59677756512747
275 172.74034954859485
276 172.8834013437163
277 173.02593670636443
278 173.16795935188068
279 173.30947295565676
280 173.45048115370494
281 173.59098754321943
282 173.73099568312543
283 173.870509094622
284 174.00953126171189
285 174.14806563172417
286 174.28611561582738
287 174.4236845895331
288 174.5607758931912
289 174.6973928324767
290 174.83353867886825
291 174.9692166701174
292 175.10443001071158
293 175.23918187232772
294 175.37347539427844
295 175.50731368395142
296 175.64069981724046
297 175.7736368389695
298 175.90612776331005
299 176.03817557419134
300 176.16978322570316
}\datatableeee
\pgfplotstableread{
4 25.353418243085784
5 45.23011705344704
6 61.79682655480542
7 75.99686380340452
8 88.42189624763452
9 99.46636953129544
10 109.40639548661079
11 118.44278271871558
12 126.7261376814782
13 134.3723114932594
14 141.47233003277034
15 148.0990140029805
16 154.31153022505265
17 160.15860431641448
18 165.68084095825628
19 170.91243356631688
20 175.88244654397454
21 180.61579223698155
22 185.13398585303378
23 189.45573626838828
24 193.59741374976971
25 197.57342413189554
26 201.39651103778587
27 205.07800213234702
28 208.62801140210226
29 212.0556065591074
30 215.3689485442122
31 218.5754085297976
32 221.68166664083355
33 224.69379571820164
34 227.61733276388233
35 230.45734017968667
36 233.21845850060754
37 235.904952002044
38 238.52074830607427
39 241.06947291000097
40 243.55447939882953
41 245.97887597329682
42 248.3455488198004
43 250.65718276289684
44 252.91627957092288
45 255.12517422765941
46 257.2860494353363
47 259.40094857476504
48 261.4717873154555
49 263.50036404103
50 265.488369232093
51 267.43739392921356
52 269.34893738215857
53 271.2244139775009
54 273.06515952478134
55 274.87243697120215
56 276.6474416060797
57 278.3913058087665
58 280.10510338726897
59 281.78985354918655
60 283.446524541739
61 285.07603699343
62 286.6792669862228
63 288.2570488838916
64 289.8101779394096
65 291.33941270176564
66 292.8454772404494
67 294.32906320392925
68 295.7908317267696
69 297.2314151985544
70 298.65141890645634
71 300.05142256213463
72 301.4319817225949
73 302.7936291137337
74 304.13687586445184
75 305.4622126584938
76 306.77011081050864
77 308.06102327223783
78 309.3353855742011
79 310.5936167077852
80 311.83611995219945
81 313.0632836503864
82 314.27548193761993
83 315.47307542621184
84 316.6564118494635
85 317.82582666773584
86 318.98164363928396
87 320.12417535828547
88 321.2537237622985
89 322.3705806112102
90 323.4750279395783
91 324.56733848411835
92 325.64777608795674
93 326.7165960831517
94 327.7740456528661
95 328.8203641744781
96 329.8557835448234
97 330.8805284886703
98 331.89481685145734
99 332.89885987724654
100 333.8928624727782
101 334.87702345845264
102 335.8515358070129
103 336.8165868706358
104 337.7723585971083
105 338.7190277357097
106 339.6567660333806
107 340.58574042172796
108 341.50611319536813
109 342.41804218209415
110 343.32168090530456
111 344.2171787391167
112 345.1046810565553
113 345.9843293711848
114 346.85626147252805
115 347.72061155559896
116 348.57751034485017
117 349.4270852128258
118 350.26946029378456
119 351.1047565925504
120 351.93309208882664
121 352.7545818371996
122 353.569338063045
123 354.37747025453405
124 355.17908525093026
125 355.97428732735546
126 356.76317827618993
127 357.5458574852698
128 358.32242201302864
129 359.09296666072714
130 359.85758404190517
131 360.61636464918104
132 361.369396918523
133 362.116767291103
134 362.8585602728427
135 363.59485849175496
136 364.3257427531752
137 365.051292092979
138 365.77158382887137
139 366.48669360982933
140 367.19669546378015
141 367.9016618435897
142 368.60166367142875
143 369.2967703815903
144 369.98704996182045
145 370.67256899322166
146 371.35339268879096
147 372.02958493064904
148 372.7012083060082
149 373.3683241419353
150 374.0309925389564
151 374.6892724035467
152 375.343221479554
153 375.9928963785941
154 376.6383526094589
155 377.27964460657574
156 377.9168257575573
157 378.5499484298704
158 379.17906399666236
159 379.8042228617765
160 380.4254744839835
161 381.0428674004627
162 381.65644924955615
163 382.26626679282697
164 382.87236593644366
165 383.4747917519172
166 384.0735884962133
167 384.66879963126223
168 385.260467842888
169 385.8486350591788
170 386.4333424683149
171 387.0146305358769
172 387.59253902165113
173 388.16710699594665
174 388.7383728554475
175 389.3063743386083
176 389.8711485406148
177 390.4327319279209
178 390.9911603523766
179 391.5464690649641
180 392.09869272914835
181 392.6478654338618
182 393.194020706132
183 393.7371915233621
184 394.2774103252813
185 394.81470902556873
186 395.3491190231664
187 395.8806712132902
188 396.4093959981471
189 396.9353232973702
190 397.4584825581763
191 397.97890276526084
192 398.4966124504335
193 399.01163970200435
194 399.5240121739278
195 400.0337570947133
196 400.540901276107
197 401.0454711215544
198 401.5474926344492
199 402.0469914261736
200 402.5439927239392
201 403.0385213784325
202 403.53060187126994
203 404.02025832227054
204 404.50751449655087
205 404.9923938114438
206 405.4749193432553
207 405.9551138338503
208 406.4329996970863
209 406.90859902509203
210 407.3819335943927
211 407.85302487189557
212 408.3218940207314
213 408.7885619059571
214 409.25304910013085
215 409.7153758887502
216 410.1755622755702
217 410.6336279877968
218 411.08959248115997
219 411.5434749448729
220 411.9952943064782
221 412.4450692365828
222 412.892818153489
223 413.3385592277183
224 413.78231038643787
225 414.22408931778494
226 414.6639134750998
227 415.1018000810608
228 415.53776613173255
229 415.971828400524
230 416.40400344205915
231 416.834307595969
232 417.26275699059454
233 417.68936754661655
234 418.1141549806044
235 418.5371348084898
236 418.95832234896943
237 419.37773272683086
238 419.79538087621376
239 420.2112815438002
240 420.6254492919383
241 421.0378985017024
242 421.4486433758889
243 421.8576979419511
244 422.2650760548741
245 422.6707913999888
246 423.0748574957334
247 423.4772876963535
248 423.8780951945517
249 424.27729302408244
250 424.67489406229515
251 425.07091103262616
252 425.4653565070435
253 425.85824290843914
254 426.24958251297903
255 426.639387452403
256 427.02766971628256
257 427.41444115423246
258 427.7997134780818
259 428.1834982640014
260 428.56580695459036
261 428.94665086092436
262 429.32604116456196
263 429.7039889195169
264 430.080505054188
265 430.4556003732565
266 430.82928555954635
267 431.2015711758503
268 431.5724676667201
269 431.9419853602263
270 432.3101344696823
271 432.6769250953397
272 433.04236722605003
273 433.40647074089657
274 433.7692454107985
275 434.13070090008273
276 434.490846768029
277 434.84969247038674
278 435.20724736086584
279 435.5635206925975
280 435.91852161957314
281 436.272259198054
282 436.6247423879589
283 436.9759800542243
284 437.3259809681439
285 437.67475380868126
286 438.0223071637623
287 438.36864953154327
288 438.71378932165857
289 439.05773485644426
290 439.4004943721447
291 439.7420760200936
292 440.0824878678784
293 440.4217379004829
294 440.7598340214119
295 441.09678405379566
296 441.43259574147515
297 441.7672767500719
298 442.10083466803536
299 442.43327700767804
300 442.7646112061884
}\datatableeeee

\pgfplotstableread{
4 50.27726714486965
5 90.09606738758619
6 123.2797959851353
7 151.72299199225063
8 176.6107884800057
9 198.73327424689546
10 218.6435114370969
11 236.74372706455222
12 253.33559138971958
13 268.6511584591045
14 282.8727564521048
15 296.1462479122377
16 308.5901461561128
17 320.30205038564105
18 331.3632932690849
19 341.8423654744525
20 351.79748406955133
21 361.2785493982169
22 370.3286572119424
23 378.9852820772454
24 387.28121423982725
25 395.2453091159057
26 402.9030926505962
27 410.2772545728906
28 417.38805356938894
29 424.25365260049057
30 430.8903983305553
31 437.31305548868227
32 443.53500461061765
33 449.5684098197668
34 455.4243619345291
35 461.11300113172666
36 466.6436225734468
37 472.0247677599854
38 477.2643038626671
39 482.3694928857926
40 487.3470521833403
41 492.2032075955816
42 496.94374025991294
43 501.5740279785613
44 506.0990818854223
45 510.5235790387976
46 514.8518914714472
47 519.0881121502106
48 523.2360782315001
49 527.299391943783
50 531.2814393818204
51 535.185407458327
52 539.0142992256699
53 542.77094775212
54 546.4580287132654
55 550.0780718387533
56 553.6334713370004
57 557.1264954054537
58 560.5592949210028
59 563.9339113939154
60 567.2522842589457
61 570.5162575688113
62 573.7275861478736
63 576.8879412574264
64 579.9989158183922
65 583.0620292322659
66 586.0787318368385
67 589.0504090294027
68 591.9783850867824
69 594.8639267085474
70 597.7082463071446
71 600.5125050663247
72 603.2778157871832
73 606.005245539262
74 608.6958181325294
75 611.3505164245527
76 613.9702844758916
77 616.5560295655255
78 619.1086240770866
79 621.6289072657166
80 624.117686914489
81 626.5757408885848
82 629.0038185947035
83 631.4026423525554
84 633.772908684719
85 636.1152895306221
86 638.4304333899447
87 640.7189664003093
88 642.9814933537376
89 645.2185986560045
90 647.4308472326898
91 649.6187853854561
92 651.782941601779
93 653.9238273211516
94 656.0419376605319
95 658.1377521016029
96 660.2117351422459
97 662.2643369144282
98 664.2959937705687
99 666.3071288402826
100 668.2981525592998
101 670.2694631721876
102 672.2214472104394
103 674.154479947349
104 676.0689258310192
105 677.9651388967496
106 679.8434631599727
107 681.7042329908298
108 683.5477734714011
109 685.3744007365544
110 687.1844222992964
111 688.978137361474
112 690.755837110596
113 692.5178050035307
114 694.2643170377561
115 695.995642010814
116 697.712041768587
117 699.4137714429608
118 701.1010796794156
119 702.7742088550598
120 704.4333952875735
121 706.0788694355211
122 707.7108560904522
123 709.3295745611973
124 710.9352388507269
125 712.5280578259398
126 714.1082353807142
127 715.6759705925385
128 717.2314578730202
129 718.7748871125672
130 720.3064438195028
131 721.8263092538666
132 723.3346605561514
133 724.8316708712014
134 726.317509467482
135 727.7923418519385
136 729.2563298806272
137 730.7096318653105
138 732.1524026761919
139 733.5847938409519
140 735.0069536402491
141 736.419027199835
142 737.8211565794237
143 739.2134808584559
144 740.596136218884
145 741.9692560251018
146 743.3329709011399
147 744.6874088052325
148 746.0326951018649
149 747.3689526314065
150 748.6963017774176
151 750.0148605317323
152 751.3247445574011
153 752.6260672495683
154 753.9189397943841
155 755.203471226007
156 756.4797684817863
157 757.7479364556818
158 759.0080780499959
159 760.2602942254774
160 761.5046840498627
161 762.7413447449032
162 763.9703717319501
163 765.1918586761318
164 766.40589752919
165 767.612578571018
166 768.8119904499431
167 770.0042202218093
168 771.1893533878902
169 772.3674739316864
170 773.5386643546367
171 774.7030057107866
172 775.8605776404472
173 777.0114584028844
174 778.1557249080656
175 779.2934527475032
176 780.4247162242164
177 781.5495883818525
178 782.668141032985
179 783.7804447866254
180 784.8865690749674
181 785.9865821793961
182 787.0805512557781
183 788.1685423590657
184 789.250620467226
185 790.3268495045317
186 791.3972923642175
187 792.4620109305366
188 793.521066100225
189 794.5745178034081
190 795.6224250239427
191 796.6648458192384
192 797.7018373395588
193 798.733455846821
194 799.7597567329117
195 800.780794537535
196 801.796622965604
197 802.8072949041905
198 803.812862439047
199 804.813376870713
200 805.8088887302209
201 806.799447794408
202 807.7851031008507
203 808.7659029624349
204 809.7418949815601
205 810.7131260640068
206 811.6796424324608
207 812.6414896397148
208 813.5987125815492
209 814.5513555093075
210 815.4994620421721
211 816.443075179146
212 817.3822373107566
213 818.3169902304824
214 819.2473751459106
215 820.1734326896386
216 821.0952029299236
217 822.0127253810826
218 822.9260390136586
219 823.8351822643509
220 824.7401930457215
221 825.6411087556836
222 826.537966286772
223 827.4308020352089
224 828.3196519097695
225 829.2045513404425
226 830.0855352869098
227 830.9626382468289
228 831.83589426394
229 832.7053369359994
230 833.5709994225282
231 834.4329144524049
232 835.2911143312907
233 836.1456309488944
234 836.9964957860802
235 837.8437399218311
236 838.6873940400586
237 839.5274884362678
238 840.3640530240899
239 841.197117341669
240 842.0267105579258
241 842.8528614786787
242 843.6755985526521
243 844.49494987735
244 845.310943204815
245 846.1236059472706
246 846.9329651826426
247 847.7390476599768
248 848.5418798047414
249 849.341487724025
250 850.1378972116312
251 850.9311337530715
252 851.7212225304585
253 852.5081884273027
254 853.2920560332144
255 854.0728496485144
256 854.8505932887548
257 855.6253106891497
258 856.3970253089235
259 857.1657603355703
260 857.931538689038
261 858.6943830258249
262 859.4543157430066
263 860.2113589821759
264 860.9655346333185
265 861.7168643386075
266 862.4653694961322
267 863.2110712635534
268 863.9539905616938
269 864.6941480780562
270 865.4315642702841
271 866.1662593695523
272 866.8982533838958
273 867.6275661014836
274 868.3542170938251
275 869.0782257189222
276 869.7996111243626
277 870.5183922503612
278 871.2345878327407
279 871.9482164058644
280 872.6592963055127
281 873.3678456717108
282 874.0738824515038
283 874.7774244016856
284 875.4784890914798
285 876.1770939051693
286 876.8732560446856
287 877.5669925321475
288 878.2583202123612
289 878.9472557552731
290 879.6338156583823
291 880.3180162491092
292 880.9998736871288
293 881.6794039666559
294 882.3566229187018
295 883.0315462132831
296 883.7041893615992
297 884.3745677181707
298 885.0426964829409
299 885.7085907033472
300 886.3722652763523
}\datatableeeeee
\begin{axis}[
        xmin=4,
        xmax=100,
        ymin=0,
        ymax=350,
        width =1.2\columnwidth,
        height = 0.4\columnwidth,
        axis lines=center,
        xmajorgrids=true,
        ymajorgrids=true,
        major grid style = {lightgray},
        minor grid style = {lightgray!25},
        scaled y ticks=false,
        yticklabel style={
        /pgf/number format/fixed,
        /pgf/number format/precision=5
        },
        xlabel style={below right},
        ylabel style={above left},
        axis line style={-{Latex[length=2mm]}},
        smooth,
        legend style={at={(0.93,-0.245)}},
        legend columns=5,
        legend style={
             draw=none,
            /tikz/column 2/.style={
                column sep=5pt,
              }
              }
              ]
\addplot[Blue,mark=diamond, mark repeat=3, name path=l1, thick, domain=4:100] table [x index=0, y index=1] {\datatablee};
\addplot[Red,mark=square, mark repeat=3, name path=l1, thick, domain=4:100] table [x index=0, y index=1] {\datatableeee};
\addplot[Black,mark=pentagon, mark repeat=3, name path=l1, thick, domain=4:100] table [x index=0, y index=1] {\datatableeeee};
\legend{$p_{\mathrm{A},i}=0.1$,$p_{\mathrm{A},i}=0.025$, $p_{\mathrm{A},i}=0.01$}
\end{axis}
\end{tikzpicture}
};

\node[inner sep=0pt,align=center, scale=1, rotate=0, opacity=1] (obs) at (5.15,1.7)
{
  $N_1$
};
\node[inner sep=0pt,align=center, scale=1, rotate=0, opacity=1] (obs) at (-3.7,3.82)
{
  $\mathbb{E}[T^{(f)}]$
};

\end{tikzpicture}

%% file: tikz/reliability_curve_2.tex
\begin{tikzpicture}[
    dot/.style={
        draw=black,
        fill=blue!90,
        circle,
        minimum size=3pt,
        inner sep=0pt,
        solid,
    },
    ]

\node[scale=1] (kth_cr) at (0,2.15)
{
  \begin{tikzpicture}[declare function={sigma(\x)=1/(1+exp(-\x));
      sigmap(\x)=sigma(\x)*(1-sigma(\x));}]
\pgfplotstableread{
1   1.000000
2   1.000000
3   1.000000
4   1.000000
5   1.000000
6   1.000000
7   1.000000
8   1.000000
9   1.000000
10   0.999999
11   0.999996
12   0.999985
13   0.999953
14   0.999875
15   0.999697
16   0.999332
17   0.998640
18   0.997422
19   0.995409
20   0.992263
21   0.987592
22   0.980962
23   0.971930
24   0.960072
25   0.945020
26   0.926489
27   0.904305
28   0.878419
29   0.848918
30   0.816016
31   0.780048
32   0.741447
33   0.700723
34   0.658434
35   0.615160
36   0.571480
37   0.527949
38   0.485076
39   0.443318
40   0.403062
41   0.364629
42   0.328266
43   0.294152
44   0.262401
45   0.233069
46   0.206160
47   0.181635
48   0.159421
49   0.139417
50   0.121501
51   0.105538
52   0.091382
53   0.078888
54   0.067906
55   0.058294
56   0.049912
57   0.042630
58   0.036324
59   0.030882
60   0.026199
61   0.022181
62   0.018743
63   0.015809
64   0.013311
65   0.011189
66   0.009391
67   0.007870
68   0.006585
69   0.005503
70   0.004593
71   0.003829
72   0.003188
73   0.002651
74   0.002203
75   0.001828
76   0.001516
77   0.001256
78   0.001040
79   0.000860
80   0.000711
81   0.000587
82   0.000484
83   0.000400
84   0.000329
85   0.000271
86   0.000223
87   0.000184
88   0.000151
89   0.000124
90   0.000102
91   0.000084
92   0.000069
93   0.000056
94   0.000046
95   0.000038
96   0.000031
97   0.000026
98   0.000021
99   0.000017
100   0.000014
}\datatablee
\pgfplotstableread{
1   1.000000
2   1.000000
3   1.000000
4   1.000000
5   1.000000
6   1.000000
7   1.000000
8   1.000000
9   1.000000
10   1.000000
11   1.000000
12   1.000000
13   1.000000
14   1.000000
15   1.000000
16   1.000000
17   1.000000
18   1.000000
19   1.000000
20   0.999999
21   0.999997
22   0.999992
23   0.999979
24   0.999949
25   0.999884
26   0.999754
27   0.999509
28   0.999070
29   0.998326
30   0.997121
31   0.995254
32   0.992474
33   0.988486
34   0.982960
35   0.975546
36   0.965891
37   0.953662
38   0.938568
39   0.920386
40   0.898970
41   0.874275
42   0.846355
43   0.815369
44   0.781572
45   0.745304
46   0.706976
47   0.667049
48   0.626013
49   0.584371
50   0.542615
51   0.501214
52   0.460597
53   0.421147
54   0.383189
55   0.346989
56   0.312755
57   0.280634
58   0.250721
59   0.223060
60   0.197651
61   0.174458
62   0.153413
63   0.134424
64   0.117382
65   0.102164
66   0.088640
67   0.076674
68   0.066133
69   0.056885
70   0.048801
71   0.041761
72   0.035652
73   0.030366
74   0.025808
75   0.021889
76   0.018528
77   0.015654
78   0.013202
79   0.011115
80   0.009342
81   0.007841
82   0.006570
83   0.005498
84   0.004595
85   0.003835
86   0.003197
87   0.002662
88   0.002214
89   0.001840
90   0.001527
91   0.001267
92   0.001049
93   0.000869
94   0.000719
95   0.000594
96   0.000491
97   0.000405
98   0.000334
99   0.000275
100   0.000227
}\datatableee

\pgfplotstableread{
1   1.000000
2   1.000000
3   1.000000
4   1.000000
5   1.000000
6   1.000000
7   1.000000
8   1.000000
9   1.000000
10   1.000000
11   1.000000
12   1.000000
13   1.000000
14   1.000000
15   1.000000
16   1.000000
17   1.000000
18   1.000000
19   1.000000
20   1.000000
21   1.000000
22   1.000000
23   1.000000
24   1.000000
25   1.000000
26   1.000000
27   1.000000
28   1.000000
29   1.000000
30   1.000000
31   1.000000
32   0.999999
33   0.999997
34   0.999992
35   0.999981
36   0.999956
37   0.999906
38   0.999808
39   0.999629
40   0.999315
41   0.998787
42   0.997935
43   0.996613
44   0.994634
45   0.991771
46   0.987757
47   0.982298
48   0.975076
49   0.965775
50   0.954088
51   0.939744
52   0.922524
53   0.902277
54   0.878934
55   0.852513
56   0.823129
57   0.790982
58   0.756357
59   0.719607
60   0.681141
61   0.641403
62   0.600858
63   0.559973
64   0.519200
65   0.478965
66   0.439653
67   0.401603
68   0.365101
69   0.330378
70   0.297611
71   0.266923
72   0.238386
73   0.212030
74   0.187845
75   0.165785
76   0.145782
77   0.127741
78   0.111555
79   0.097105
80   0.084263
81   0.072903
82   0.062895
83   0.054112
84   0.046435
85   0.039748
86   0.033943
87   0.028920
88   0.024587
89   0.020860
90   0.017662
91   0.014927
92   0.012593
93   0.010606
94   0.008917
95   0.007486
96   0.006275
97   0.005253
98   0.004391
99   0.003666
100   0.003057
}\datatableeee
\pgfplotstableread{
1   1.000000
2   1.000000
3   1.000000
4   1.000000
5   1.000000
6   1.000000
7   1.000000
8   1.000000
9   1.000000
10   1.000000
11   1.000000
12   1.000000
13   1.000000
14   1.000000
15   1.000000
16   1.000000
17   1.000000
18   1.000000
19   1.000000
20   1.000000
21   1.000000
22   1.000000
23   1.000000
24   1.000000
25   1.000000
26   1.000000
27   1.000000
28   1.000000
29   1.000000
30   1.000000
31   1.000000
32   1.000000
33   1.000000
34   1.000000
35   1.000000
36   1.000000
37   1.000000
38   1.000000
39   1.000000
40   1.000000
41   1.000000
42   1.000000
43   1.000000
44   0.999999
45   0.999998
46   0.999996
47   0.999990
48   0.999977
49   0.999950
50   0.999898
51   0.999798
52   0.999621
53   0.999315
54   0.998808
55   0.998002
56   0.996762
57   0.994920
58   0.992267
59   0.988559
60   0.983522
61   0.976860
62   0.968269
63   0.957452
64   0.944137
65   0.928096
66   0.909156
67   0.887220
68   0.862271
69   0.834378
70   0.803700
71   0.770474
72   0.735010
73   0.697679
74   0.658895
75   0.619099
76   0.578744
77   0.538277
78   0.498125
79   0.458683
80   0.420307
81   0.383304
82   0.347929
83   0.314385
84   0.282820
85   0.253334
86   0.225982
87   0.200774
88   0.177689
89   0.156671
90   0.137643
91   0.120509
92   0.105158
93   0.091470
94   0.079320
95   0.068583
96   0.059134
97   0.050849
98   0.043613
99   0.037315
100   0.031851
}\datatableeeee

\pgfplotstableread{
1   1.000000
2   1.000000
3   1.000000
4   1.000000
5   1.000000
6   1.000000
7   1.000000
8   1.000000
9   1.000000
10   1.000000
11   1.000000
12   1.000000
13   1.000000
14   1.000000
15   1.000000
16   1.000000
17   1.000000
18   1.000000
19   1.000000
20   1.000000
21   1.000000
22   1.000000
23   1.000000
24   1.000000
25   1.000000
26   1.000000
27   1.000000
28   1.000000
29   1.000000
30   1.000000
31   1.000000
32   1.000000
33   1.000000
34   1.000000
35   1.000000
36   1.000000
37   1.000000
38   1.000000
39   1.000000
40   1.000000
41   1.000000
42   1.000000
43   1.000000
44   1.000000
45   1.000000
46   1.000000
47   1.000000
48   1.000000
49   1.000000
50   1.000000
51   1.000000
52   1.000000
53   1.000000
54   1.000000
55   1.000000
56   1.000000
57   0.999999
58   0.999998
59   0.999996
60   0.999991
61   0.999980
62   0.999956
63   0.999911
64   0.999825
65   0.999671
66   0.999407
67   0.998968
68   0.998267
69   0.997186
70   0.995572
71   0.993236
72   0.989951
73   0.985462
74   0.979487
75   0.971730
76   0.961897
77   0.949710
78   0.934925
79   0.917347
80   0.896847
81   0.873372
82   0.846949
83   0.817693
84   0.785800
85   0.751542
86   0.715255
87   0.677325
88   0.638177
89   0.598250
90   0.557989
91   0.517827
92   0.478171
93   0.439392
94   0.401820
95   0.365734
96   0.331361
97   0.298877
98   0.268408
99   0.240030
100   0.213778
}\datatableeeeee
\begin{axis}[
        xmin=1,
        xmax=100,
        ymin=0,
        ymax=1,
        width =1.2\columnwidth,
        height = 0.4\columnwidth,
        axis lines=center,
        xmajorgrids=true,
        ymajorgrids=true,
        major grid style = {lightgray},
        minor grid style = {lightgray!25},
        scaled y ticks=false,
        yticklabel style={
        /pgf/number format/fixed,
        /pgf/number format/precision=5
        },
        xlabel style={below right},
        ylabel style={above left},
        axis line style={-{Latex[length=2mm]}},
        smooth,
        legend style={at={(1,-0.25)}},
        legend columns=5,
        legend style={
          draw=none,
            /tikz/column 2/.style={
                column sep=5pt,
              }
              }
              ]
\addplot[Blue,mark=diamond, mark repeat=3, name path=l1, thick, domain=1:100] table [x index=0, y index=1] {\datatablee};
\addplot[Red,mark=triangle, mark repeat=3, name path=l1, thick, domain=1:100] table [x index=0, y index=1] {\datatableee};
\addplot[Black,mark=square, mark repeat=3, name path=l1, thick, domain=1:100] table [x index=0, y index=1] {\datatableeee};
\addplot[OliveGreen,mark=pentagon, mark repeat=3, name path=l1, thick, domain=1:100] table [x index=0, y index=1] {\datatableeeee};
\legend{$N_1=25$, $N_1=50$, $N_1=100$, $N_1=200$}
\end{axis}
\end{tikzpicture}
};

\node[inner sep=0pt,align=center, scale=1, rotate=0, opacity=1] (obs) at (5.04,1.6)
{
  $t$
};
\node[inner sep=0pt,align=center, scale=1, rotate=0, opacity=1] (obs) at (-4,3.9)
{
  $R(t)$
};

  \end{tikzpicture}

%% file: tikz/pomdp_times.tex
\begin{tikzpicture}

\pgfplotstableread{
1.04 0.05
7.8 0.1
8.84 0.1
12.25 0.1
14.48 0.4
}\datatablee

\pgfplotstableread{
2.34 0.1
7.98 0.2
8.98 0.3
12.9 0.5
15.65 1
}\datatableee

\pgfplotstableread{
29.18 4
40.18 6
62.57 7
75.5 12
90.25 15
}\datatableeee

\pgfplotstableread{
10.78 9
40.78 15
88.78 18
100.78 22
123.78 28
}\datatableeeee

\pgfplotstableread{
11.12 3
87 4
237 10
466 15
743 17
}\datatableeeeee

%
\begin{axis}[
   ybar=-0.55cm,
    title style={align=center},
    ticks=both,
    xmax=4,
    ymax=1120,
    ymode=log,
    log basis y={10},
    axis x line = bottom,
    axis y line = left,
    axis line style={-|},
    nodes near coords align={vertical}
    ylabel=Convergence time (min),
    xtick=data,
    ymajorgrids,
    xticklabels={
        $5$, $10$, $15$, $20$, $25$
      },
    legend style={at={(0.19, 1.32)}, anchor=north, legend columns=3, draw=none},
    every axis legend/.append style={nodes={right}, inner sep = 0.2cm},
   x tick label style={align=center, yshift=-0.1cm},
   enlarge x limits=0.15,
    width=11cm,
    height=3.5cm,
    bar width=0.2cm,
    ]

\addplot[draw=black,fill=Periwinkle!40,postaction={
        pattern=crosshatch
      }] plot [error bars/.cd, y dir=both, y explicit] table [x expr=\coordindex, y index=0, y error plus index=1, y error minus index=1] {\datatablee};

\addplot[draw=black,fill=RedViolet!30,postaction={
        pattern=north west lines
      }] plot [error bars/.cd, y dir=both, y explicit] table [x expr=\coordindex, y index=0, y error plus index=1, y error minus index=1] {\datatableee};
\addplot[draw=black,fill=RawSienna!60,postaction={
        pattern=grid
      }] plot [error bars/.cd, y dir=both, y explicit] table [x expr=\coordindex, y index=0, y error plus index=1, y error minus index=1] {\datatableeee};
\addplot[draw=black,fill=SeaGreen!60,postaction={
        pattern=north east lines
      }] plot [error bars/.cd, y dir=both, y explicit] table [x expr=\coordindex, y index=0, y error plus index=1, y error minus index=1] {\datatableeeee};
\addplot[draw=black,fill=Red,postaction={
        pattern=bricks
    }] plot [error bars/.cd, y dir=both, y explicit] table [x expr=\coordindex, y index=0, y error plus index=1, y error minus index=1] {\datatableeeeee};

    \legend{\textsc{cem}, \textsc{de}, \textsc{bo}, \textsc{spsa}, \textsc{ip}}
  \end{axis}
\node[inner sep=0pt,align=center, scale=1, rotate=0, opacity=1] (obs) at (5.1,-0.72)
{
  $\Delta_{\mathrm{R}}$
};
\node[inner sep=0pt,align=center, scale=1, rotate=0, opacity=1, rotate=90] (obs) at (-0.9,1.02)
{
  Time (min)
};
\end{tikzpicture}

%% file: tikz/cmdp_times.tex
\begin{tikzpicture}

\pgfplotstableread{
0.05368 0.000071
0.3324394 0.002386
2.582967 0.12
18.960885 1.64
}\datatable
%

\pgfplotstableread{
0.0658 0.0002
0.0639 0.001
0.0745 0.005
0.108 0.05
0.2696 0.09
0.9682 0.1
3.0928 0.25
9.0928 0.71
29.41 3.12
69.41 6.12
}\datatablee

\pgfplotstableread{
1926.2546 30.3
1926.2546 30.3
1926.2546 30.3
1926.2546 30.3
}\datatableee

\pgfplotstableread{
0.00145 0.0001
0.0035 0.0002
0.0051 0.0001
0.032315 0.0001
}\datatableeee

%
\begin{axis}[
  spy using outlines={rectangle, magnification=3,
   width=3cm,height=2.5cm,connect spies},
   ybar,
    title style={align=center},
    ticks=both,
    xmin=1.5,
    xmax=7.5,
    ymax=75,
    ymin=10^(-1.5),
    ymode=log,
    log basis y={10},
    axis x line = bottom,
    axis y line = left,
    axis line style={-|},
    enlarge y limits={lower, value=0.1},
    enlarge y limits={upper, value=0.22},
    xtick=data,
    ymajorgrids,
    xticklabels={
        $4$,
        $8$,
        $16$,
        $32$,
        $64$,
        $128$,
        $256$,
        $512$,
        $1024$,
        $2048$,
      },
    legend style={at={(0.42, -0.15)}, anchor=north, legend columns=2, draw=none},
    every axis legend/.append style={nodes={right}, inner sep = 0.2cm},
   x tick label style={align=center, yshift=-0.1cm},
    enlarge x limits=0.3,
    width=10.5cm,
    height=3.25cm,
    bar width=0.3cm
    ]

\addplot[draw=black,fill=Periwinkle!40,postaction={
        pattern=crosshatch
    }] plot [error bars/.cd, y dir=both, y explicit] table [x expr=\coordindex, y index=0, y error plus index=1, y error minus index=1] {\datatablee};
  \end{axis}
\node[inner sep=0pt,align=center, scale=1, rotate=0, opacity=1] (obs) at (5.1,-0.75)
{
  $s_{\mathrm{max}}$
};
\node[inner sep=0pt,align=center, scale=1, rotate=0, opacity=1, rotate=90] (obs) at (-0.9,1.02)
{
  Time (s)
};
\end{tikzpicture}

%% file: tikz/tolerance_24.tex
\begin{tikzpicture}[
    dot/.style={
        draw=black,
        fill=blue!90,
        circle,
        minimum size=3pt,
        inner sep=0pt,
        solid,
    },
    ]

\node[scale=1] (kth_cr) at (0,2.15)
{
  \begin{tikzpicture}[declare function={sigma(\x)=1/(1+exp(-\x));
      sigmap(\x)=sigma(\x)*(1-sigma(\x));}]
\pgfplotstableread{
3   11    1
4   5.6   0.6
5   3.45   0.5
6   2.45   0.3
7   1.81   0.2
8   1.35   0.2
9   1.00   0.1
10   0.75   0.05
}\datatablee

\pgfplotstableread{
3   60    3.5
4   29   2
5   16   2.5
6   12   1
7   9.5   0.9
8   6.5   0.7
9   5   0.6
10   3.25   0.5
}\datatableee

\begin{axis}[
        xmin=2.9,
        xmax=10.4,
        ymin=0,
        ymax=69,
        width =1.1\columnwidth,
        height = 0.4\columnwidth,
        axis lines=center,
        xmajorgrids=true,
        ymajorgrids=true,
        major grid style = {lightgray},
        minor grid style = {lightgray!25},
        scaled y ticks=false,
        yticklabel style={
        /pgf/number format/fixed,
        /pgf/number format/precision=5
        },
        xlabel style={below right},
        ylabel style={above left},
        axis line style={-{Latex[length=2mm]}},
        smooth,
        legend style={at={(0.96,0.7)}},
        legend columns=3,
        legend style={
          draw=none,
            /tikz/column 2/.style={
                column sep=5pt,
              }
              }
              ]
\addplot[Black,mark=diamond, mark repeat=1, name path=l1, thick, domain=1:100] plot[error bars/.cd, y dir=both, y explicit] table [x index=0, y index=1, y error plus index=2, y error minus index=2] {\datatablee};
\addplot[Gray,mark=triangle, mark repeat=1, name path=l1, thick, domain=1:100] plot[error bars/.cd, y dir=both, y explicit] table [x index=0, y index=1, y error plus index=2, y error minus index=2] {\datatableee};
\legend{$1$ client, $20$ clients}
\end{axis}
\end{tikzpicture}
};

\node[inner sep=0pt,align=center, scale=0.9, rotate=0, opacity=1] (obs) at (0.6,0.75)
{
 Number of nodes ($N$)
};
\node[inner sep=0pt,align=center, scale=0.9, rotate=90, opacity=1] (obs) at (-4.55,2.4)
{
  \# Requests/s
};

  \end{tikzpicture}

%% file: tikz/tolerance_41.tex
\begin{tikzpicture}[
    dot/.style={
        draw=black,
        fill=blue!90,
        circle,
        minimum size=3pt,
        inner sep=0pt,
        solid,
    },
    ]

\node[scale=1] (kth_cr) at (0,-0.345)
{
  \begin{tikzpicture}[declare function={sigma(\x)=1/(1+exp(-\x));
      sigmap(\x)=sigma(\x)*(1-sigma(\x));}]
\pgfplotstableread{
5   0.08   0.029
15   0.08   0.029
25   0.08   0.029
35   0.08   0.029
}\norecovone

\pgfplotstableread{
5   0.981   0.005
15   0.967   0.01
25   0.933   0.002
35   0.08   0.029
}\recovone

\pgfplotstableread{
5   0.975   0.008
15   0.954   0.02
25   0.941   0.02
35   0.09   0.04
}\recovadapt

\pgfplotstableread{
5   0.986   0.0055
15   0.986   0.0055
25   0.986  0.0055
35   0.986   0.0055
}\controlone

\begin{axis}[
        xmin=0,
        xmax=40,
        ymin=0,
        ymax=1.2,
        ybar=-0.4cm,
        width =7cm,
        height = 2.75cm,
        axis x line=bottom,
        axis y line=left,
        axis line style={-|},
        nodes near coords align={vertical}
        xtick={},
        xticklabels={},
       bar width=0.17cm,
    legend style={at={(0.105, 1)}, anchor=north, legend columns=1},
    every axis legend/.append style={nodes={right}, inner sep = 0.2cm},
   x tick label style={align=center, yshift=-0.1cm},
    enlarge x limits=0.02,
    ]
\addplot[draw=black,fill=RedViolet!30,postaction={
        pattern=north east lines
      }] plot [error bars/.cd, y dir=both, y explicit] table [x index=0, y index=1, y error plus index=2, y error minus index=2] {\recovadapt};
\addplot[draw=black,fill=SeaGreen!60,postaction={
        pattern=north west lines
      }] plot [error bars/.cd, y dir=both, y explicit] table [x index=0, y index=1, y error plus index=2, y error minus index=2] {\recovone};
\addplot[draw=black,fill=Periwinkle!40,postaction={
        pattern=crosshatch
      }] plot [error bars/.cd, y dir=both, y explicit] table [x index=0, y index=1, y error plus index=2, y error minus index=2] {\norecovone};
\addplot[draw=black,fill=Red,postaction={
        pattern=grid
      }] plot [error bars/.cd, y dir=both, y explicit] table [x index=0, y index=1, y error plus index=2, y error minus index=2] {\controlone};
    \end{axis}
%
\end{tikzpicture}
};

\node[scale=1] (kth_cr) at (6.55,-0.3)
{
  \begin{tikzpicture}[declare function={sigma(\x)=1/(1+exp(-\x));
      sigmap(\x)=sigma(\x)*(1-sigma(\x));}]
\pgfplotstableread{
5   100.0   0.0
15   100.0   0.0
25   100.0   0.0
35   100.0   0.0
}\norecovone

\pgfplotstableread{
5   5.025   0.42
15   6.06   1.16
25   8.64   1.48
35   100.0   0.0
}\recovone

\pgfplotstableread{
5   4.41   0.52
15   5.42   0.93
25   6.57   1.01
35   100.0   0.0
}\recovadapt

\pgfplotstableread{
5   1.438   0.09
15  1.438   0.09
25  1.438   0.09
35  1.438   0.09
}\controlone

  \begin{axis}[
        xmin=0,
        xmax=40,
        ymin=0,
        ymax=120,
        ymode=log,
        log basis y={10},
        ybar=-0.4cm,
        width =7cm,
        height = 2.75cm,
        axis x line=bottom,
        axis y line=left,
        axis line style={-|},
        nodes near coords align={vertical}
        xtick={},
        xticklabels={},
       bar width=0.17cm,
    legend style={at={(0.105, 1)}, anchor=north, legend columns=1, draw=none},
    every axis legend/.append style={nodes={right}, inner sep = 0.2cm},
   x tick label style={align=center, yshift=-0.1cm},
    enlarge x limits=0.02,
    ]
\addplot[draw=black,fill=RedViolet!30,postaction={
        pattern=north east lines
      }] plot [error bars/.cd, y dir=both, y explicit] table [x index=0, y index=1, y error plus index=2, y error minus index=2] {\recovadapt};
\addplot[draw=black,fill=SeaGreen!60,postaction={
        pattern=north west lines
      }] plot [error bars/.cd, y dir=both, y explicit] table [x index=0, y index=1, y error plus index=2, y error minus index=2] {\recovone};
\addplot[draw=black,fill=Periwinkle!40,postaction={
        pattern=crosshatch
      }] plot [error bars/.cd, y dir=both, y explicit] table [x index=0, y index=1, y error plus index=2, y error minus index=2] {\norecovone};
\addplot[draw=black,fill=Red,postaction={
        pattern=grid
      }] plot [error bars/.cd, y dir=both, y explicit] table [x index=0, y index=1, y error plus index=2, y error minus index=2] {\controlone};
\end{axis}
%
\end{tikzpicture}
};

\node[scale=1] (kth_cr) at (12.97,-0.395)
{
  \begin{tikzpicture}[declare function={sigma(\x)=1/(1+exp(-\x));
      sigmap(\x)=sigma(\x)*(1-sigma(\x));}]
\pgfplotstableread{
5   0.0   0.0
15   0.0   0.0
25   0.0   0.0
35   0.0   0.0
}\norecovone

\pgfplotstableread{
5   0.2   0.0
15   0.065   0.001
25   0.04   0.001
35   0.0   0.0
}\recovone

\pgfplotstableread{
5   0.168   0.002
15   0.05   0.0001
25   0.03   0.001
35   0.0   0.0
}\recovadapt

\pgfplotstableread{
5   0.085   0.003
15  0.065   0.001
25  0.065   0.001
35  0.065   0.001
}\controlone

\begin{axis}[
        xmin=0,
        xmax=40,
        ymin=0,
        ymax=0.25,
        ybar=-0.4cm,
        width =7cm,
        height = 2.75cm,
        axis x line=bottom,
        axis y line=left,
        axis line style={-|},
        nodes near coords align={vertical}
        xtick={},
        xticklabels={},
       bar width=0.17cm,
    legend style={at={(0.105, 1)}, anchor=north, legend columns=1},
    every axis legend/.append style={nodes={right}, inner sep = 0.2cm},
   x tick label style={align=center, yshift=-0.1cm},
    enlarge x limits=0.02,
    ]
\addplot[draw=black,fill=RedViolet!30,postaction={
        pattern=north east lines
      }] plot [error bars/.cd, y dir=both, y explicit] table [x index=0, y index=1, y error plus index=2, y error minus index=2] {\recovadapt};
\addplot[draw=black,fill=SeaGreen!60,postaction={
        pattern=north west lines
      }] plot [error bars/.cd, y dir=both, y explicit] table [x index=0, y index=1, y error plus index=2, y error minus index=2] {\recovone};
\addplot[draw=black,fill=Periwinkle!40,postaction={
        pattern=crosshatch
      }] plot [error bars/.cd, y dir=both, y explicit] table [x index=0, y index=1, y error plus index=2, y error minus index=2] {\norecovone};
\addplot[draw=black,fill=Red,postaction={
        pattern=grid
      }] plot [error bars/.cd, y dir=both, y explicit] table [x index=0, y index=1, y error plus index=2, y error minus index=2] {\controlone};
\end{axis}
%
\end{tikzpicture}
};

\node[scale=1] (kth_cr) at (0,-1.865)
{
  \begin{tikzpicture}[declare function={sigma(\x)=1/(1+exp(-\x));
      sigmap(\x)=sigma(\x)*(1-sigma(\x));}]
    \pgfplotstableread{
5   0.164   0.06
15   0.164   0.06
25   0.164   0.06
35   0.164   0.06
}\norecovone

\pgfplotstableread{
5   0.99   0.004
15   0.98   0.012
25   0.95   0.015
35   0.08   0.029
}\recovone

\pgfplotstableread{
5   0.99   0.002
15   0.99   0.001
25   0.97   0.02
35   0.11   0.029
}\recovadapt

\pgfplotstableread{
5   0.99   0.003
15  0.99   0.003
25  0.99   0.003
35  0.99   0.003
}\controlone

\begin{axis}[
        xmin=0,
        xmax=40,
        ymin=0,
        ymax=1.2,
        ybar=-0.4cm,
        width =7cm,
        height = 2.75cm,
        axis x line=bottom,
        axis y line=left,
        axis line style={-|},
        nodes near coords align={vertical}
        xtick={},
        xticklabels={},
       bar width=0.17cm,
    legend style={at={(0.105, 1)}, anchor=north, legend columns=1},
    every axis legend/.append style={nodes={right}, inner sep = 0.2cm},
   x tick label style={align=center, yshift=-0.1cm},
    enlarge x limits=0.02,
    ]
\addplot[draw=black,fill=RedViolet!30,postaction={
        pattern=north east lines
      }] plot [error bars/.cd, y dir=both, y explicit] table [x index=0, y index=1, y error plus index=2, y error minus index=2] {\recovadapt};
\addplot[draw=black,fill=SeaGreen!60,postaction={
        pattern=north west lines
      }] plot [error bars/.cd, y dir=both, y explicit] table [x index=0, y index=1, y error plus index=2, y error minus index=2] {\recovone};
\addplot[draw=black,fill=Periwinkle!40,postaction={
        pattern=crosshatch
      }] plot [error bars/.cd, y dir=both, y explicit] table [x index=0, y index=1, y error plus index=2, y error minus index=2] {\norecovone};
\addplot[draw=black,fill=Red,postaction={
        pattern=grid
      }] plot [error bars/.cd, y dir=both, y explicit] table [x index=0, y index=1, y error plus index=2, y error minus index=2] {\controlone};
    \end{axis}
%
\end{tikzpicture}
};

\node[scale=1] (kth_cr) at (6.55,-1.7)
{
  \begin{tikzpicture}[declare function={sigma(\x)=1/(1+exp(-\x));
      sigmap(\x)=sigma(\x)*(1-sigma(\x));}]
\pgfplotstableread{
5   100.0   0.0
15   100.0   0.0
25   100.0   0.0
35   100.0   0.0
}\norecovone

\pgfplotstableread{
5   4.92   0.42
15   5.96   1.16
25   8.1375   1.48
35   100.0   0.0
}\recovone

\pgfplotstableread{
5   4.33   0.43
15   5.02   0.34
25   6.16   0.54
35   100.0   0.0
}\recovadapt

\pgfplotstableread{
5   1.47   0.07
15  1.47   0.07
25  1.47   0.07
35  1.47   0.07
}\controlone

\begin{axis}[
        xmin=0,
        xmax=40,
        ymin=0,
        ymax=120,
        ymode=log,
        log basis y={10},
        ybar=-0.4cm,
        width =7cm,
        height = 2.75cm,
        axis x line=bottom,
        axis y line=left,
        axis line style={-|},
        nodes near coords align={vertical}
        xtick={},
        xticklabels={},
       bar width=0.17cm,
    legend style={at={(0.105, 1)}, anchor=north, legend columns=1},
    every axis legend/.append style={nodes={right}, inner sep = 0.2cm},
   x tick label style={align=center, yshift=-0.1cm},
    enlarge x limits=0.02,
    ]
\addplot[draw=black,fill=RedViolet!30,postaction={
        pattern=north east lines
      }] plot [error bars/.cd, y dir=both, y explicit] table [x index=0, y index=1, y error plus index=2, y error minus index=2] {\recovadapt};
\addplot[draw=black,fill=SeaGreen!60,postaction={
        pattern=north west lines
      }] plot [error bars/.cd, y dir=both, y explicit] table [x index=0, y index=1, y error plus index=2, y error minus index=2] {\recovone};
\addplot[draw=black,fill=Periwinkle!40,postaction={
        pattern=crosshatch
      }] plot [error bars/.cd, y dir=both, y explicit] table [x index=0, y index=1, y error plus index=2, y error minus index=2] {\norecovone};
\addplot[draw=black,fill=Red,postaction={
        pattern=grid
      }] plot [error bars/.cd, y dir=both, y explicit] table [x index=0, y index=1, y error plus index=2, y error minus index=2] {\controlone};
    \end{axis}
%
\end{tikzpicture}
};

\node[scale=1] (kth_cr) at (12.97,-1.79)
{
  \begin{tikzpicture}[declare function={sigma(\x)=1/(1+exp(-\x));
      sigmap(\x)=sigma(\x)*(1-sigma(\x));}]
\pgfplotstableread{
5   0.0   0.0
15   0.0   0.0
25   0.0   0.0
35   0.0   0.0
}\norecovone

\pgfplotstableread{
5   0.2   0.0
15   0.065   0.001
25   0.04   0.001
35   0.0   0.0
}\recovone

\pgfplotstableread{
5   0.17   0.001
15   0.06   0.002
25   0.033   0.002
35   0.0   0.0
}\recovadapt

\pgfplotstableread{
5   0.072   0.005
15  0.069   0.006
25  0.069   0.006
35  0.069   0.006
}\controlone

\begin{axis}[
        xmin=0,
        xmax=40,
        ymin=0,
        ymax=0.25,
        ybar=-0.4cm,
        width =7cm,
        height = 2.75cm,
        axis x line=bottom,
        axis y line=left,
        axis line style={-|},
        nodes near coords align={vertical}
        xtick={},
        xticklabels={},
       bar width=0.17cm,
    legend style={at={(0.105, 1)}, anchor=north, legend columns=1},
    every axis legend/.append style={nodes={right}, inner sep = 0.2cm},
   x tick label style={align=center, yshift=-0.1cm},
    enlarge x limits=0.02,
    ]
\addplot[draw=black,fill=RedViolet!30,postaction={
        pattern=north east lines
      }] plot [error bars/.cd, y dir=both, y explicit] table [x index=0, y index=1, y error plus index=2, y error minus index=2] {\recovadapt};
\addplot[draw=black,fill=SeaGreen!60,postaction={
        pattern=north west lines
      }] plot [error bars/.cd, y dir=both, y explicit] table [x index=0, y index=1, y error plus index=2, y error minus index=2] {\recovone};
\addplot[draw=black,fill=Periwinkle!40,postaction={
        pattern=crosshatch
      }] plot [error bars/.cd, y dir=both, y explicit] table [x index=0, y index=1, y error plus index=2, y error minus index=2] {\norecovone};
\addplot[draw=black,fill=Red,postaction={
        pattern=grid
      }] plot [error bars/.cd, y dir=both, y explicit] table [x index=0, y index=1, y error plus index=2, y error minus index=2] {\controlone};
    \end{axis}
%
\end{tikzpicture}
};

\node[scale=1] (kth_cr) at (0,-3.3)
{
  \begin{tikzpicture}[declare function={sigma(\x)=1/(1+exp(-\x));
      sigmap(\x)=sigma(\x)*(1-sigma(\x));}]
\pgfplotstableread{
5   0.17   0.04
15   0.17   0.04
25   0.17   0.04
35   0.17   0.04
}\norecovone

\pgfplotstableread{
5   1.0   0.0
15   0.99   0.001
25   0.98   0.007
35   0.17   0.04
}\recovone

\pgfplotstableread{
5   1.0   0.0
15   1.0   0.00
25   0.99   0.001
35   0.18   0.02
}\recovadapt

\pgfplotstableread{
5   1.0   0.0
15  1.0   0.0
25  1.0   0.0
35  1.0   0.0
}\controlone

\begin{axis}[
        xmin=0,
        xmax=40,
        ymin=0,
        ymax=1.2,
        ybar=-0.4cm,
        width =7cm,
        height = 2.75cm,
        axis x line=bottom,
        axis y line=left,
        axis line style={-|},
        nodes near coords align={vertical}
        xtick={},
        xticklabels={},
       bar width=0.17cm,
    legend style={at={(0.105, 1)}, anchor=north, legend columns=1},
    every axis legend/.append style={nodes={right}, inner sep = 0.2cm},
   x tick label style={align=center, yshift=-0.1cm},
    enlarge x limits=0.02,
    ]
\addplot[draw=black,fill=RedViolet!30,postaction={
        pattern=north east lines
      }] plot [error bars/.cd, y dir=both, y explicit] table [x index=0, y index=1, y error plus index=2, y error minus index=2] {\recovadapt};
\addplot[draw=black,fill=SeaGreen!60,postaction={
        pattern=north west lines
      }] plot [error bars/.cd, y dir=both, y explicit] table [x index=0, y index=1, y error plus index=2, y error minus index=2] {\recovone};
\addplot[draw=black,fill=Periwinkle!40,postaction={
        pattern=crosshatch
      }] plot [error bars/.cd, y dir=both, y explicit] table [x index=0, y index=1, y error plus index=2, y error minus index=2] {\norecovone};
\addplot[draw=black,fill=Red,postaction={
        pattern=grid
      }] plot [error bars/.cd, y dir=both, y explicit] table [x index=0, y index=1, y error plus index=2, y error minus index=2] {\controlone};
    \end{axis}
\node[inner sep=0pt,align=center, scale=0.8, rotate=0, opacity=1] (obs) at (0.78,-0.2)
{
  $5$
};
\node[inner sep=0pt,align=center, scale=0.8, rotate=0, opacity=1] (obs) at (2.12,-0.2)
{
  $15$
};

\node[inner sep=0pt,align=center, scale=0.8, rotate=0, opacity=1] (obs) at (3.4,-0.2)
{
  $25$
};
\node[inner sep=0pt,align=center, scale=0.8, rotate=0, opacity=1] (obs) at (4.73,-0.2)
{
  $\infty$
};
\end{tikzpicture}
};

\node[scale=1] (kth_cr) at (7.5,-3.685)
{
  \begin{tikzpicture}[declare function={sigma(\x)=1/(1+exp(-\x));
      sigmap(\x)=sigma(\x)*(1-sigma(\x));}]
\pgfplotstableread{
5   100.0   0.0
15   100.0   0.0
25   100.0   0.0
35   100.0   0.0
}\norecovone

\pgfplotstableread{
5   4.92   0.24
15   5.37   0.34
25   7.74   0.51
35   100.0   0.0
}\recovone

\pgfplotstableread{
5   4.25   0.18
15   4.44   0.25
25   6.01   0.39
35   100.0   0.0
}\recovadapt

\pgfplotstableread{
5   1.44   0.05
15  1.44   0.05
25  1.44  0.05
35  1.44   0.05
}\controlone

\begin{axis}[
        xmin=0,
        xmax=40,
        ymin=0,
        ymax=120,
        ymode=log,
        log basis y={10},
        ybar=-0.4cm,
        width =7cm,
        height = 2.75cm,
        axis x line=bottom,
        axis y line=left,
        axis line style={-|},
        nodes near coords align={vertical}
        xtick={},
        reverse legend,
        xticklabels={},
       bar width=0.17cm,
       legend style={at={(0.6, -0.55)}, anchor=north, legend columns=4, nodes={scale=0.8, transform shape},
         draw=none
       },
    every axis legend/.append style={nodes={right}, inner sep = 0.2cm},
   x tick label style={align=center, yshift=-0.1cm},
    enlarge x limits=0.02,
    ]
\addplot[draw=black,fill=RedViolet!30,postaction={
        pattern=north east lines
      }] plot [error bars/.cd, y dir=both, y explicit] table [x index=0, y index=1, y error plus index=2, y error minus index=2] {\recovadapt};
\addplot[draw=black,fill=SeaGreen!60,postaction={
        pattern=north west lines
      }] plot [error bars/.cd, y dir=both, y explicit] table [x index=0, y index=1, y error plus index=2, y error minus index=2] {\recovone};
\addplot[draw=black,fill=Periwinkle!40,postaction={
        pattern=crosshatch
      }] plot [error bars/.cd, y dir=both, y explicit] table [x index=0, y index=1, y error plus index=2, y error minus index=2] {\norecovone};
\addplot[draw=black,fill=Red,postaction={
        pattern=grid
      }] plot [error bars/.cd, y dir=both, y explicit] table [x index=0, y index=1, y error plus index=2, y error minus index=2] {\controlone};
\legend{\textsc{periodic-adaptive}, \textsc{periodic}$\quad$, \textsc{no-recovery}$\quad$, \textsc{tolerance}$\quad$}
    \end{axis}
\node[inner sep=0pt,align=center, scale=0.8, rotate=0, opacity=1] (obs) at (0.78,-0.2)
{
  $5$
};
\node[inner sep=0pt,align=center, scale=0.8, rotate=0, opacity=1] (obs) at (2.12,-0.2)
{
  $15$
};

\node[inner sep=0pt,align=center, scale=0.8, rotate=0, opacity=1] (obs) at (3.4,-0.2)
{
  $25$
};
\node[inner sep=0pt,align=center, scale=0.8, rotate=0, opacity=1] (obs) at (4.73,-0.2)
{
  $\infty$
};
\end{tikzpicture}
};
\node[scale=1] (kth_cr) at (12.97,-3.35)
{
  \begin{tikzpicture}[declare function={sigma(\x)=1/(1+exp(-\x));
      sigmap(\x)=sigma(\x)*(1-sigma(\x));}]
\pgfplotstableread{
5   0.0   0.0
15   0.0   0.0
25   0.0   0.0
35   0.0   0.0
}\norecovone

\pgfplotstableread{
5   0.2   0.0
15   0.065   0.001
25   0.04   0.001
35   0.0   0.0
}\recovone

\pgfplotstableread{
5   0.17   0.001
15   0.06   0.002
25   0.035   0.001
35   0.0   0.0
}\recovadapt

\pgfplotstableread{
5   0.071   0.002
15  0.07   0.003
25  0.07  0.003
35  0.07   0.003
}\controlone

\begin{axis}[
        xmin=0,
        xmax=40,
        ymin=0,
        ymax=0.25,
        ybar=-0.4cm,
        width =7cm,
        height = 2.75cm,
        axis x line=bottom,
        axis y line=left,
        axis line style={-|},
        nodes near coords align={vertical}
        xtick={},
        xticklabels={},
       bar width=0.17cm,
    legend style={at={(0.105, 1)}, anchor=north, legend columns=1},
    every axis legend/.append style={nodes={right}, inner sep = 0.2cm},
   x tick label style={align=center, yshift=-0.1cm},
    enlarge x limits=0.02,
    ]
\addplot[draw=black,fill=RedViolet!30,postaction={
        pattern=north east lines
      }] plot [error bars/.cd, y dir=both, y explicit] table [x index=0, y index=1, y error plus index=2, y error minus index=2] {\recovadapt};
\addplot[draw=black,fill=SeaGreen!60,postaction={
        pattern=north west lines
      }] plot [error bars/.cd, y dir=both, y explicit] table [x index=0, y index=1, y error plus index=2, y error minus index=2] {\recovone};
\addplot[draw=black,fill=Periwinkle!40,postaction={
        pattern=crosshatch
      }] plot [error bars/.cd, y dir=both, y explicit] table [x index=0, y index=1, y error plus index=2, y error minus index=2] {\norecovone};
\addplot[draw=black,fill=Red,postaction={
        pattern=grid
      }] plot [error bars/.cd, y dir=both, y explicit] table [x index=0, y index=1, y error plus index=2, y error minus index=2] {\controlone};
    \end{axis}
\node[inner sep=0pt,align=center, scale=0.8, rotate=0, opacity=1] (obs) at (0.78,-0.2)
{
  $5$
};
\node[inner sep=0pt,align=center, scale=0.8, rotate=0, opacity=1] (obs) at (2.12,-0.2)
{
  $15$
};

\node[inner sep=0pt,align=center, scale=0.8, rotate=0, opacity=1] (obs) at (3.4,-0.2)
{
  $25$
};
\node[inner sep=0pt,align=center, scale=0.8, rotate=0, opacity=1] (obs) at (4.73,-0.2)
{
  $\infty$
};
\end{tikzpicture}
};

\node[inner sep=0pt,align=center, scale=0.9, rotate=0, opacity=1] (obs) at (0.45,-4.22)
{
  Maximum time-to-recovery $\Delta_{\mathrm{R}}$
};

\node[inner sep=0pt,align=center, scale=0.9, rotate=0, opacity=1] (obs) at (7.1,-4.22)
{
  Maximum time-to-recovery $\Delta_{\mathrm{R}}$
};

\node[inner sep=0pt,align=center, scale=0.9, rotate=0, opacity=1] (obs) at (13.7,-4.22)
{
  Maximum time-to-recovery $\Delta_{\mathrm{R}}$
};

\node[inner sep=0pt,align=center, scale=0.9, rotate=0, opacity=1] (obs) at (0.45,0.6)
{
  Average availability $T^{(\mathrm{A})}$
};

\node[inner sep=0pt,align=center, scale=0.9, rotate=0, opacity=1] (obs) at (7.1,0.6)
{
  Average time-to-recovery $T^{(\mathrm{R})}$
};

\node[inner sep=0pt,align=center, scale=0.9, rotate=0, opacity=1] (obs) at (13.7,0.6)
{
  Average recovery frequency $F^{(\mathrm{R})}$
};


\node[inner sep=0pt,align=center, scale=0.9, rotate=0, opacity=1, rotate=90] (obs) at (-3.25,-0.165)
{
  $N_1=3$
};
\node[inner sep=0pt,align=center, scale=0.9, rotate=0, opacity=1, rotate=90] (obs) at (-3.25,-1.6)
{
  $N_1=6$
};

\node[inner sep=0pt,align=center, scale=0.9, rotate=0, opacity=1, rotate=90] (obs) at (-3.25,-3.15)
{
  $N_1=9$
};
  \end{tikzpicture}

%% file: tikz/strategy_structure.tex
\begin{tikzpicture}[
    dot/.style={
        draw=black,
        fill=blue!90,
        circle,
        minimum size=3pt,
        inner sep=0pt,
        solid,
    },
    ]
\pgfplotsset{every tick label/.append style={font=\tiny}}
\node[scale=1] (kth_cr) at (0,2.15)
{
  \begin{tikzpicture}[declare function={sigma(\x)=1/(1+exp(-\x));
      sigmap(\x)=sigma(\x)*(1-sigma(\x));}]
\pgfplotstableread{
3   1.0
4   0.065
5   0.0
6   0.0
7   0.0
8   0.0
9   0.0
10  0.0
11  0.0
12  0.0
13  0.0
}\datatablee

\pgfplotstableread{
3   60    3.5
4   29   2
5   16   2.5
6   12   1
7   9.5   0.9
8   6.5   0.7
9   5   0.6
10   3.25   0.5
}\datatableee

\begin{axis}[
        xmin=2.7,
        xmax=13.8,
        ymin=0,
        ymax=1.3,
        width =8cm,
        height = 3cm,
        axis lines=center,
        xtick={3,4,5,6,7,8,9,10,11,12,13},
        ytick={0,0.5,1},
        scaled y ticks=false,
        yticklabel style={
        /pgf/number format/fixed,
        /pgf/number format/precision=5
        },
        xlabel style={below right},
        ylabel style={above left},
        axis line style={-{Latex[length=2mm]}},
        legend style={at={(0.96,0.7)}},
        legend columns=3,
        legend style={
            /tikz/column 2/.style={
                column sep=5pt,
              }
              }
              ]
\addplot[Black,mark=diamond, mark repeat=1, name path=l1, thick, domain=1:100] table [x index=0, y index=1] {\datatablee};
\end{axis}
\end{tikzpicture}
};

\node[inner sep=0pt,align=center, scale=0.8, rotate=0, opacity=1] (obs) at (3.7,1.58)
{
 $s$
};
\node[inner sep=0pt,align=center, scale=0.8, rotate=0, opacity=1] (obs) at (3.15,1.1)
{
 $s_{\mathrm{max}}$
};
\node[inner sep=0pt,align=center, scale=0.8, opacity=1] (obs) at (-2.3,3.2)
{
  $\pi(a=1 \mid s)$
};
\node[inner sep=0pt,align=center, scale=0.8] (time) at (0.6,1.05)
{
(a) Replication strategy.
};

\node[inner sep=0pt,align=center, scale=0.8] (time) at (0.6,-1.3)
{
(b) Recovery strategy.
};

\node[scale=1] (system) at (0.1,-0.25)
{
  \begin{tikzpicture}
\draw[->, color=black] (0,0) to (6,0);
\draw[->, color=black] (0,0) to (0,1.1);

\draw[-, color=black] (0,0.1) to (4.5,0.1) to (4.5,1) to (5.8,1);
\node[inner sep=0pt,align=center, scale=0.6] (time) at (-0.15,0.11)
{
$0$
};

\node[inner sep=0pt,align=center, scale=0.6] (time) at (-0.15,0.9)
{
$1$
};

\node[inner sep=0pt,align=center, scale=0.6] (time) at (5.85,-0.3)
{
$1$
};

\node[inner sep=0pt,align=center, scale=0.6] (time) at (0.05,-0.3)
{
$0$
};

\draw[-, color=black] (5.8,0.1) to (5.8,-0.1);

\draw[-, color=black] (0.1,0.9) to (-0.1,0.9);
\draw[-, color=black] (0.1,0.1) to (-0.1,0.1);

\draw[-, color=black] (4.5,-0.1) to (4.5,0.1);

\node[inner sep=0pt,align=center, scale=0.8] (time) at (0.2,1.28)
{
$\pi_{i,t}^{\star}(b)$
};

\node[inner sep=0pt,align=center, scale=0.8] (time) at (4.65,-0.2)
{
  $\alpha_i^{\star}=0.76$
};

\node[inner sep=0pt,align=center, scale=0.8] (time) at (6.2,0)
{
$b$
};

\node[inner sep=0pt,align=center, scale=0.8] (time) at (-0.8,0.09)
{
($\mathfrak{W}$)ait
};

\node[inner sep=0pt,align=center, scale=0.8] (time) at (-0.85,0.9)
{
($\mathfrak{R}$)ecover
};

%
%


    \end{tikzpicture}
  };

  \end{tikzpicture}

%% file: tikz/kl_1.tex
\begin{tikzpicture}[
    dot/.style={
        draw=black,
        fill=blue!90,
        circle,
        minimum size=3pt,
        inner sep=0pt,
        solid,
    },
    ]
\node[scale=1] (kth_cr) at (0,0)
{
\begin{tikzpicture}
\node[scale=1] (kth_cr) at (0,2.15)
{
  \begin{tikzpicture}[declare function={sigma(\x)=1/(1+exp(-\x));
      sigmap(\x)=sigma(\x)*(1-sigma(\x));}]

\pgfplotstableread{
0.0 0.325 0.07
0.117 0.285 0.06
0.279 0.23 0.03
0.467 0.21 0.05
0.679 0.195 0.04
0.789 0.185 0.02
0.899 0.177 0.03
1.121 0.17 0.04
1.34 0.165 0.03
1.55 0.16 0.02
1.766 0.154 0.01
1.972 0.149 0.02
2.173 0.14 0.01
2.369 0.133 0.01
2.561 0.125 0.02
}\intrusionkls

\begin{axis}[
        xmin=0,
        xmax=2.7,
        ymin=0,
        ymax=0.5,
        width =1.1\columnwidth,
        height = 0.43\columnwidth,
        axis lines=center,
        xmajorgrids=true,
        ymajorgrids=true,
        major grid style = {lightgray},
        minor grid style = {lightgray!25},
        scaled y ticks=false,
        yticklabel style={
        /pgf/number format/fixed,
        /pgf/number format/precision=5
      },
        axis line style={-{Latex[length=2mm]}},
        smooth,
        legend style={at={(0.96,0.7)}},
        legend columns=3,
        legend style={
            /tikz/column 2/.style={
                column sep=5pt,
              }
              }
              ]
\addplot[Black,mark=diamond, mark repeat=1, name path=l1, thick, domain=1:100] plot[error bars/.cd, y dir=both, y explicit] table [x index=0, y index=1, y error plus index=2, y error minus index=2] {\intrusionkls};
\end{axis}
\end{tikzpicture}
};

\node[inner sep=0pt,align=center, scale=0.9, rotate=0, opacity=1] (obs) at (0.7,0.6)
{
$D_{\text{\textsc{kl}}}(Z_i(\cdot \mid \mathbb{H}) \parallel Z_i(\cdot \mid \mathbb{C}))$
};
\node[inner sep=0pt,align=center, scale=0.9, rotate=90, opacity=1] (obs) at (-4.7,2.55)
{
  $J^{\star}_i$ (\ref{eq:objective_recovery})
};
\end{tikzpicture}
};
\node[scale=1] (kth_cr) at (0,-3)
{
  \begin{tikzpicture}
\node[scale=1] (kth_cr) at (0,2.15)
{
  \begin{tikzpicture}[declare function={sigma(\x)=1/(1+exp(-\x));
      sigmap(\x)=sigma(\x)*(1-sigma(\x));}]

\pgfplotstableread{
0.0 0.13 0.025
0.038 0.151 0.02
0.097 0.153 0.045
0.147 0.155 0.025
0.195 0.157 0.02
0.25 0.16 0.025
0.306 0.162 0.02
0.363 0.16 0.035
0.407 0.17 0.02
0.62 0.19 0.035
0.93 0.23 0.045
1.264 0.28 0.07
}\intrusionkls

\begin{axis}[
        xmin=0,
        xmax=1.3,
        ymin=0,
        ymax=0.5,
        width =1.1\columnwidth,
        height = 0.43\columnwidth,
        axis lines=center,
        xmajorgrids=true,
        ymajorgrids=true,
        major grid style = {lightgray},
        minor grid style = {lightgray!25},
        scaled y ticks=false,
        yticklabel style={
        /pgf/number format/fixed,
        /pgf/number format/precision=5
      },
        axis line style={-{Latex[length=2mm]}},
        smooth,
        legend style={at={(0.96,0.7)}},
        legend columns=3,
        legend style={
            /tikz/column 2/.style={
                column sep=5pt,
              }
              }
              ]
\addplot[Black,mark=diamond, mark repeat=1, name path=l1, thick, domain=1:100] plot[error bars/.cd, y dir=both, y explicit] table [x index=0, y index=1, y error plus index=2, y error minus index=2] {\intrusionkls};
\end{axis}
\end{tikzpicture}
};

\node[inner sep=0pt,align=center, scale=0.9, rotate=0, opacity=1] (obs) at (0.7,0.6)
{
$D_{\text{\textsc{kl}}}(Z_i(\cdot \mid \mathbb{C}) \parallel \widehat{Z}_i(\cdot \mid \mathbb{C}))$
};
\node[inner sep=0pt,align=center, scale=0.9, rotate=90, opacity=1] (obs) at (-4.7,2.55)
{
  $J^{\star}_i$ (\ref{eq:objective_recovery})
};
\end{tikzpicture}
};
\end{tikzpicture}

%% file: tikz/threshold_strategy.tex
      \begin{tikzpicture}[fill=white, >=stealth,
    node distance=3cm,
    database/.style={
      cylinder,
      cylinder uses custom fill,
      shape border rotate=90,
      aspect=0.25,
      draw}]

    \tikzset{
node distance = 9em and 4em,
sloped,
   box/.style = {%
    shape=rectangle,
    rounded corners,
    draw=blue!40,
    fill=blue!15,
    align=center,
    font=\fontsize{12}{12}\selectfont},
 arrow/.style = {%
    line width=0.1mm,
    -{Triangle[length=5mm,width=2mm]},
    shorten >=1mm, shorten <=1mm,
    font=\fontsize{8}{8}\selectfont},
}

\node[scale=1] (system) at (0,-1)
{
  \begin{tikzpicture}
\draw[->, color=black] (0,0) to (6,0);
\draw[->, color=black] (0,0) to (0,1.1);

\draw[-, color=black] (0,0.1) to (3,0.1) to (3,1) to (5.8,1);
\node[inner sep=0pt,align=center, scale=0.8] (time) at (-0.15,0.11)
{
$0$
};

\node[inner sep=0pt,align=center, scale=0.8] (time) at (-0.15,0.9)
{
$1$
};

\node[inner sep=0pt,align=center, scale=0.8] (time) at (5.9,-0.3)
{
$1$
};

\node[inner sep=0pt,align=center, scale=0.8] (time) at (0.05,-0.3)
{
$0$
};

\draw[-, color=black] (5.8,0.1) to (5.8,-0.1);

\draw[-, color=black] (0.1,0.9) to (-0.1,0.9);
\draw[-, color=black] (0.1,0.1) to (-0.1,0.1);

\draw[-, color=black, dashed] (3,-0.25) to (3,1.0);

\node[inner sep=0pt,align=center, scale=0.8] (time) at (0.2,1.25)
{
$a_{i,t} = \pi_{i,t}^{\star}(b)$
};

\node[inner sep=0pt,align=center, scale=0.8] (time) at (3,-0.5)
{
  $\alpha^{\star}_{t}$
};

\node[inner sep=0pt,align=center, scale=0.8] (time) at (6.2,0)
{
$b$
};

\node[inner sep=0pt,align=center, scale=0.8] (time) at (-0.8,0.09)
{
($\mathfrak{W}$)ait
};

\node[inner sep=0pt,align=center, scale=0.8] (time) at (-0.85,0.9)
{
($\mathfrak{R}$)ecover
};

\draw [decorate,decoration={brace,amplitude=4pt,mirror,raise=4.5pt},yshift=0pt,line width=0.15mm]
(0,0.11) -- (2.95,0.11) node [black,midway,xshift=0.2cm] {};

\node[inner sep=0pt,align=center, scale=0.8] (time) at (1.52,-0.35)
{
$\mathcal{W}$
};

\node[inner sep=0pt,align=center, scale=0.8] (time) at (4.45,-0.35)
{
$\mathcal{R}$
};

\draw [decorate,decoration={brace,amplitude=4pt,mirror,raise=4.5pt},yshift=0pt,line width=0.15mm]
(3.05,0.11) -- (5.75,0.11) node [black,midway,xshift=0.2cm] {};

    \end{tikzpicture}
  };

\end{tikzpicture}

%% file: tikz/thresholds.tex
\begin{tikzpicture}[
    dot/.style={
        draw=black,
        fill=blue!90,
        circle,
        minimum size=3pt,
        inner sep=0pt,
        solid,
    },
    ]
\node[scale=1] (kth_cr) at (0,2.15)
{
  \begin{tikzpicture}[declare function={sigma(\x)=1/(1+exp(-\x));
      sigmap(\x)=sigma(\x)*(1-sigma(\x));}]
\pgfplotstableread{
100 0.9388150125316148
95 0.876998583590625
90 0.8728223056679183
85 0.8719118131284906
80 0.8690461306980984
75 0.865317731849339
70 0.8618157894077201
65 0.8608482535785417
60 0.8589515038171087
55 0.8543072944641904
50 0.8543072944641904
45 0.8543072944641904
40 0.8543072944641904
35 0.8543072944641904
30 0.8543072944641904
25 0.8543072944641904
20 0.8543072944641904
15 0.8543072944641904
10 0.8543072944641904
5 0.8543072944641904
0 0.8543072944641904
}\datatablee
\begin{axis}[
        xmin=1,
        xmax=108,
        ymin=0.83,
        ymax=0.95,
        width =1.15\columnwidth,
        height = 0.35\columnwidth,
        axis lines=center,
        xmajorgrids=true,
        ymajorgrids=true,
        major grid style = {lightgray},
        minor grid style = {lightgray!25},
        scaled y ticks=false,
        yticklabel style={
        /pgf/number format/fixed,
        /pgf/number format/precision=5
        },
        xlabel style={below right},
        ylabel style={above left},
        axis line style={-{Latex[length=2mm]}},
        legend style={at={(0.65,0.8)}},
        legend columns=3,
        legend style={
            /tikz/column 2/.style={
                column sep=5pt,
              }
              }
              ]
\addplot[Blue,mark=diamond, name path=l1, thick, domain=1:200] table [x index=0, y index=1] {\datatablee};

\legend{$\alpha^{\star}_{t}$}
\end{axis}
\end{tikzpicture}
};
\node[inner sep=0pt,align=center, scale=1, rotate=0, opacity=1] (obs) at (5.07,1.5)
{
  $t$
};
  \end{tikzpicture}

%% file: tikz/pmf_1.tex
\begin{tikzpicture}[
    dot/.style={
        draw=black,
        fill=blue!90,
        circle,
        minimum size=3pt,
        inner sep=0pt,
        solid,
    },
    ]

\node[scale=1] (kth_cr) at (0,2.15)
{
  \begin{tikzpicture}[declare function={sigma(\x)=1/(1+exp(-\x));
      sigmap(\x)=sigma(\x)*(1-sigma(\x));}]
\pgfplotstableread{
0 0.333451
1 0.222301
2 0.148200
3 0.098800
4 0.065867
5 0.043911
6 0.029274
7 0.019516
8 0.013010
9 0.008673
10 0.005781
11 0.003853
12 0.002567
13 0.001710
14 0.001137
15 0.000754
16 0.000497
17 0.000323
18 0.000203
19 0.000118
20 0.000054
}\datatablee
\pgfplotstableread{
0 0.005066
1 0.007290
2 0.010491
3 0.015098
4 0.021726
5 0.031265
6 0.044991
7 0.064743
8 0.093167
9 0.134069
10 0.192929
11 0.128619
12 0.085746
13 0.057164
14 0.038109
15 0.025406
16 0.016937
17 0.011291
18 0.007527
19 0.005018
20 0.003345
}\datatableee
\pgfplotstableread{
0 0.000113
1 0.000238
2 0.000393
3 0.000599
4 0.000884
5 0.001287
6 0.001861
7 0.002685
8 0.003868
9 0.005569
10 0.008016
11 0.011537
12 0.016603
13 0.023893
14 0.034383
15 0.049478
16 0.071200
17 0.102459
18 0.147441
19 0.212172
20 0.305320
}\datatableeee
\begin{axis}[
        xmin=0,
        xmax=21,
        ymin=0,
        ymax=0.45,
        width =1.2\columnwidth,
        height = 0.35\columnwidth,
        axis lines=center,
        xmajorgrids=true,
        ymajorgrids=true,
        major grid style = {lightgray},
        minor grid style = {lightgray!25},
        ytick={0.0, 0.1,0.2, 0.3, 0.4, 0.5, 0.6, 0.7, 0.8, 0.9, 1.0},
        yticklabels={$0$, $$, $$, $$, $$, $$, $$, $$, $$, $$,$$},
        xtick={0,2,4,6,8,10,12,14,16,18,20},
        xticklabels={$0$, $2$, $4$, $6$, $8$, $10$, $12$, $14$, $16$, $18$, $20$},
        scaled y ticks=false,
        yticklabel style={
        /pgf/number format/fixed,
        /pgf/number format/precision=5
        },
        xlabel style={below right},
        ylabel style={above left},
        axis line style={-{Latex[length=2mm]}},
        legend style={at={(0.8,-0.28)}},
        legend columns=3,
        legend style={
          draw=none,
            /tikz/column 2/.style={
                column sep=5pt,
              }
              }
              ]
              \addplot+[ycomb,Blue,thick, mark=diamond] table [x index=0, y index=1] {\datatablee};
              \addplot+[ycomb,Red,thick, mark=triangle] table [x index=0, y index=1] {\datatableee};
              \addplot+[ycomb,Black,thick, mark=square] table [x index=0, y index=1] {\datatableeee};
\legend{$s=0$, $s=10$, $s=20$}
\end{axis}
\end{tikzpicture}
};

\node[inner sep=0pt,align=center, scale=1, rotate=0, opacity=1] (obs) at (-4.29,1.57)
{
  $0$
};

\node[inner sep=0pt,align=center, scale=1, rotate=0, opacity=1] (obs) at (4.9,1.85)
{
  $s^{\prime}$
};
\node[inner sep=0pt,align=center, scale=1, rotate=0, opacity=1] (obs) at (-2.35,3.65)
{
  $f_{\mathrm{S}}(S_{t+1}=s^{\prime} \mid S_t=s, A_t=0)$
};

\end{tikzpicture}

%% file: tikz/tolerance_28.tex
      \begin{tikzpicture}[fill=white, >=stealth,
    node distance=3cm,
    database/.style={
      cylinder,
      cylinder uses custom fill,
      shape border rotate=90,
      aspect=0.25,
      draw}]

    \tikzset{
node distance = 9em and 4em,
sloped,
   box/.style = {%
    shape=rectangle,
    rounded corners,
    draw=blue!40,
    fill=blue!15,
    align=center,
    font=\fontsize{12}{12}\selectfont},
 arrow/.style = {%
    line width=0.1mm,
    -{Triangle[length=5mm,width=2mm]},
    shorten >=1mm, shorten <=1mm,
    font=\fontsize{8}{8}\selectfont},
}

\node [scale=1] (node1) at (2,-3.1) {
\begin{tikzpicture}[fill=white, >=stealth,
    node distance=3cm,
    database/.style={
      cylinder,
      cylinder uses custom fill,
      shape border rotate=90,
      aspect=0.25,
      draw}]
    \tikzset{
node distance = 9em and 4em,
sloped,
   box/.style = {%
    shape=rectangle,
    rounded corners,
    draw=blue!40,
    fill=blue!15,
    align=center,
    font=\fontsize{12}{12}\selectfont},
 arrow/.style = {%
    line width=0.1mm,
    shorten >=1mm, shorten <=1mm,
    font=\fontsize{8}{8}\selectfont},
}
\node[inner sep=0pt,align=center, scale=0.5, color=black] (hacker) at (0.1,0.5) {
Client
};
\node[inner sep=0pt,align=center, scale=0.5, color=black] (hacker) at (0.1,0) {
Replica $1$
};
\node[inner sep=0pt,align=center, scale=0.5, color=black] (hacker) at (0.1,-0.2) {
(leader)
};
\node[inner sep=0pt,align=center, scale=0.5, color=black] (hacker) at (0.1,-0.5) {
Replica $2$
};
\node[inner sep=0pt,align=center, scale=0.5, color=black] (hacker) at (0.1,-1) {
Replica $3$
};

\node[inner sep=0pt,align=center, scale=0.42, color=black] (hacker) at (0.87,0.65) {
\textsc{request}
};

\node[inner sep=0pt,align=center, scale=0.42, color=black] (hacker) at (1.76,0.65) {
\textsc{prepare}
};
\node[inner sep=0pt,align=center, scale=0.42, color=black] (hacker) at (2.76,0.65) {
\textsc{commit}
};
\node[inner sep=0pt,align=center, scale=0.42, color=black] (hacker) at (3.76,0.65) {
\textsc{reply}
};

\draw[-, black, dashed, thick, line width=0.1mm] (1.25,0.8) to (1.25,-1);
\draw[-, black, dashed, thick, line width=0.1mm] (2.25,0.8) to (2.25,-1);
\draw[-, black, dashed, thick, line width=0.1mm] (3.25,0.8) to (3.25,-1);

\draw[->, black, thick, line width=0.1mm] (0.5,0.5) to (4.5,0.5);
\draw[->, black, thick, line width=0.1mm] (0.5,0) to (4.5,0);
\draw[->, black, thick, line width=0.1mm] (0.5,-0.5) to (4.5,-0.5);
\draw[->, black, thick, line width=0.1mm] (0.5,-1) to (4.5,-1);

\draw[->, black, thick, line width=0.1mm] (0.5,0.5) to (1.25,0);
\draw[->, black, thick, line width=0.1mm] (0.5,0.5) to (1.25,-0.5);
\draw[->, black, thick, line width=0.1mm] (0.5,0.5) to (1.25,-1);
\draw[->, black, thick, line width=0.1mm] (1.25,0) to (2.25,-0.5);
\draw[->, black, thick, line width=0.1mm] (1.25,0) to (2.25,-1);
\draw[->, black, thick, line width=0.1mm] (2.25,0) to (3.25,-0.5);
\draw[->, black, thick, line width=0.1mm] (2.25,0) to (3.25,-1);

\draw[->, black, thick, line width=0.1mm] (2.25,-0.5) to (3.25,0);
\draw[->, black, thick, line width=0.1mm] (2.25,-0.5) to (3.25,-1);
\draw[->, black, thick, line width=0.1mm] (2.25,-1) to (3.25,0);
\draw[->, black, thick, line width=0.1mm] (2.25,-1) to (3.25,-0.5);

\draw[->, black, thick, line width=0.1mm] (3.25,-0) to (3.75,0.5);
\draw[->, black, thick, line width=0.1mm] (3.25,-0.5) to (4,0.5);
\draw[->, black, thick, line width=0.1mm] (3.25,-1) to (4.25,0.5);
\end{tikzpicture}
};

\node [scale=1] (node1) at (2,-5) {
\begin{tikzpicture}[fill=white, >=stealth,
    node distance=3cm,
    database/.style={
      cylinder,
      cylinder uses custom fill,
      shape border rotate=90,
      aspect=0.25,
      draw}]
    \tikzset{
node distance = 9em and 4em,
sloped,
   box/.style = {%
    shape=rectangle,
    rounded corners,
    draw=blue!40,
    fill=blue!15,
    align=center,
    font=\fontsize{12}{12}\selectfont},
 arrow/.style = {%
    line width=0.1mm,
    shorten >=1mm, shorten <=1mm,
    font=\fontsize{8}{8}\selectfont},
}
\node[inner sep=0pt,align=center, scale=0.8, color=black] (hacker) at (0.85,0) {
\Crosss
};

\node[inner sep=0pt,align=center, scale=0.5, color=black] (hacker) at (0.1,0) {
Replica $1$
};
\node[inner sep=0pt,align=center, scale=0.5, color=black] (hacker) at (0.1,-0.2) {
(leader $v$)
};
\node[inner sep=0pt,align=center, scale=0.5, color=black] (hacker) at (0.1,-0.7) {
(leader $v+1$)
};
\node[inner sep=0pt,align=center, scale=0.5, color=black] (hacker) at (0.1,-0.5) {
Replica $2$
};
\node[inner sep=0pt,align=center, scale=0.5, color=black] (hacker) at (0.1,-1) {
Replica $3$
};

\node[inner sep=0pt,align=center, scale=0.42, color=black] (hacker) at (0.87,0.2) {
\textsc{crash}
};

\node[inner sep=0pt,align=center, scale=0.42, color=black] (hacker) at (1.76,0.2) {
  \textsc{request}\\
  \textsc{view-change}
};
\node[inner sep=0pt,align=center, scale=0.42, color=black] (hacker) at (2.76,0.2) {
\textsc{view-change}
};
\node[inner sep=0pt,align=center, scale=0.42, color=black] (hacker) at (3.76,0.2) {
\textsc{new-view}
};

\draw[-, black, dashed, thick, line width=0.1mm] (1.25,0.3) to (1.25,-1);
\draw[-, black, dashed, thick, line width=0.1mm] (2.25,0.3) to (2.25,-1);
\draw[-, black, dashed, thick, line width=0.1mm] (3.25,0.3) to (3.25,-1);

\draw[->, black, thick, line width=0.1mm] (0.5,0) to (4.5,0);
\draw[->, black, thick, line width=0.1mm] (0.5,-0.5) to (4.5,-0.5);
\draw[->, black, thick, line width=0.1mm] (0.5,-1) to (4.5,-1);
\draw[->, black, thick, line width=0.1mm] (1.25,-0.5) to (2.25,0);
\draw[->, black, thick, line width=0.1mm] (1.25,-0.5) to (2.25,-1);
\draw[->, black, thick, line width=0.1mm] (1.25,-1) to (2.25,0);
\draw[->, black, thick, line width=0.1mm] (1.25,-1) to (2.25,-0.5);

\draw[->, black, thick, line width=0.1mm] (2.25,-0.5) to (3.25,0);
\draw[->, black, thick, line width=0.1mm] (2.25,-0.5) to (3.25,-1);
\draw[->, black, thick, line width=0.1mm] (2.25,-1) to (3.25,0);
\draw[->, black, thick, line width=0.1mm] (2.25,-1) to (3.25,-0.5);

\draw[->, black, thick, line width=0.1mm] (3.25,-0.5) to (4.25,0);
\draw[->, black, thick, line width=0.1mm] (3.25,-0.5) to (4.25,-1);
\end{tikzpicture}
};

\node [scale=1] (node1) at (0,-6.8) {
\begin{tikzpicture}[fill=white, >=stealth,
    node distance=3cm,
    database/.style={
      cylinder,
      cylinder uses custom fill,
      shape border rotate=90,
      aspect=0.25,
      draw}]
    \tikzset{
node distance = 9em and 4em,
sloped,
   box/.style = {%
    shape=rectangle,
    rounded corners,
    draw=blue!40,
    fill=blue!15,
    align=center,
    font=\fontsize{12}{12}\selectfont},
 arrow/.style = {%
    line width=0.1mm,
    shorten >=1mm, shorten <=1mm,
    font=\fontsize{8}{8}\selectfont},
}
\node[inner sep=0pt,align=center, scale=0.5, color=black] (hacker) at (0.6,0) {
Replica $1$
};
\node[inner sep=0pt,align=center, scale=0.5, color=black] (hacker) at (0.6,-0.5) {
Replica $2$
};
\node[inner sep=0pt,align=center, scale=0.5, color=black] (hacker) at (0.6,-1) {
Replica $3$
};

\node[inner sep=0pt,align=center, scale=0.42, color=black] (hacker) at (1.7,0.2) {
\textsc{checkpoint}
};

\draw[->, black, thick, line width=0.1mm] (1,0) to (2.5,0);
\draw[->, black, thick, line width=0.1mm] (1,-0.5) to (2.5,-0.5);
\draw[->, black, thick, line width=0.1mm] (1,-1) to (2.5,-1);

\draw[->, black, thick, line width=0.1mm] (1.25,0) to (2.25,-0.5);
\draw[->, black, thick, line width=0.1mm] (1.25,0) to (2.25,-1);

\draw[->, black, thick, line width=0.1mm] (1.25,-0.5) to (2.25,0);
\draw[->, black, thick, line width=0.1mm] (1.25,-0.5) to (2.25,-1);

\draw[->, black, thick, line width=0.1mm] (1.25,-1) to (2.25,0);
\draw[->, black, thick, line width=0.1mm] (1.25,-1) to (2.25,-0.5);
\end{tikzpicture}
};

\node [scale=1] (node1) at (3.15,-7) {
\begin{tikzpicture}[fill=white, >=stealth,
    node distance=3cm,
    database/.style={
      cylinder,
      cylinder uses custom fill,
      shape border rotate=90,
      aspect=0.25,
      draw}]
    \tikzset{
node distance = 9em and 4em,
sloped,
   box/.style = {%
    shape=rectangle,
    rounded corners,
    draw=blue!40,
    fill=blue!15,
    align=center,
    font=\fontsize{12}{12}\selectfont},
 arrow/.style = {%
    line width=0.1mm,
    shorten >=1mm, shorten <=1mm,
    font=\fontsize{8}{8}\selectfont},
}
\node[inner sep=0pt,align=center, scale=0.5, color=black] (hacker) at (0,0.5) {
Controller
};
\node[inner sep=0pt,align=center, scale=0.5, color=black] (hacker) at (0.1,0) {
Replica $1$
};
\node[inner sep=0pt,align=center, scale=0.5, color=black] (hacker) at (0.1,-0.2) {
(compromised)
};
\node[inner sep=0pt,align=center, scale=0.5, color=black] (hacker) at (0.1,-0.5) {
Replica $2$
};
\node[inner sep=0pt,align=center, scale=0.5, color=black] (hacker) at (0.1,-1) {
Replica $3$
};

\node[inner sep=0pt,align=center, scale=0.42, color=black] (hacker) at (0.9,0.65) {
\textsc{recover}
};

\node[inner sep=0pt,align=center, scale=0.42, color=black] (hacker) at (1.76,0.65) {
  \textsc{request}\\
  \textsc{state}
};
\node[inner sep=0pt,align=center, scale=0.42, color=black] (hacker) at (2.76,0.65) {
\textsc{state}
};

\draw[-, black, dashed, thick, line width=0.1mm] (1.25,0.8) to (1.25,-1);
\draw[-, black, dashed, thick, line width=0.1mm] (2.25,0.8) to (2.25,-1);
\draw[->, black, thick, line width=0.1mm] (0.5,0.5) to (3.25,0.5);
\draw[->, black, thick, line width=0.1mm] (0.5,0) to (3.25,0);
\draw[->, black, thick, line width=0.1mm] (0.5,-0.5) to (3.25,-0.5);
\draw[->, black, thick, line width=0.1mm] (0.5,-1) to (3.25,-1);

\draw[->, black, thick, line width=0.1mm] (0.5,0.5) to (1.25,0);
\draw[->, black, thick, line width=0.1mm] (1.25,0) to (2.25,-0.5);
\draw[->, black, thick, line width=0.1mm] (1.25,0) to (2.25,-1);
\draw[->, black, thick, line width=0.1mm] (2.25,-0.5) to (2.75,0);
\draw[->, black, thick, line width=0.1mm] (2.25,-1) to (3,0);

\end{tikzpicture}
};

\node [scale=1] (node1) at (2,-9.2) {
\begin{tikzpicture}[fill=white, >=stealth,
    node distance=3cm,
    database/.style={
      cylinder,
      cylinder uses custom fill,
      shape border rotate=90,
      aspect=0.25,
      draw}]
    \tikzset{
node distance = 9em and 4em,
sloped,
   box/.style = {%
    shape=rectangle,
    rounded corners,
    draw=blue!40,
    fill=blue!15,
    align=center,
    font=\fontsize{12}{12}\selectfont},
 arrow/.style = {%
    line width=0.1mm,
    shorten >=1mm, shorten <=1mm,
    font=\fontsize{8}{8}\selectfont},
}
\node[inner sep=0pt,align=center, scale=0.5, color=black] (hacker) at (0,0.5) {
New replica
};
\node[inner sep=0pt,align=center, scale=0.5, color=black] (hacker) at (0.1,0) {
Replica $1$
};
\node[inner sep=0pt,align=center, scale=0.5, color=black] (hacker) at (0.1,-0.2) {
(leader $v$)
};
\node[inner sep=0pt,align=center, scale=0.5, color=black] (hacker) at (0.1,-0.7) {
(leader $v+1$)
};
\node[inner sep=0pt,align=center, scale=0.5, color=black] (hacker) at (0.1,-0.5) {
Replica $2$
};
\node[inner sep=0pt,align=center, scale=0.5, color=black] (hacker) at (0.1,-1) {
Replica $3$
};

\node[inner sep=0pt,align=center, scale=0.42, color=black] (hacker) at (0.75,0.65) {
\textsc{join-request}
};

\node[inner sep=0pt,align=center, scale=0.42, color=black] (hacker) at (1.76,0.65) {
\textsc{join}
};
\node[inner sep=0pt,align=center, scale=0.42, color=black] (hacker) at (2.76,0.65) {
\textsc{new-view}
};
\node[inner sep=0pt,align=center, scale=0.42, color=black] (hacker) at (3.76,0.65) {
\textsc{join-reply}
};

\draw[-, black, dashed, thick, line width=0.1mm] (1.25,0.8) to (1.25,-1);
\draw[-, black, dashed, thick, line width=0.1mm] (2.25,0.8) to (2.25,-1);
\draw[-, black, dashed, thick, line width=0.1mm] (3.25,0.8) to (3.25,-1);

\draw[->, black, thick, line width=0.1mm] (0.5,0.5) to (4.5,0.5);
\draw[->, black, thick, line width=0.1mm] (0.5,0) to (4.5,0);
\draw[->, black, thick, line width=0.1mm] (0.5,-0.5) to (4.5,-0.5);
\draw[->, black, thick, line width=0.1mm] (0.5,-1) to (4.5,-1);

\draw[->, black, thick, line width=0.1mm] (0.5,0.5) to (1.25,0);
\draw[->, black, thick, line width=0.1mm] (0.5,0.5) to (1.25,-0.5);
\draw[->, black, thick, line width=0.1mm] (0.5,0.5) to (1.25,-1);
\draw[->, black, thick, line width=0.1mm] (1.25,0) to (2.25,-0.5);
\draw[->, black, thick, line width=0.1mm] (1.25,0) to (2.25,-1);
\draw[->, black, thick, line width=0.1mm] (1.25,-0.5) to (2.25,0);
\draw[->, black, thick, line width=0.1mm] (1.25,-0.5) to (2.25,-1);
\draw[->, black, thick, line width=0.1mm] (1.25,-1) to (2.25,0);
\draw[->, black, thick, line width=0.1mm] (1.25,-1) to (2.25,-0.5);
\draw[->, black, thick, line width=0.1mm] (2.25,-0.5) to (3.25,0);
\draw[->, black, thick, line width=0.1mm] (2.25,-0.5) to (3.25,-1);
\draw[->, black, thick, line width=0.1mm] (2.25,-0.5) to (3.25,0.5);
\draw[->, black, thick, line width=0.1mm] (3.25,-0) to (3.75,0.5);
\draw[->, black, thick, line width=0.1mm] (3.25,-0.5) to (4,0.5);
\draw[->, black, thick, line width=0.1mm] (3.25,-1) to (4.25,0.5);
\end{tikzpicture}
};

\node [scale=1] (node1) at (2,-11.25) {
\begin{tikzpicture}[fill=white, >=stealth,
    node distance=3cm,
    database/.style={
      cylinder,
      cylinder uses custom fill,
      shape border rotate=90,
      aspect=0.25,
      draw}]
    \tikzset{
node distance = 9em and 4em,
sloped,
   box/.style = {%
    shape=rectangle,
    rounded corners,
    draw=blue!40,
    fill=blue!15,
    align=center,
    font=\fontsize{12}{12}\selectfont},
 arrow/.style = {%
    line width=0.1mm,
    shorten >=1mm, shorten <=1mm,
    font=\fontsize{8}{8}\selectfont},
}
\node[inner sep=0pt,align=center, scale=0.5, color=black] (hacker) at (0.1,0.43) {
  System\\
  controller
};
\node[inner sep=0pt,align=center, scale=0.5, color=black] (hacker) at (0.1,0) {
Replica $1$
};
\node[inner sep=0pt,align=center, scale=0.5, color=black] (hacker) at (0.1,-0.2) {
(leader $v$)
};
\node[inner sep=0pt,align=center, scale=0.5, color=black] (hacker) at (0.1,-0.7) {
(leader $v+1$)
};
\node[inner sep=0pt,align=center, scale=0.5, color=black] (hacker) at (0.1,-0.5) {
Replica $2$
};
\node[inner sep=0pt,align=center, scale=0.5, color=black] (hacker) at (0.1,-1) {
Replica $3$
};

\node[inner sep=0pt,align=center, scale=0.42, color=black] (hacker) at (0.65,0.65) {
\textsc{evict-request}
};

\node[inner sep=0pt,align=center, scale=0.42, color=black] (hacker) at (1.76,0.65) {
\textsc{evict}
};
\node[inner sep=0pt,align=center, scale=0.42, color=black] (hacker) at (2.76,0.65) {
\textsc{new-view}
};
\node[inner sep=0pt,align=center, scale=0.42, color=black] (hacker) at (3.76,0.65) {
\textsc{exit-reply}
};

\draw[-, black, dashed, thick, line width=0.1mm] (1.25,0.8) to (1.25,-1);
\draw[-, black, dashed, thick, line width=0.1mm] (2.25,0.8) to (2.25,-1);
\draw[-, black, dashed, thick, line width=0.1mm] (3.25,0.8) to (3.25,-1);

\draw[->, black, thick, line width=0.1mm] (0.5,0.5) to (4.5,0.5);
\draw[->, black, thick, line width=0.1mm] (0.5,0) to (4.5,0);
\draw[->, black, thick, line width=0.1mm] (0.5,-0.5) to (4.5,-0.5);
\draw[->, black, thick, line width=0.1mm] (0.5,-1) to (4.5,-1);

\draw[->, black, thick, line width=0.1mm] (0.5,0.5) to (1.25,0);
\draw[->, black, thick, line width=0.1mm] (0.5,0.5) to (1.25,-0.5);
\draw[->, black, thick, line width=0.1mm] (0.5,0.5) to (1.25,-1);
\draw[->, black, thick, line width=0.1mm] (1.25,0) to (2.25,-0.5);
\draw[->, black, thick, line width=0.1mm] (1.25,0) to (2.25,-1);
\draw[->, black, thick, line width=0.1mm] (1.25,-0.5) to (2.25,0);
\draw[->, black, thick, line width=0.1mm] (1.25,-0.5) to (2.25,-1);
\draw[->, black, thick, line width=0.1mm] (1.25,-1) to (2.25,0);
\draw[->, black, thick, line width=0.1mm] (1.25,-1) to (2.25,-0.5);
\draw[->, black, thick, line width=0.1mm] (2.25,-0.5) to (3.25,0);
\draw[->, black, thick, line width=0.1mm] (2.25,-0.5) to (3.25,-1);
\draw[->, black, thick, line width=0.1mm] (2.25,-0.5) to (3.25,0.5);
\draw[->, black, thick, line width=0.1mm] (3.25,-0) to (3.75,0.5);
\draw[->, black, thick, line width=0.1mm] (3.25,-0.5) to (4,0.5);
\draw[->, black, thick, line width=0.1mm] (3.25,-1) to (4.25,0.5);
\end{tikzpicture}
};

\node[inner sep=0pt,align=center, scale=0.55, color=black] (hacker) at (2,-2.07) {
a) Normal operation
};
\node[inner sep=0pt,align=center, scale=0.55, color=black] (hacker) at (2,-4.1) {
b) View change
};
\node[inner sep=0pt,align=center, scale=0.55, color=black] (hacker) at (0,-5.9) {
c) Checkpoint
};
\node[inner sep=0pt,align=center, scale=0.55, color=black] (hacker) at (3,-5.9) {
d) State transfer
};
\node[inner sep=0pt,align=center, scale=0.55, color=black] (hacker) at (2,-8.15) {
e) Join
};
\node[inner sep=0pt,align=center, scale=0.55, color=black] (hacker) at (2,-10.23) {
f) Evict
};
\end{tikzpicture}